\documentclass{statsoc}
\pdfoutput=1 
\usepackage[left=7mm,right=15mm]{geometry} 
\usepackage{vmargin} 

\RequirePackage[OT1]{fontenc}
\RequirePackage{graphicx}
\RequirePackage{enumerate}
\RequirePackage{booktabs}
\RequirePackage{color, colortbl, transparent}
\RequirePackage{multirow}
\RequirePackage{amssymb}
\RequirePackage[cmex10]{amsmath}
\RequirePackage{placeins}
\RequirePackage{xtab}
\RequirePackage[caption=false]{subfig}
\usepackage{multicol}
\allowdisplaybreaks
\newtheorem{thm}{Theorem}
\newtheorem{mydef}{Definition}
\newtheorem{lemma}{Lemma}
\newtheorem{slemma}{Supplementary Lemma}
\newcommand {\bs} { \boldsymbol}

\title[cjBitSeq]{A Bayesian model selection approach for identifying differentially expressed transcripts from RNA-Seq data}

\author[P.~Papastamoulis and M.~Rattray]{Panagiotis Papastamoulis\\
\texttt{panagiotis.papastamoulis@manchester.ac.uk}}

\address{University of Manchester, Faculty of Life Science,
Manchester,
UK}

\author[P.~Papastamoulis and M.~Rattray]{Magnus Rattray\\
\texttt{magnus.rattray@manchester.ac.uk}}

\address{University of Manchester, Faculty of Life Science,
Manchester,
UK}

\begin{document}

\begin{abstract}
Recent advances in molecular biology allow the quantification of the transcriptome and scoring transcripts as differentially or equally expressed between two biological conditions. Although these two tasks are closely linked, the available inference methods treat them separately: a  primary model is used to estimate expression and its output is post-processed using a differential expression model. In this paper, both issues are simultaneously addressed  by proposing the joint estimation of expression levels and differential expression: the unknown relative abundance of each transcript can either be equal or not between two conditions. A hierarchical Bayesian model builds upon the BitSeq framework and the posterior distribution of transcript expression and differential expression is inferred using Markov Chain Monte Carlo (MCMC). It is shown that the proposed model enjoys conjugacy for fixed dimension variables, thus the full conditional distributions are analytically derived. Two samplers are constructed, a reversible jump MCMC sampler and a collapsed Gibbs sampler, and the latter is found to perform best. A cluster representation of the aligned reads to the transcriptome is introduced, allowing parallel estimation of the marginal posterior distribution of subsets of transcripts under reasonable computing time. The proposed algorithm is benchmarked against alternative methods using synthetic datasets and applied to real RNA-sequencing data. Source code is available online \footnote{{\tt https://github.com/mqbssppe/cjBitSeq}}.
\end{abstract}

\section{Introduction}

Quantifying the transcriptome of a given organism or cell is a fundamental task in molecular biology. RNA-sequencing (RNA-Seq)  technology  produces transcriptomic data in the form of short reads \citep{mort}. These reads can be used either in order to reconstruct the transcriptome using de novo or guided assembly, or to estimate the abundance of known transcripts given a reference annotation. Here, we consider the latter scenario in which transcripts are defined by annotation. In such a case, millions of short reads are aligned to the reference transcriptome (or genome) using mapping tools such as Bowtie \citep{bowtie} (or Tophat \citep{tophat}). Of particular interest is the identification of differentially expressed transcripts (or isoforms) across different samples. Throughout this paper the term transcript refers to  isoforms, so differential transcript detection has the same meaning as differential isoform detection. Most genes in higher eukaryotes can be spliced into alternative transcripts that share specific parts of their nucleotide sequence. Thus, a short read is not uniquely aligned to the transcriptome and its origin remains uncertain, making transcript expression estimation non-trivial. Probabilistic models provide a powerful means to estimate transcript abundances as they are able to take this ambiguous read assignment into consideration in a principled manner.

There are numerous methods that estimate transcript expression from RNA-Seq data, including RSEM \citep{rsem}, IsoEM \citep{isoEM}, Cufflinks \citep{cuffdif, cuffdif2}, BitSeq (Stage 1)\citep{bitseq}, TIGAR \citep{newvb} and Casper \citep{casper}. Some of these methods also include a second stage for performing DE analysis at the transcript level (e.g. Cuffdiff and BitSeq Stage 2) and stand-alone methods for transcript-level DE calling have also been developed such as EBSeq \citep{EBSeq} and MetaDiff \citep{Jia2015}. Cuffdiff uses an asymptotically normal test statistic by applying the delta method to the log-ratio of transcript abundances between two samples, given the estimated expression levels using Cufflinks. EBSeq estimates the Bayes factor of a model under DE or nonDE for each transcript, building a Negative Binomial model upon the estimated read counts from any method. BitSeq Stage 2 ranks transcripts as differentially expressed by the probability of positive log-ratio (PPLR) based on the MCMC output from BitSeq Stage 1, which estimates the expression levels assuming a mixture model. Gene-level DE analysis is also available using count-based methods such as edgeR \citep{edger} and DESeq \citep{Anders:2010} but here we limit our attention to methods designed for transcript-level DE calling. 

All existing methods for transcript-level DE calling apply a two-step procedure. The mapped RNA-Seq data is used as input of a first stage analysis to estimate transcript expression. The output of this stage is then post-processed at a second stage in order to classify transcripts as DE or non-DE. The bridge between the two stages is based upon certain parametric assumptions for the distribution of the estimates of the first stage and/or the use of asymptotic results (as previously described above). Also, transcript-level expression estimates are correlated through sharing of reads and this correlation is typically ignored in the second stage. Such two-stage approaches are quite useful in practice since the differential expression question is not always the main aim of the analysis; therefore estimating expression is useful in itself. However, when the main purpose of an experiment is DE calling then the two-stage procedure increases the modelling complexity and may result in overfitting, since there is no guarantee that the underlying assumptions are valid. Note that a recent method \citep{badge} addresses the joint estimation of expression and differential expression modelling of exon counts under a Bayesian approach but at the gene level rather than the transcript level considered here.

The contribution of this paper is to develop a method for the joint estimation of expression and differential expression at the transcript level. The method builds upon the Bayesian framework of the BitSeq (Stage 1) model where transcript expression estimation reduces to estimating the posterior distribution of the weights of a mixture model using MCMC~\cite{bitseq}. The novelty in the present study is that differential expression is addressed by inferring which weights differ between two mixture models. This is achieved by using two samplers. A reversible jump MCMC (rjMCMC) algorithm \citep{Green:95} updates both transcript expression and differential expression parameters, while a collapsed Gibbs algorithm is developed which avoids transdimensional transitions. The high-dimensional setting of RNA-seq data studies makes the convergence to the joint posterior distribution computationally challenging. To alleviate this computational burden and allow easier parallelization, a new cluster representation of the transcriptome is introduced which collapses the problem to  subsets of transcripts sharing aligned reads.

The rest of the paper is organized as follows. The mixture model used in the original BitSeq setup is reviewed in Section \ref{sec:bitseq}. The prior assumptions of the new cjBitSeq (clusterwise joint BitSeq) model is introduced in Section \ref{sec:model}. The full conditional distributions are given in Section \ref{sec:gibbs} and two MCMC samplers are described in Section \ref{sec:rjMCMC}. A cluster representation of aligned reads and transcripts is discussed in Section \ref{sec:clusters} and details over False Discovery Rate (FDR) estimation are given in Section \ref{sec:fdr}. Large scale simulation studies are presented in Section \ref{sec:sim} and the proposed method is illustrated to a real human dataset in Section \ref{sec:real2}. The paper concludes in Section \ref{sec:dis} with a synopsis and discussion. 


\section{Methods}\label{sec:methods}

In the BitSeq model, the mixture components correspond to annotated transcript sequences and the mixture weights correspond to their relative expression levels. The data likelihood is then computed by considering the alignment of reads (or read-pairs) against each mixture component. Essentially, this model is modified here in order to construct a well-defined probability of DE or non-DE when two samples are available. 

We induce a set of free parameters of varying dimension, depending on the number of different weights between two mixture models. Assuming two 
independent Dirichlet prior distributions, the Gibbs sampler \citep{geman, gelfand} draws samples from the full conditionals, which are independent Dirichlet and Generalized Dirichlet \citep{connor, Wong:98, Wong:10} distributions. This representation allows the integration of the corresponding parameters as stated at Theorem \ref{thm:collapsed}. Therefore, we provide two MCMC samplers depending on whether transcript expression levels are integrated out or not. These samplers converge to the same target distribution but using different steps in order to update the state of each transcript: the first one uses a birth-death move type \citep{Richardson:97,papRJ} and the second one is a block update from the full conditional distribution. After detecting clusters of 
transcripts and reads, it is shown that the parallel application of the algorithm to each cluster converges to proper marginals of the full posterior distribution.

\subsection{BitSeq}\label{sec:bitseq}

Let $\bs x = (x_1,\ldots,x_r)$, $x_i\in\mathcal X$, $i = 1,\ldots,r$, denote a sample of $r$ short reads aligned to a given set of $K$ transcripts. The sample space $\mathcal X$ consists of all sequences of letters A, C, G, T. Assuming that reads are independent, the joint probability density function of the data is written as 
\begin{equation}\label{standardBitseq}
\bs x|\bs\theta\sim \prod_{i=1}^{r}\sum_{k=1}^{K}\theta_kf_k(x_i).
\end{equation}
The number of components ($K$) is equal to the number of transcripts and it is considered as known since the transcriptome is given. The parameter vector $\bs\theta = (\theta_1,\ldots,\theta_K)\in\mathcal P_{K-1}$ denotes relative  abundances, where 
$$\mathcal P_{K-1}:=\{p_k\geqslant 0, k=1,\ldots,K-1:\sum_{k=1}^{K-1}p_k\leqslant 1;p_K:=1-\sum_{k=1}^{K-1}p_k\}.$$
The component specific density $f_k(\cdot)$ corresponds to the probability of a read aligning at some position of transcript $k$, $k=1,\ldots,K$. Since we assume a known transcriptome, $\{f_k\}_{k=1}^{K}$ are known as well and they are computed according to the methodology described in \cite{bitseq} (see also Appendix A in supplementary material), taking into account position and sequence-specific bias correction methods. 

A priori it is assumed that $\bs\theta\sim\mathcal D_{K-1}(\alpha_1,\ldots,\alpha_K)$, with $\mathcal D_{j}$ denoting the Dirichlet distribution defined over $\mathcal P_j$. Furthermore, it is assumed that $\alpha_1=\ldots=\alpha_K=1$, which is equivalent to the uniform distribution in $\mathcal P_{K-1}$. In the original implementation of BitSeq \citep{bitseq},  MCMC samples are drawn from the posterior distribution of $\bs\theta|\bs x$ using the Gibbs sampler while more recently variational Bayes approximations have also been included for faster inference  \citep{papVB,bitseqVB}. 

Given the output of BitSeq stage 1 for two different samples, BitSeq stage 2 implements a one-sided test (PPLR) for DE analysis. However, this approach does not define transcripts as DE or non-DE and is therefore not directly comparable to standard 2-sided tests available in most other packages \citep{cuffdif2,EBSeq}. Also, correlations between transcripts in the posterior distribution for each sample are discarded during the DE stage, leading to potential loss of accuracy when making inferences. In order to deal with these limitations, a new method for performing DE analysis is presented next.

\subsection{cjBitSeq}\label{sec:model}

Assume that we have at hand two samples $\bs x := (x_1,\ldots,x_{r})$ and $\bs y:=(y_1,\ldots,y_{s})$ denoting the number of (mapped) reads for sample $\bs x$ and $\bs y$, respectively. Now, let $\theta_k$ and $w_k$ denote the unknown relative abundance of transcript $k=1,\ldots,K$ in sample $\bs x$ and $\bs y$, respectively. Define the parameter vector of relative abundances as
$\bs\theta = (\theta_1,\ldots,\theta_{K-1};\theta_K)\in\mathcal P_{K-1}$
and $\bs w = (w_1,\ldots,w_{K-1};w_K)\in\mathcal P_{K-1}$. Under the standard BitSeq model the prior on the parameters $\bs \theta$ and $\bs w$ would be a product of independent Dirichlet distributions. In this case the probability $\theta_k= w_k$ under the prior is zero and it is not straightforward to define non-DE transcripts. To model differential expression we would instead like to identify instances where transcript expression has not changed between samples. Therefore, we introduce a non-zero probability for the event $\theta_k=w_k$. This leads us to define a new model with a non-independent prior for the parameters $\bs\theta$ and $\bs w$.

\begin{mydef}[State vector]\label{def:state} Let $c:=(c_1,\ldots,c_K)\in\mathcal C$, where $\mathcal C$ is the set defined by:
\begin{enumerate}
\item $c_k\in\{0,1\}$, $k=1,\ldots,K$
\item $c_+:=\sum_{k=1}^{K}c_k\neq 1$.
\end{enumerate}
Then, for $k=1,\ldots,K$ let:
$
\begin{cases}
\theta_k = w_k, & \text{if }c_k=0\\
\theta_k \neq w_k, & \text{if }c_k=1.
\end{cases}
$

\noindent
We will refer to vector $c$ as the state vector of the model.
\end{mydef}
For example, assume that $K = 6$ and $c=(1,0,0,1,0,1)$. According to Definition \ref{def:state}, $\theta_k=w_k$ for $k=2,3,5$ and $\theta_k\neq w_k$ for $k  = 1,4,6$. From Definition \ref{def:state} it is obvious that the sum of the elements in $c$ cannot be equal to 1 because  either all $\theta$'s have to be equal to $w$'s, or at least two of them have to be different. The introduction of such dependencies between the elements of $\bs\theta$ and $\bs w$ has non-trivial effects on the prior assumptions of course. It is clear that with this approach we should define a valid conditional prior distribution for $\bs\theta,\bs w|c$. 

At first we impose a prior assumption on $c$. We will consider the Jeffreys' \citep{jeffreys} prior distribution for a Bernoulli trial, that is $P(c_k=1|\pi) = \pi$ with $\pi$ following a Beta distribution. 
Since $c_+\neq 1$, the prior distribution of the state vector $c$ is expressed as 
\begin{eqnarray}\label{eq:beta_prior}
\pi &\sim& \mbox{Beta}(1/2,1/2)\\
P(c|\pi)&=& P(c|c_+\neq 1,\pi)= \frac{\pi^{c_+}(1-\pi)^{K-c_+}}{1-K\pi(1-\pi)^{K-1}},\quad c\in\mathcal C. \label{eq:jeffreys_prior}
\end{eqnarray}

Next we proceed to the definition of a proper prior structure for the weights of the mixture. At this step extra care should be taken for everything to make sense as a probabilistic space. It is obvious that $(\bs\theta,\bs w)$ should be defined conditional to the state vector $c$. What it is less obvious, is that $(\bs\theta,\bs w)$ should be defined conditional on a parameter of varying dimension. At this point, we introduce some extra notation. 
\begin{mydef}[Dead and alive subsets and permutation of the labels]\label{def:dead_alive} For a given state vector $c$, define the  order-specific subsets
$$C_0(c):=\{\tau_1<\ldots<\tau_{K-c_+}\in\{1,\ldots,K\}: c_{\tau_k}=0\quad \forall k = 1,\ldots,K-c_+\}$$ and
$$C_1(c):=\{\tau_{K-c_+ + 1}<\ldots<\tau_{K}\in\{1,\ldots,K\}: c_{\tau_k}=1\quad \forall k = K-c_+ + 1,\ldots,K\}.$$ 
These sets will be called dead and alive subsets of the transcriptome index, respectively. Moreover, $\tau = (\tau_1,\ldots,\tau_K)$ denotes the unique permutation of $\{1,\ldots,K\}$ obeying the ordering within the dead and alive subsets.\end{mydef} 
As it will be made clear later, it is convenient to define a unique labelling within the dead and alive subsets so we also explicitly defined the corresponding permutation ($\tau$) of the labels. In order to clarify Definition \ref{def:dead_alive}, assume that $c=(1,0,0,1,0,1)$. Then Definition \ref{def:dead_alive} implies that $C_0(c) = \{2,3,5\}$, $C_1(c)=\{1,4,6\}$ and $\tau = (2,3,5,1,4,6)$. The order-specific definition of these subsets excludes  $\{3,2,5\}$ (for example) from the definition of a dead subset. 

It is clear that if $C_0(c)=\emptyset$, then both $ \bs\theta$ and $\bs w$ have $K-1$ free parameters each. However, if $C_0(c)\neq\emptyset$, the free parameters are lying in a lower dimensional space. This means that  $(\bs\theta,\bs w)$ should be defined given $c$ by taking into account the set of free parameters that are actually allowed by the state vector. In particular, $(\bs\theta,\bs w)$ are pseudo-parameters. The actual parameters of our problem are defined in Lemma \ref{lem:free}.  

In what follows, the notation $\tau \bs \sigma$ should be interpreted as the reordering of vector $\bs \sigma = (\sigma_1,\ldots,\sigma_{K})$ under permutation $\tau$. E.g: assume that $\tau = (3,1,2)$ and $\bs \sigma = (\sigma_1,\sigma_2,\sigma_3)$, then: $\tau\bs\sigma = (\sigma_3,\sigma_1,\sigma_2)$. Let also $\tau^{-1}$ denote the inverse permutation of $\tau$. 

\begin{lemma}[Existence and uniqueness of free parameters]\label{lem:free} For every $(c,\tau,\bs\theta,\bs w)$ respecting Definitions \ref{def:state} and \ref{def:dead_alive} there exists a unique set of free parameters:
\begin{equation}
(\bs u,\bs v)\in\mathcal P_{K-1}\times\mathcal P_{c_+ -1},
\end{equation}
such that:
\begin{eqnarray}\label{eq:theta}
\bs\theta &=& \tau^{-1}\bs u\\
\label{eq:w}
\bs w &=& \tau^{-1}\bs\varpi,\end{eqnarray}
where $\bs\varpi = \left(\{u_{\tau^{-1}_k}:k\in C_0(c)\},\bs v\sum_{k\in C_1(c)} u_{\tau^{-1}_k}\right)$
under the conventions $\mathcal P_{-1}:=\emptyset$ and $ \emptyset\sum_{k\in\emptyset}u_k:=\emptyset$.
\end{lemma}
\begin{proof}
It is trivial to show that  $(c,\tau,\bs u,\bs v)\rightarrow (\bs\theta,\bs w)$ is an ``one to one'' and ``onto''  mapping (bijective function).
\end{proof}
\noindent
\textbf{Example:} Assume that $c = (1,0,0,1,0,1)$, where $C_0(c)=\{2,3,5\}$ and $C_1(c)=\{1,4,6\}$. Then, $\tau = (2,3,5,1,4,6)$ and $\tau^{-1}=(4, 1, 2, 5, 3, 6)$. According to state $c$ we should have that $\theta_2 = w_2$, $\theta_3 = w_3$ and $\theta_5 = w_5$, while $\theta_k\neq w_k$ for $k\in C_1(c)$. Lemma \ref{lem:free} states that $\bs\theta$ and $\bs w$ can be expressed as a transformation of two independent parameters: $\bs u = (u_1,u_2,u_3,u_4,u_5,u_6)\in\mathcal P_5$ and $\bs v = (v_1,v_2,v_3)\in\mathcal P_2,$. According to Equation \eqref{eq:theta}, $\bs\theta$ is a permutation of the vector $\bs u$:
\begin{equation*}
\bs\theta|(c,\bs u) = (u_4,u_1,u_2,u_5,u_3,u_6).
\end{equation*}
Next, $\bs w$ is obtained by a permutation of $\bs\varpi$, which is a linear transformation of $\bs u$ and $\bs v$, that is,
$\bs\varpi = (u_1,u_2,u_3,v_1(u_4+u_5+u6),v_2(u_4+u_5+u6),v_3(u_4+u_5+u6))$. According to Equation \eqref{eq:w}:
\begin{equation*}
\bs w|(c,\bs u,\bs v) = \left( v_1(u_4 + u_5 + u_6), u_1,u_2,v_2(u_4 + u_5 + u_6),u_3,v_3(u_4 + u_5 + u_6)  \right).
\end{equation*}
Comparing the last two expressions for $\bs\theta$ and $\bs w$, it is obvious that  $\theta_2 = w_2$, $\theta_3 = w_3$ and $\theta_5 = w_5$, while $\theta_k\neq w_k$ for all remaining entries,  which is the configuration implied by the state vector $c$.
Note also that   $\{u_{\tau^{-1}_k};k\in C_0(c)\} = (u_1,\ldots,u_{K-c_+})$ and $\{u_{\tau^{-1}_k};k\in C_1(c)\} = (u_{K-c_+ + 1},\ldots,u_K)$ and $\sum_{k\in C_1(c)}w_k = \sum_{k \in C_1(c)}\theta_k=\sum_{k \in C_1(c)}u_{\tau^{-1}_k}$.

Now, it should be clear that given a state vector $c$, as well as the independent free parameters $\bs u$ and $\bs v$, the pseudo-parameters $\bs\theta$ and $\bs w$ are deterministically defined. In other words, the  conditional distributions of $\bs\theta$ and $\bs w$ are Dirac, gathering all their probability mass into the single points defined by Equations \eqref{eq:theta} and \eqref{eq:w}. Hence, the conditional prior distribution for transcript expression is written as:
\begin{equation}\label{eq:dirac}
f(\bs\theta,\bs w|c,\tau,\bs u,\bs v) = 1_{\bs\theta,\bs w}(\{\bs\theta(c,\tau,\bs u),\bs w(c,\tau,\bs u,\bs v)\}),
\end{equation}
with $\bs\theta(c,\tau,\bs u)$ and $\bs w(c,\tau,\bs u,\bs v)$ as in Equations \eqref{eq:theta} and \eqref{eq:w}, respectively.

Moreover, we stress that if the permutation $\tau$ was not uniquely defined according to Definition \ref{def:dead_alive}, then we would have had to take into account all the possible permutations within the dead and alive subsets. However, such an approach would lead to an increased modelling complexity without making any difference at the inference. That said, the conditional prior distribution of $\tau$ given $c$ is Dirac:
\begin{equation}\label{eq:tau}
f(\tau|c) = 1_{\tau}(\tau(c)),
\end{equation}
where $\tau(c)$ denotes the unique permutation (given $c$) in Definition \ref{def:dead_alive}.
 
At this point we state our prior assumptions for the free parameters, given a state vector $c$. We assume that a priori $\bs u$ and $\bs v$ are independent random variables distributed according to Dirichlet distribution, that is:
\begin{eqnarray}\label{eq:uprior}
\bs u|c &\sim &\mathcal D_{K-1}(\alpha_1,\ldots,\alpha_K)\\\label{eq:vprior}
\bs v|c &\sim &\mathcal D_{c_+-1}(\gamma_1,\ldots,\gamma_{c_+}).
\end{eqnarray}
In the applications, we will furthermore assume that  $\alpha_k = 1$ for all $k = 1,\ldots,K$ and $\gamma_\ell = 1$ for all $\ell = 1,\ldots,c_+$, in order to assign a uniform prior distributions over $\mathcal P_{K-1}\times\mathcal P_{c_+-1}$. Now, the following Theorem holds.

\begin{figure}[t]
  \centering
    \includegraphics[width = 0.9\textwidth]{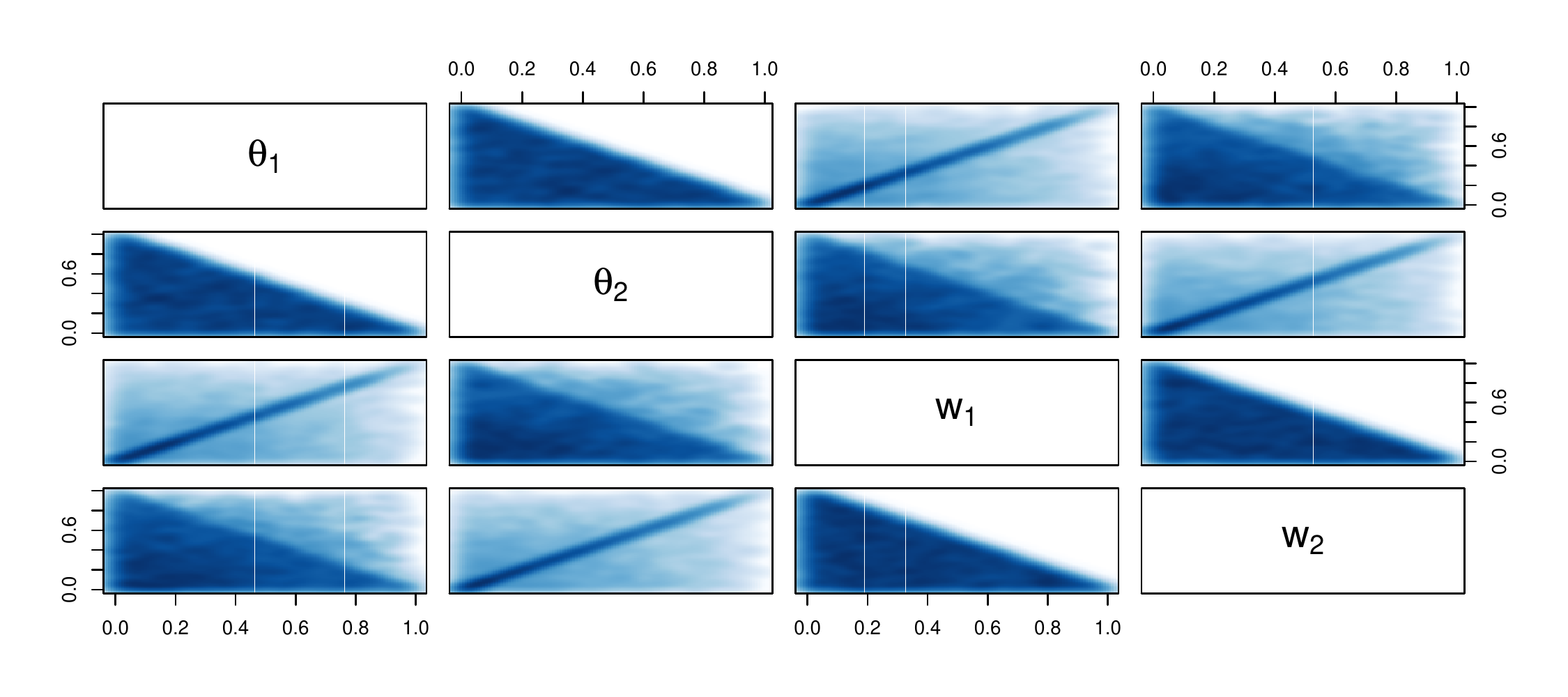}
      \caption{Simulation from the prior distribution \eqref{eq:dirac} of $(\bs\theta,\bs w)$ for $K = 3$, $\alpha_k =\gamma_k = 1$ for $k = 1,2,3$, and also assuming the Jeffreys' prior for $c$. Theorem \ref{lem:prior} states that marginally: $\bs\theta\sim \mathcal D(1,1,1)$ and $\bs w\sim \mathcal D(1,1,1)$.}\label{fig:prior}
\end{figure}

\begin{thm}\label{lem:prior}
Assume that \eqref{eq:uprior} and \eqref{eq:vprior} hold true and furthermore: $\alpha_k = \gamma_k = \alpha$ for all $k = 1,\ldots,K$. Then, $\bs\theta$ and $\bs w$ are marginally identical random variables following the $\mathcal D_{K-1}(\alpha,\ldots,\alpha)$ distribution.
\end{thm}
\begin{proof}
See Appendix C in supplementary material.
\end{proof}
\noindent
Note here that Theorem \ref{lem:prior} does not imply that $\bs\theta$ and $\bs w$ are a priori independent. As shown in Figure \ref{fig:prior}, $\theta_k$ is exactly equal to $w_k$ with probability $P(c_k = 0)>0$, $k = 1,\ldots,K$.

The model definition is completed by considering the latent allocation variables of the mixture model. Let $\bs\xi=\{\xi_1,\ldots,\xi_r\}$ and $\bs z=\{z_1,\ldots,z_s\}$ with 
\begin{eqnarray*}
P(\xi_i = k|\bs\theta) &=& \theta_k, \quad \mbox{ independent for}\quad i=1,\ldots,r\\
P(z_j = k|\bs w) &=& w_k, \quad \mbox{independent for }\quad j=1,\ldots,s,
\end{eqnarray*}
for $k=1,\ldots,K$. Moreover, $\bs \xi,\bs z$ are assumed conditionally independent given $\bs\theta$ and $\bs w$, that is, $P(\bs \xi,\bs z|\bs\theta,\bs w) = P(\bs \xi|\bs\theta)P(\bs z|\bs w)$. Now, the joint distribution of the complete data ($\bs x,\bs y,\bs\xi,\bs z$) factorizes as follows:
\begin{equation}\label{eq:complete}
f(\bs x,\bs y,\bs\xi,\bs z|\bs\theta,\bs w)= \prod_{i=1}^{r}\theta_{\xi_i}f_{\xi_i}(x_i)\prod_{j=1}^{s}w_{z_j}f_{z_j}(y_j).\end{equation}
Let $\bs g = (\bs x,\bs y,\bs\xi,\bs z,\bs\theta,\bs w,\bs u,\bs v,c,\tau,\pi)$. From Equations \eqref{eq:beta_prior}, \eqref{eq:jeffreys_prior} and \eqref{eq:dirac}-\eqref{eq:complete}, the joint distribution of $\bs g$ is defined as 
\begin{align}\nonumber
f(\bs g|\bs\alpha,\bs\gamma,K) &=& f(\bs x,\bs y,\bs\xi,\bs z|\bs\theta,\bs w)f(\bs u|\bs\alpha,K)f(\bs v|c,\bs\gamma)f(\bs\theta|\tau,\bs u)\\
\label{eq:model}
&\times&f(\bs w|c,\tau,\bs u,\bs v)f(\tau|c)f(c|K,\pi)f(\pi).
\end{align}
Equation \eqref{eq:model} defines a hierarchical model whose graphical representation is given in Figure \ref{fig:dag} with circles (squares) denoting unobserved (observed/known) variables.

\begin{figure}[t]
  \centering
    \includegraphics[width = 0.4\textwidth]{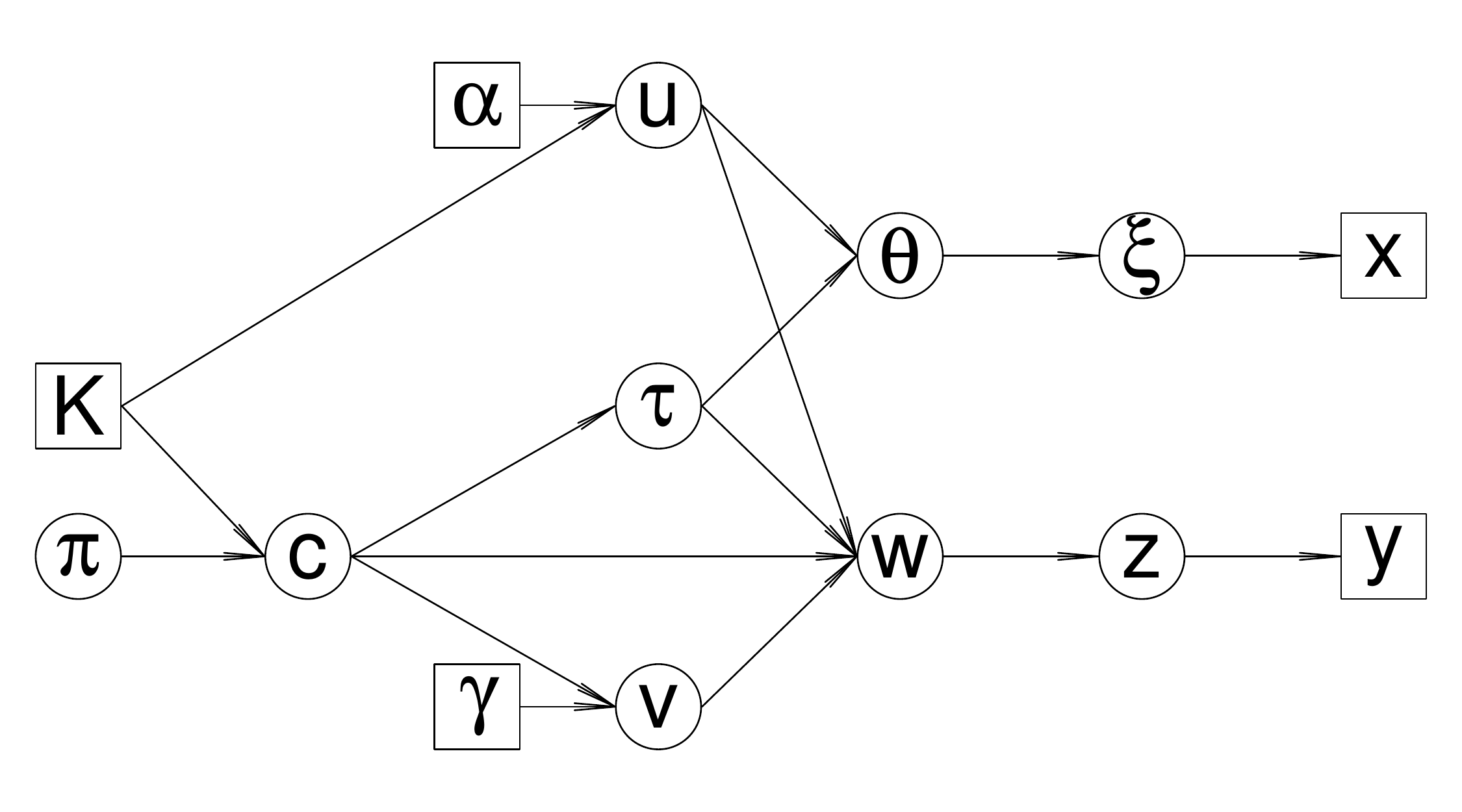}
      \caption{Directed Acyclic Graph representation of the hierarchical model \eqref{eq:model}.}\label{fig:dag}
\end{figure}

\subsection{Full conditional distributions for the Gibbs updates}\label{sec:gibbs}

In this section, the full conditional distributions are derived. Let $h|\cdots$ denote the conditional distribution of a random variable $h$ given the values of the rest of the variables. We also denote by $\bs x_{[-i]}$ all remaining members of a generic vector after excluding its $i$-th item. 

It is straightforward to show that $\pi|\cdots \sim \mbox{Beta}(c_+ +1/2,K-c_+ + 1/2)$. For the allocation variables it follows that:
\begin{eqnarray}\label{eq:xi_cond}
P(\xi_i=k|\cdots)&\propto& \theta_k f_k(x_i) \quad k = 1,\ldots,K\\
P(z_j=k|\cdots)&\propto& w_k f_k(y_i) \quad k = 1,\ldots,K \label{eq:z_cond}
\end{eqnarray}
independent for $i = 1,\ldots,r$ and $j = 1,\ldots,s$. Now, given $(\bs u,\bs v,c,\tau)$, it is again trivial to see that the full conditional distributions of $\bs \theta,\bs w|\cdots$ is the same as in \eqref{eq:dirac}. Let $\mathcal GD(\cdot,\cdot)$ denotes the Generalized Dirichlet distribution (see Appendix B in supplementary material) and also define
\begin{eqnarray*}
s_k(\bs \xi)&:=&\sum_{i=1}^{r}I(\xi_i=k),\quad
s_k(\bs z):=\sum_{j=1}^{s}I(z_j=k)
\end{eqnarray*}
for $k=1,\ldots,K$. Regarding the full conditional distribution of the free parameters, we have the following result. 

\begin{lemma}\label{lem:gibbs}
The full conditional distribution of $(\bs u,\bs v|\cdots)$ is 
\begin{eqnarray}\label{eq:u_gibbs}
\bs u|\cdots&\sim&\mathcal{GD}(\lambda_1,\ldots,\lambda_{K-1};\beta_1,\ldots,\beta_{K-1})\\\label{eq:v_gibbs}
\bs v|\cdots&\sim&\mathcal{D}_{c_{+}-1}(\{\gamma_{\ell}+s_{\tau_{\ell+k^*}}(\bs z);\ell=1,\ldots,c_+\}),
\end{eqnarray}
with $k^* := K-c_+$, conditionally independent (given all other variables), where
\begin{equation*}
\lambda_{k}:=\begin{cases}
\alpha_k + s_{\tau_k}(\bs \xi) + s_{\tau_k}(\bs z), & k = 1,\ldots,k^*\\
\alpha_k + s_{\tau_k}(\bs \xi), & k = k^*+1,\ldots,K-1
\end{cases}
\end{equation*}
and
\begin{equation*}
\beta_{k}:=\begin{cases}
\sum_{j=k+1}^{K}(\alpha_j + s_{\tau_j}(\bs \xi) + s_{\tau_j}(\bs z)), & k = 1,\ldots,k^*\\
\sum_{j=k+1}^{K}(\alpha_j + s_{\tau_k}(\bs \xi)), & k = k^*+1,\ldots,K-1.
\end{cases}
\end{equation*}
\end{lemma}
\begin{proof} 
See Appendix D in the supplementary material.
\end{proof}

Here, we underline that we essentially derived an alternative construction of the Generalized Dirichlet distribution. Assuming that two vectors of weights share some common elements, and independent Dirichlet prior distributions are assigned to the free parameters of these weights, the posterior distribution of the first free parameter vector is a Generalized Dirichlet. Finally, notice that if $\bs v = \emptyset$ (this is the case when the corresponding elements of the weights of the two mixtures are all equal to each other), the Generalized distribution \eqref{eq:u_gibbs} reduces to the distribution $\mathcal D_{K-1}(\{\alpha_{k} + s_{k}(\bs \xi) + s_{k}(\bs z);k=1,\ldots,K\})$,
as expected, since in such a case $(\bs x,\bs y)$ forms a random sample of size $r+s$ from the same population. On the other hand, if all weights are different, the full conditional distribution of $\bs u,\bs v$ becomes a product of two independent Dirichlet distributions, as expected. Next we show that we can integrate out the parameters related to transcript expression and directly sample from the marginal posterior distribution of $\bs\xi,\bs z,c|\bs x,\bs y$.

\begin{thm}\label{thm:collapsed} Integrating out the transcript expression parameters $\bs u,\bs v$, the full conditional distributions of allocation variables are written as:
\begin{align}\label{eq:marginal}
f(\bs\xi,\bs z|\bs x,\bs y,c) &\propto &
\frac{\Gamma\left(\sum\limits_{k\in C_1}\widetilde\alpha_k + s_k(\bs\xi) + s_k(\bs z)\right)}{\Gamma\left(\sum\limits_{k\in C_1}\widetilde\alpha_k + s_k(\bs\xi)\right)\Gamma\left(\sum\limits_{k\in C_1}\gamma_{\ell(k)}+s_{k}(\bs z)\right)}\\
\nonumber
&\times&\prod_{k \in C_1}\Gamma(\widetilde\alpha_k + s_k(\bs \xi))\Gamma(\gamma_{\ell(k)} + s_{k}(\bs z))\\
\nonumber
&\times&\prod_{k \in C_0}\Gamma(\widetilde\alpha_k + s_k(\bs\xi) + s_k(\bs z))\prod\limits_{i=1}^{r}f_{\xi_i}(x_i)\prod\limits_{j=1}^{s}f_{z_j}(y_j)\\
\label{eq:xi_collapsed}
P(\xi_i = k|\bs\xi_{[-i]},\bs z,c,\bs x) &\propto&
\begin{cases}
(\widetilde\alpha_k + s_{k}^{(i)}(\bs\xi) + s_{k}(\bs z))f_k(x_i), & k\in C_0\\
\frac{\sum\limits_{t\in C_1}\widetilde\alpha_k + s_{t}^{(i)}(\bs\xi) + s_{t}(\bs z)}{\sum\limits_{t\in C_1}\widetilde\alpha_t + s_{t}^{(i)}(\bs\xi)}(\widetilde\alpha_k + s_{k}^{(i)}(\bs\xi))f_k(x_i), & k\in C_1
\end{cases}
\\
\label{eq:z_collapsed}
P(z_j = k|\bs z_{[-j]},\bs \xi,c,\bs y) &\propto&
\begin{cases}
(\widetilde\alpha_k + s_{k}(\bs\xi) + s^{(j)}_{k}(\bs z))f_k(y_j), & k\in C_0\\
\frac{\sum\limits_{t\in C_1}\widetilde\alpha_t + s_{t}(\bs\xi) + s^{(j)}_{t}(\bs z)}{\sum\limits_{t\in C_1}\gamma_{\ell(t)} + s_{t}^{(j)}(\bs z)}(\gamma_{\ell(k)} + s_{k}^{(j)}(\bs z))f_k(y_j), & k\in C_1
\end{cases}
\end{align}
where $\widetilde\alpha_k = \alpha_{\tau^{-1}_k}$, $\ell(k) = \tau^{-1}_k - k^*$, $s_k^{(i)}(\bs\xi) = \sum_{t\neq i}I(\xi_i = k)$, $s_k^{(j)}(\bs z) = \sum_{t\neq j}I(z_i = k)$ for $k = 1,\ldots,K$, $i = 1,\ldots,r$, $j = 1,\ldots,s$.
\end{thm}
\begin{proof}
See Appendix E in supplementary material.
\end{proof}

Once again, note the intuitive interpretation of our model in the special cases where $C_0 = \emptyset$ or $C_1 = \emptyset$. If $C_0 = \emptyset$ (all transcripts are DE) then the nominator at the first line of Equation \eqref{eq:marginal} becomes equal to $\Gamma(\sum_{k}\alpha_k + r + s)$, that is, independent of $\bs\xi,\bs z$. Hence, \eqref{eq:marginal} reduces to the conditional distribution of the allocation variables when independent Dirichlet prior distributions are imposed to the mixture weights. On the contrary, when $C_1 = \emptyset$ (all transcripts are EE), the distribution reduces to the product appearing in the last row of Equation \eqref{eq:marginal}. This is the marginal distribution of the allocations when considering that $(\bs x,\bs y)$ arise from the same population and after imposing a Dirichlet prior on the weights, as expected.

\subsection{MCMC samplers}\label{sec:rjMCMC}

In this section we consider the problem of sampling from the posterior distribution of the model in \eqref{eq:model}. We propose two (alternative) MCMC sampling schemes, depending on whether the transdimensional random variable $\bs v$ is updated before or after $c$. 

Note that given $c$ everything has fixed dimension. However, as $c$ varies on the set of its possible values, then $\bs v\in \cup_{k\in\{0,2,...\ldots,K\}}\mathcal P_{k-1}$. This means that whenever $c$ is updated, $\bs v$ should change dimension. In order to construct a sampler that switches between different dimensions, a Reversible Jump MCMC method \citep{Green:95} can be implemented (see also \cite{Richardson:97} and \cite{papRJ}). However, this step can be avoided since we have already shown that the transcript expression parameters can be integrated out. Thus, a collapsed sampler is also available. Given an initial state, the general work flow for the proposed samplers is the following (we avoid to explicitly state that all  distributions appearing next are conditionally defined on the observed data $\bs x$, $\bs y$, although they should be understood as such).

\begin{minipage}{0.99\columnwidth}
\begin{multicols}{2}
\underline{\textbf{rjMCMC Sampler}}
\begin{enumerate}
\item Update $(\bs \xi,\bs z)|\bs\theta,\bs w$.
\item Update $(\bs u,\bs v)|c,\bs \xi,\bs z$.
\item Update $(\bs\theta,\bs w)|c,\tau,\bs u,\bs v$.
\item Propose update of $(c,\tau,\bs v)|\ldots$.
\item Update $\pi|c$.
\end{enumerate} 
\columnbreak
\underline{\textbf{Collapsed Sampler}}
\begin{enumerate}
\item Update $\xi_i|\bs\xi_{[-i]},\bs z,c$, $i = 1,\ldots,r$.
\item Update $z_j|\bs\xi,\bs z_{[-j]},c$, $j = 1,\ldots,s$.
\item Update a block of $c|\bs\xi,\bs z$.
\item Update $\pi|c$.
\item Update $(\bs\theta,\bs w,\tau,\bs u,\bs v)|c,\bs \xi,\bs z$ (optional).
\end{enumerate} 
\hspace{1ex}
\end{multicols}
\end{minipage}

Note that step (e) is optional for the collapsed sampler. It is implemented only to derive the estimates of transcript expression but it is not necessary for the previous steps. The next paragraphs outline the workflow for step (d) of rjMCMC sampler and step (c) of the collapsed sampler. For full details the reader is referred to Appendices F and G in the Supplementary material.

\paragraph{Reversible Jump sampler} Models of different dimensions are bridged using two move types, namely: ``birth'' and ``death'' of an index. The effect of a birth (death) move is to increase (decrease) the number of differentially expressed transcripts. These  moves are complementary in the sense that the one is the reverse of the other. Note that this step proposes a candidate state which is accepted according to the acceptance probability. 

\paragraph{Collapsed sampler} In this case we randomly choose two transcripts ($j_1$ and $j_2$) and perform an update from the conditional distribution $c_{j_1,j_2}|c_{-[j_1,j_2]}\bs \xi,\bs z,\bs x,\bs y,\pi$, which is detailed in Equations (G.1)--(G.4) in Section G of supplementary material. The random selection of the block $\{j_1,j_2\}\subseteq \{1,\ldots,K\}$ and the corresponding update of $c_{j_1,j_2}$ from its full conditional distribution is a valid MCMC step because it corresponds to a Metropolis-Hastings step in which the acceptance probability equals 1 (see Lemma 2 in Appendix G of the supplementary material). 

\subsection{Clustering of reads and transcripts}\label{sec:clusters}

In real RNA-seq datasets the number of transcripts could be very large. This imposes a great obstacle for the practical implementation of the proposed approach: the search space of the MCMC sampler consists of $2^{K}$ elements (state vectors) and convergence of the sampler may be very slow. This problem can be alleviated by a cluster representation of aligned reads to the transcriptome. High quality mapped reads exhibit a sparse behaviour in terms of their mapping places: each read aligns to a small number of transcripts and there are groups of reads mapping to specific groups of transcripts. Hence, we can take advantage of this sparse representation of alignments and break the initial problem into simpler ones, by performing MCMC per cluster.  

This clustering representation introduces an efficient way to perform parallel MCMC sampling by using multiple threads for transcript expression estimation. For this purpose we used the GNU parallel \citep{parallel} tool, which effectively handles the problem of splitting a series of jobs (MCMC per cluster) into the available threads. The jobs are ordered according to the number of reads per cluster and the ones containing more reads are queued first. GNU parallel efficiently spawns a new process when one finishes and keeps all available CPUs active, thus saving time compared to an arbitrary assignment of the same amount of jobs to the same number of available threads. For further details see Appendix H.

\subsection{False Discovery Rate}
\label{sec:fdr}

Controlling the False Discovery Rate (FDR) \citep{benjamini1995, storey2003} is a crucial issue in multiple comparisons problems. Under a Bayesian perspective, any probabilistic model that defines a positive prior probability for DE and EE yields that $\mathbb E(\mbox{FDR}|\mbox{data}) = \sum(1-\hat{P}(c_k = 1|\bs x,\bs y))d_k/D$ (see for example \citet{fdrJASA,fdrISBA}), where $d_k\in\{0,1\}$ and $D=\sum d_k$ denote the decision for transcript $k$, $k = 1,\ldots,K$ and the total number of rejections, respectively. Consequently, FDR can be controlled at a desired level $\alpha$ by choosing the transcripts that $\hat{P}(c_k=1|\bs x,\bs y)>1-\alpha$, which is also the approach proposed by \citet{EBSeq}. We have found that this rule achieves small false discovery rates compared to the desired level $\alpha$, but sometimes results to small true positive rate. 

A less conservative choice is the following. Let $q_1 \geqslant \ldots \geqslant q_K$ denote the ordered values of $\hat{P}(c_k=1|\bs x,\bs y)$, $k = 1,\ldots,K$ and define $G_k:=\frac{\sum_{j=1}^{k}(1-q_k)}{k}$, $k = 1,\ldots,K$. 
For any given $0 < \alpha <1$, consider the decision rule:
\begin{equation}\label{eq:decision}
d_k = \begin{cases}
1, & 1\leqslant k \leqslant g\\
0, & g + 1\leqslant k \leqslant K
\end{cases}
\end{equation}
where $g:=\max\{k=1\,\ldots,K:G_k\leqslant\alpha\}$.  It is quite straightforward to see that \eqref{eq:decision} controls the Expected False Discovery Rate at the desired level $\alpha$, since by direct substitution we have that
\begin{eqnarray*}
\mathbb E(\mbox{FDR}|\mbox{data}) = \frac{\sum_{k=1}^{K}(1-\hat{P}(c_k=1|\bs x,\bs y))d_k}{D}= \frac{\sum_{k=1}^{g}(1-q_k)}{g}\leqslant \alpha.
\end{eqnarray*}

An alternative is to use a rule optimizing the posterior expected loss of a predefined loss function. For example, the threshold $c/(c+1)$ is the optimal cutoff under the loss function $L = c\bar{\mbox{FD}} + \bar{\mbox{FN}}$, where $\bar{\mbox{FD}}$ and $\bar{\mbox{FN}}$ denote the posterior expected counts of false discoveries and false negatives, respectively. Note that $L$ is an extension of the $(0,1,c)$ loss functions for traditional hypothesis testing \citep{lindley}, while a variety of alternative loss functions can be devised as discussed in \citet{fdrJASA}.

\begin{figure}[t]
\centering
\begin{tabular}{c}
\includegraphics[width = 0.99\textwidth]{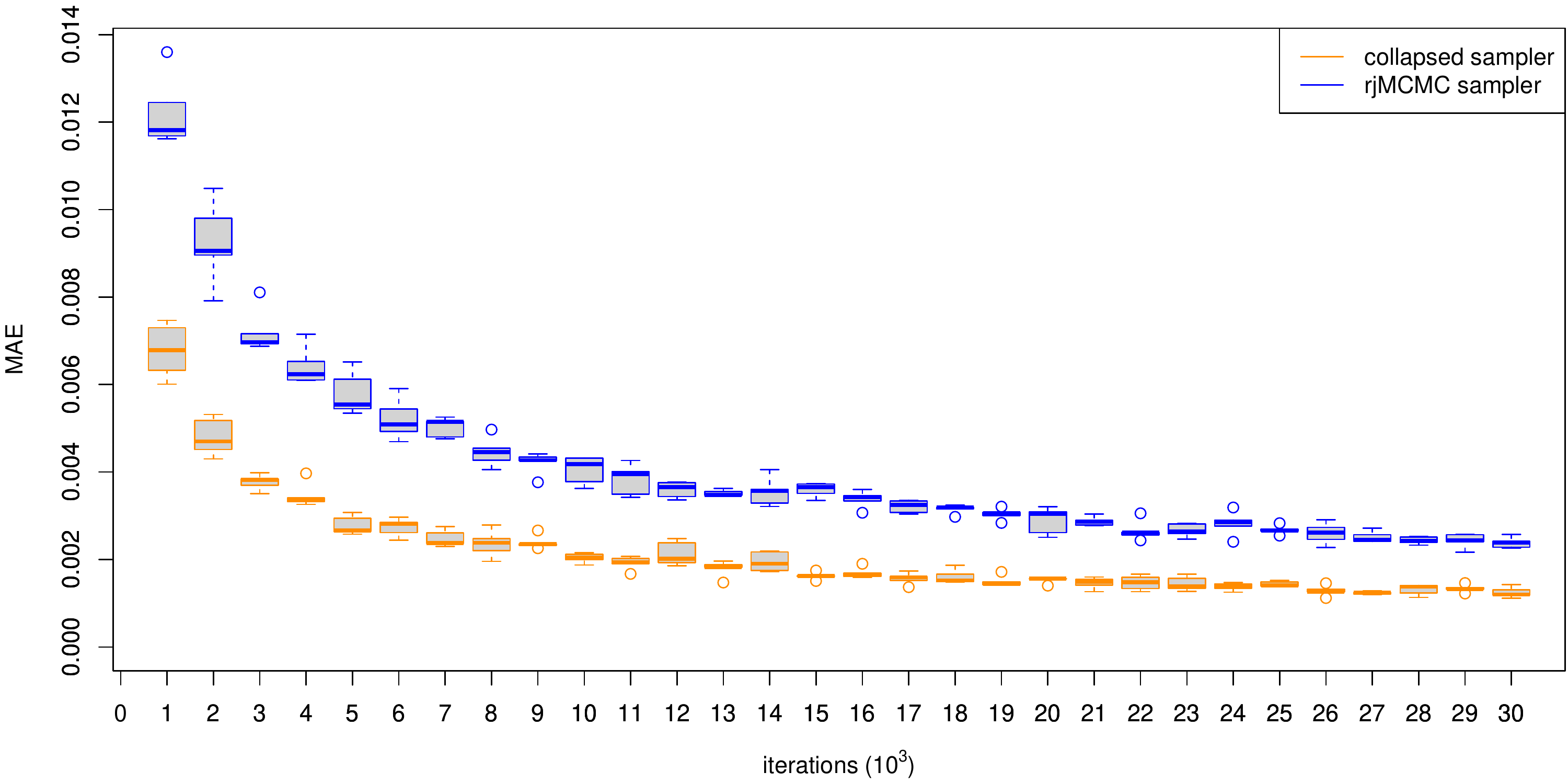} 
\end{tabular}
      \caption{Convergence of the ergodic means of posterior probabilities of DE for a toy example of $K=630$ transcripts. The ``ground truth'' for the posterior mean estimates ($\widehat P_g(c_k = 1)$; $k = 1,\ldots,K$) of these probabilities was inferred by running each sampler for $500000$ iterations. Then, each sampler ran for a smaller number of $m$ iterations resulting to the posterior mean estimates $\widehat P_m(c_k = 1)$; $k = 1,\ldots,K$, for $m = 1000,2000,\ldots,30000$. Finally, the averaged Mean Absolute Error of the posterior mean estimates was computed as: $\frac{1}{K}\sum_{k=1}^{K}|\widehat P_m(c_k=1)-\widehat P_g(c_k = 1)|$. The boxplots correspond to five replications of the previous procedure.}\label{fig:rjVScj}
\end{figure}

\begin{figure}[p]
\centering
\includegraphics[width = 0.99\textwidth]{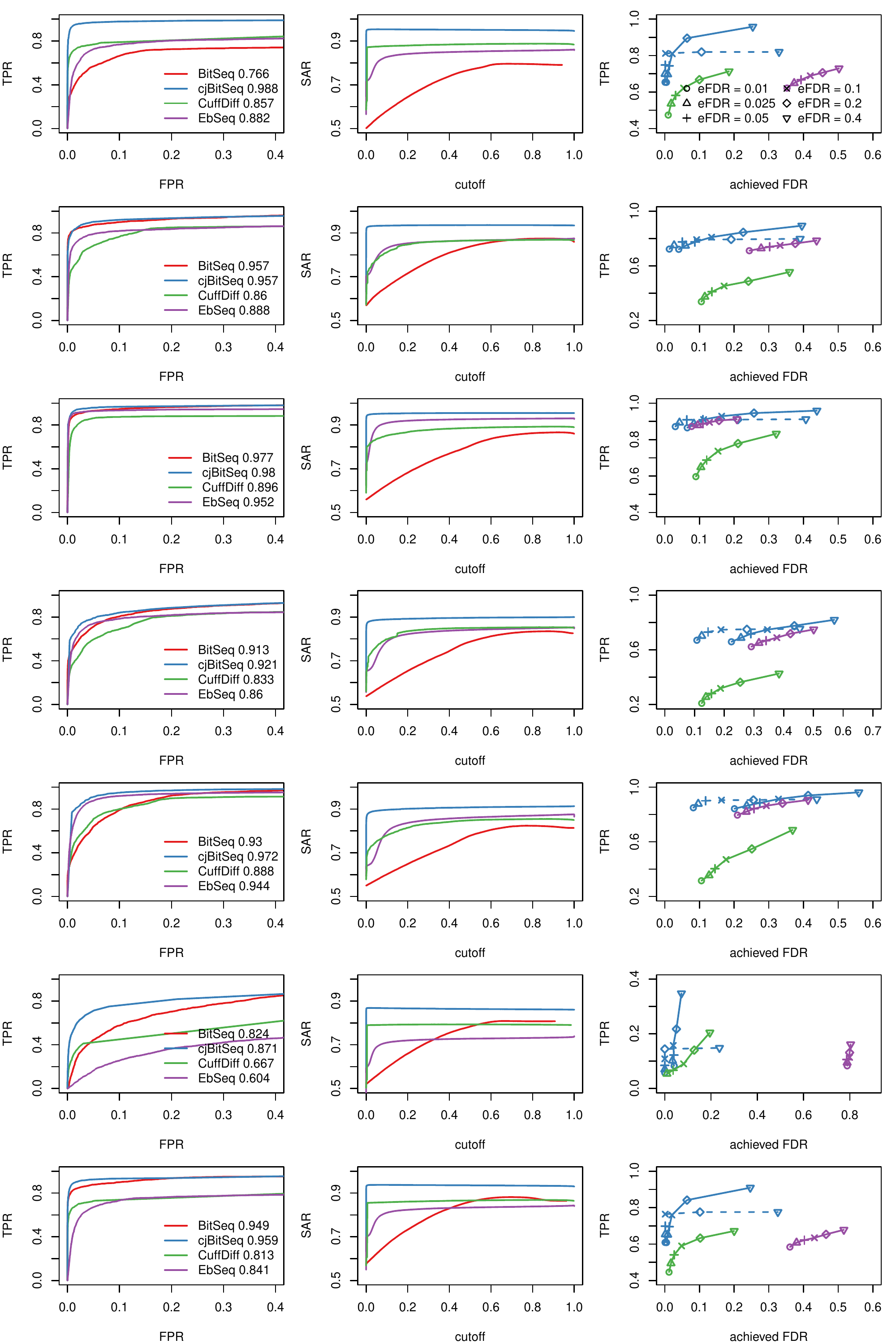}
\caption{Receiver Operating Characteristic (a), SAR measure (b) and Power-to-achieved-FDR (c) curves for scenario 1-7 (1st-7th row). The blue dashed lines correspond to the filtered cjBitSeq output by discarding transcripts with absolute log2 fold change less than 1.}\label{fig:nbNEW}
\end{figure}

\section{Results}

A set of simulation studies is used to benchmark the proposed methodology using synthetic RNA-seq reads from the {\em Drosophila melanogaster} transcriptome. The Spanki software \citep{spanki} is used for this purpose. In addition to the simulated data study we also perform a comparison for two real datasets: a low and high coverage sequencing experiment using human data and a dataset from drosophila. In all cases, the reads are mapped to the reference transcriptome using Bowtie (version 2.0.6), allowing up to 100 alignments per read. Tophat (version 2.0.9) is also used for Cufflinks.

\subsection{Evaluation of samplers}

We used a simulated dataset from $K= 630$ transcripts (more details are described in Appendix H) and compare the posterior mean estimates between short and long runs. As shown in Figure \ref{fig:rjVScj}, the collapsed sampler exhibits faster convergence than the rjMCMC sampler, hence in what follows we will only present results corresponding to the collapsed sampler. The reader is referred to the supplementary material (Appendices J and K) for further comparisons (including autocorrelation function estimation and prior sensitivity) between our two MCMC schemes.

\subsection{Simulated data}\label{sec:sim}

The input of the Spanki simulator is a set of reads per kilobase (rpk) values per sample. This file is provided under a variety of different generative scenarios. Given the input files, Spanki simulates RNA-seq reads (in fastq format) according to the specified rpk values. Seven scenarios are used to generate the data: two Poisson replicates per condition (scenario 1), three Negative Binomial replicates per condition (scenario 2), 9 Negative Binomial replicates (scenario 3), three Negative Binomial replicates per condition with five times higher variability among replicates compared to scenario 2 (scenario 4), same variability with scenario 4 but a smaller range for the mean  rpk values (scenario 5). The last two scenarios are revisions of the first scenario with smaller fold changes (scenario 6) and large differences in the number of reads between conditions (scenario 7). See supplementary Figure 9 and Appendix K for the details of the ground truth used in our simulations.

\begin{figure}[t]
\centering
\begin{tabular}{c}
\includegraphics[width = 0.08\textwidth,angle =270]{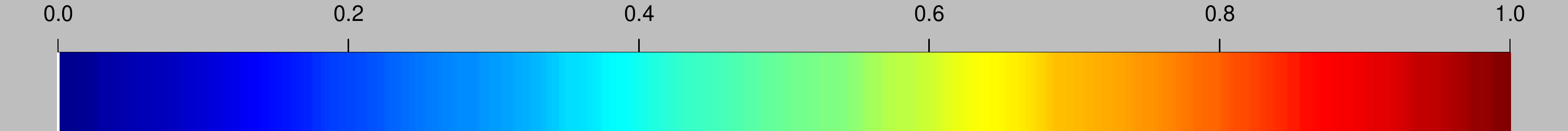}\\
\includegraphics[width = 0.99\textwidth]{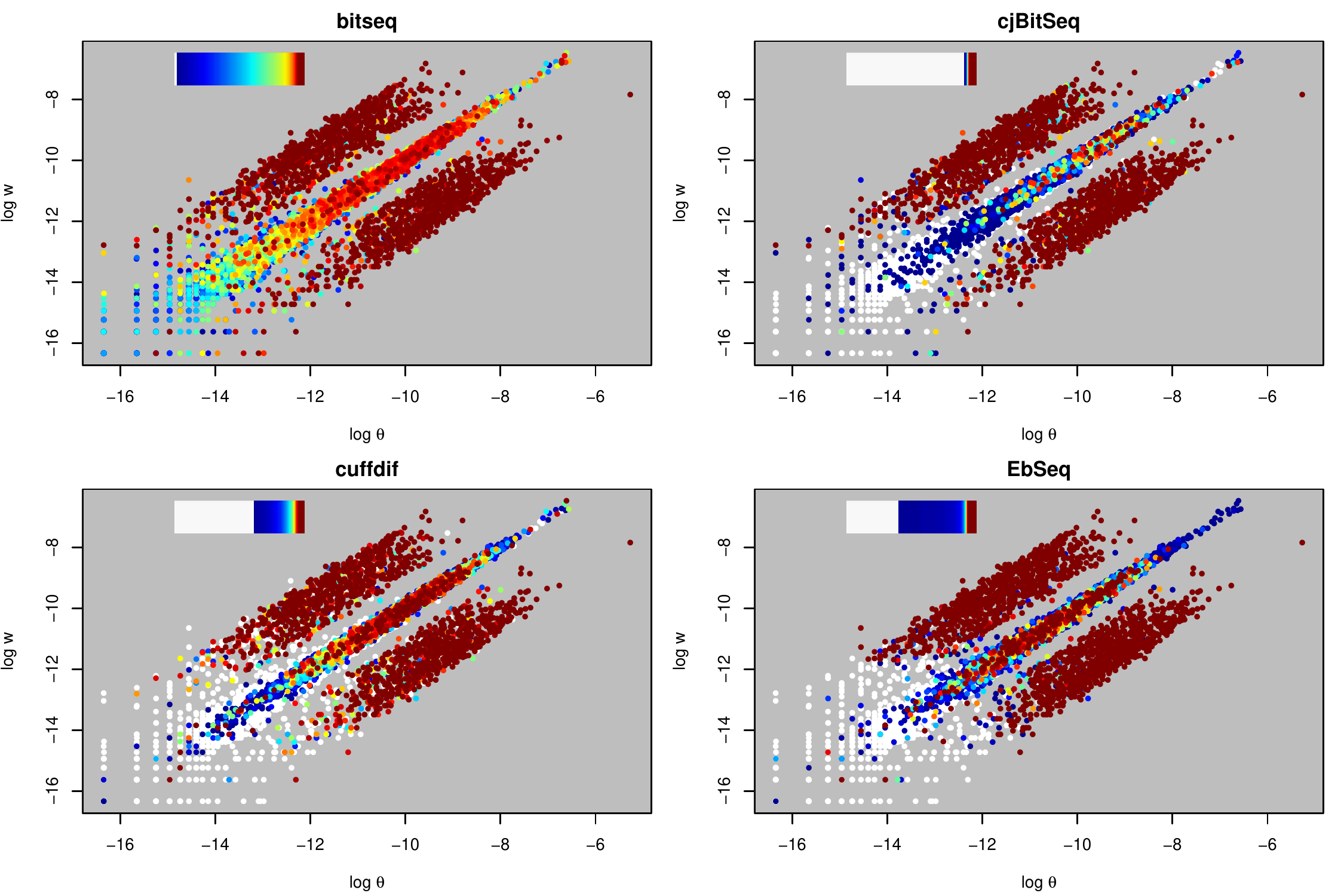}  
\end{tabular}
      \caption{True log-relative expression values for Scenario 3 (average of nine replicates per condition, $\approx 24$ million reads in total). The color corresponds to the evidence of differential expression according to each method and the legend shows the relative frequency of colors.}\label{fig:trueValuesc}
\end{figure}

Next, we applied the proposed method and compared our results against Bitseq, Cuffdiff and EBSeq, using (a) ROC, (b) SAR-measure \citep{rocr} and (c) Power-to-achieved-FDR curves, as shown in Figure \ref{fig:nbNEW}. For the comparison in (c) the FDR decision of our model is based on the rule \eqref{eq:decision}. Moreover, only methods that control the FDR are taken into account in (c), hence BitSeq Stage 2 is excluded. In addition to this FDR control procedure, we also provide adjusted rates after imposing a threshold to the log-fold change of the cjBitSeq sampler: all transcripts with estimated absolute log2 fold change less than 1 are filtered out (results correspond to the blue dashed line). A typical behaviour of the compared methods is illustrated in Figure \ref{fig:trueValuesc}, displaying true expression values used in Scenario 3. We conclude that our method infers an almost ideal classification, something that is not the case for the other methods despite the large number of replicates used.

In order to summarize our findings, Figure \ref{fig:cauc} displays the complementary area under the curve for each scenario. Averaging across all simulation scenarios, we conclude that our method is almost 2 times better than BitSeq Stage 2, 3 times better than EBSeq and 3.2 times better than Cuffdif. Finally, we compare the estimated relative abundance of transcripts against the true values used to generate the data, using the average across all replicates of a given condition. Figure \ref{fig:cauc} (bottom) displays the Mean Absolute Error between the logarithm of true transcript expression and the corresponding estimates according to each method. We see that cjBitSeq, BitSeq stage 1 and RSEM exhibit a similar behaviour, while all performing significantly better than Cufflinks. Although there is no consistent ordering among the first three methods, averaging across all experiments we conclude that  cjBitSeq is ranked first.

We have also tested the sensitivity of our method with respect to the prior distributions of differential expression \eqref{eq:jeffreys_prior} by setting $\pi = 0.5$ (see supplementary Figure 11 and the corresponding discussion in Appendix K). We conclude that the prior distribution does not affect the ranking of methods both for differential and expression estimation.

\begin{figure}[t]
\centering
\begin{tabular}{c}
\includegraphics[scale=0.32]{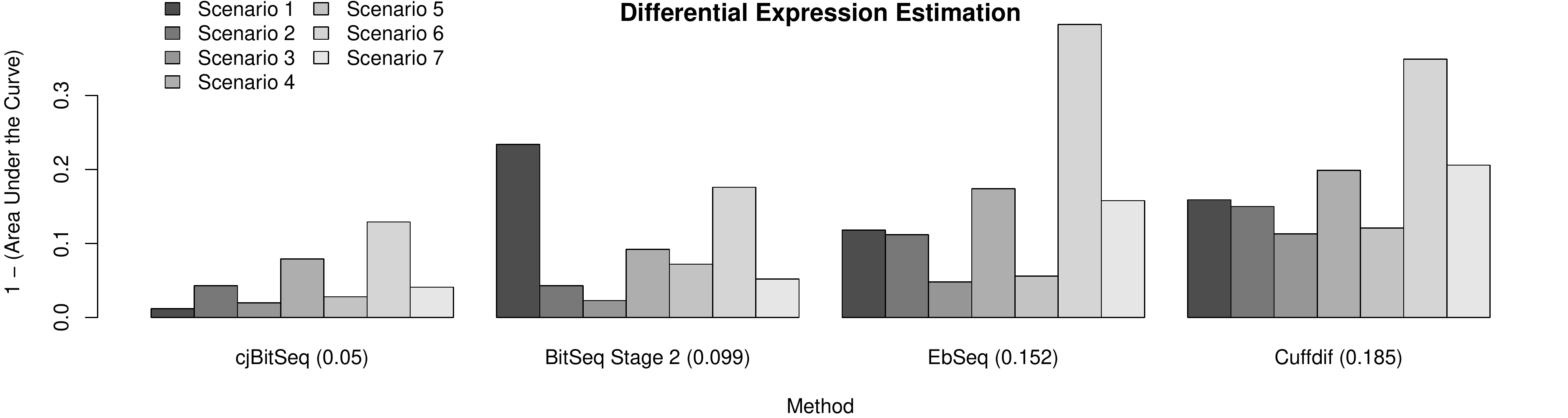}\\
\includegraphics[scale=0.32]{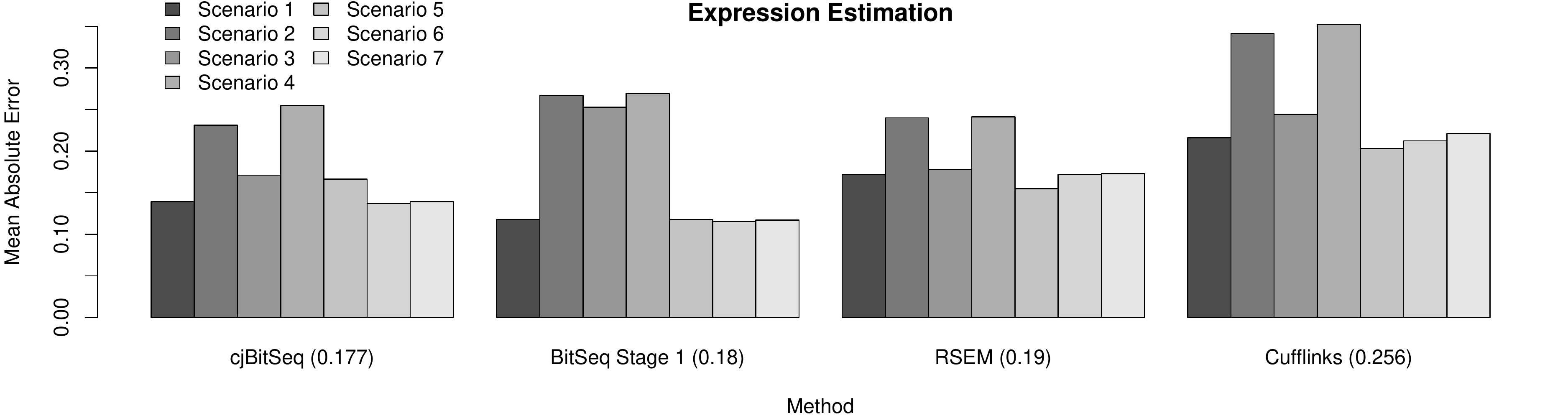}
\end{tabular}
\caption{Simulated data: Ranking of methods with respect to estimation of differential expression (top) and the log of relative expression (bottom). The methods are ordered according to the averaged complementary Area Under the Curve and Mean Absolute Error (shown in parenthesis).}\label{fig:cauc}
\end{figure}

\subsection{Human data}\label{sec:real2}

This example demonstrates the proposed algorithm to differential analysis of lung fibroblasts in response to loss of the developmental transcription factor HOXA1, see \citet{cuffdif2} for full details. There are three biological replicates in the two conditions. The experiment is carried out using two sequencing platforms: HiSeq and MiSeq, where MiSeq produced only $23\%$ of the number of reads in the HiSeq data. Here, these reads are mapped to hg19 (UCSC annotation) using Bowtie 2, consisting of $K = 48009$ transcripts. In total, there are 96969106 and 21271542 mapped reads for HiSeq and MiSeq sequencers, respectively. \citet{cuffdif2} demonstrated the ability of Cuffdiff2 to recover the transcript dynamics from the HOXA1 knockdown when using the significantly smaller amount of data generated by MiSeq compared to HiSeq. 

Applying cjBitSeq to the MiSeq data recovers $50.2\%$ of the DE transcripts from HiSeq. On the other hand, there are 183  transcripts reported as DE with the MiSeq data but not the HiSeq data (Figures \ref{fig:hoxa}(a) and \ref{fig:hoxa}(b)). The corresponding percentages for BitSeq stage 2, EBSeq and Cuffdiff are $43.3\%$, $40.6\%$ and $15.7\%$, respectively (see Figures \ref{fig:hoxa}(b), \ref{fig:hoxa}(c) and \ref{fig:hoxa}(d)). We conclude that the proposed model returns the largest proportion of consistently DE transcripts between platforms. The number of transcripts which are simultaneously reported as DE is equal to 2173 and 390 for HiSeq and MiSeq data, respectively (see Figures \ref{fig:hoxa1}.a and \ref{fig:hoxa1}.b). Finally, cjBitSeq and EBSeq provide the most highly correlated classifications (see Table 1 of supplementary material).

\begin{figure}[t]
\centering
\begin{tabular}{cccc}
cjBitSeq& 
BitSeq& 
EBSeq&
CuffDiff\\
\includegraphics[scale=0.18]{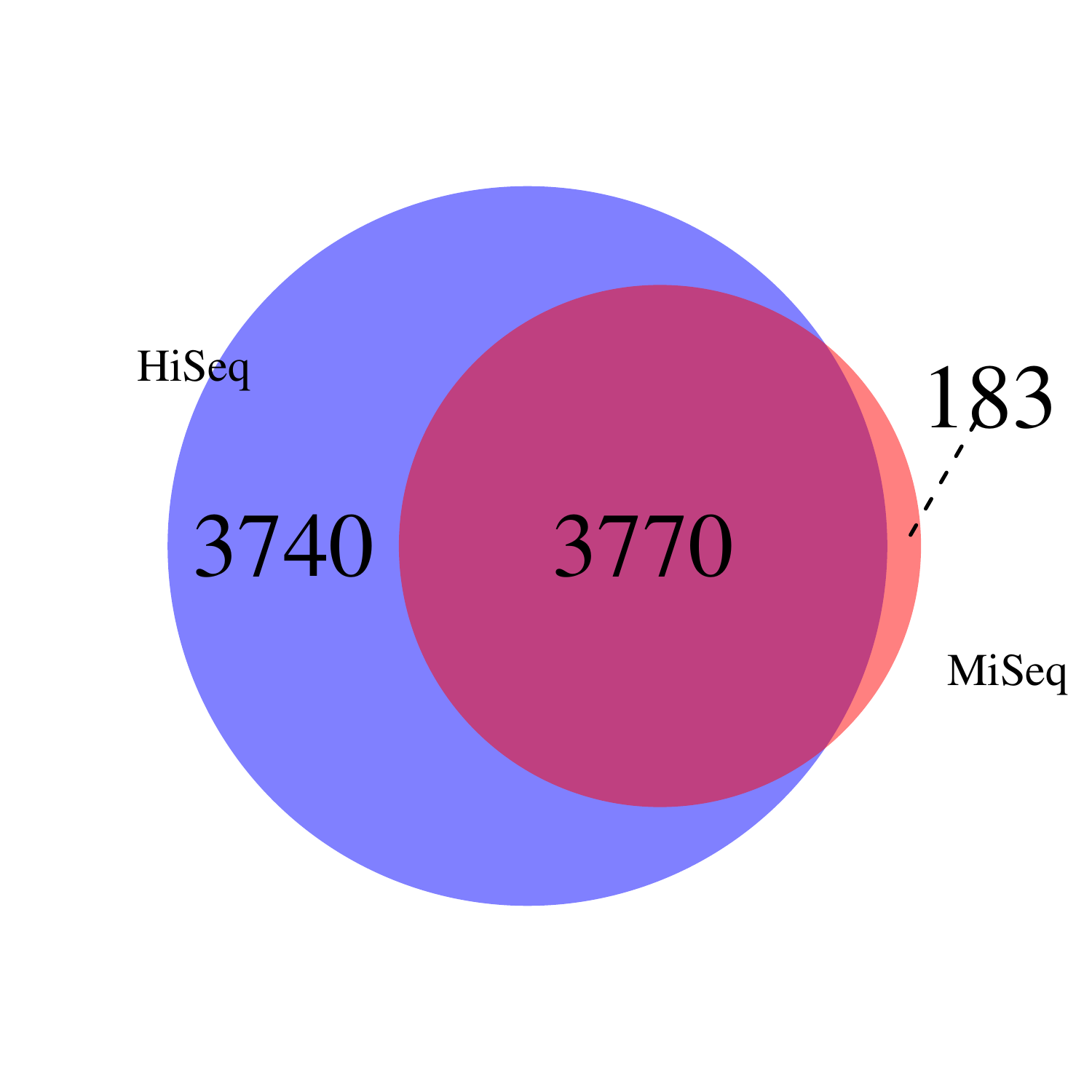} &
\includegraphics[scale=0.18]{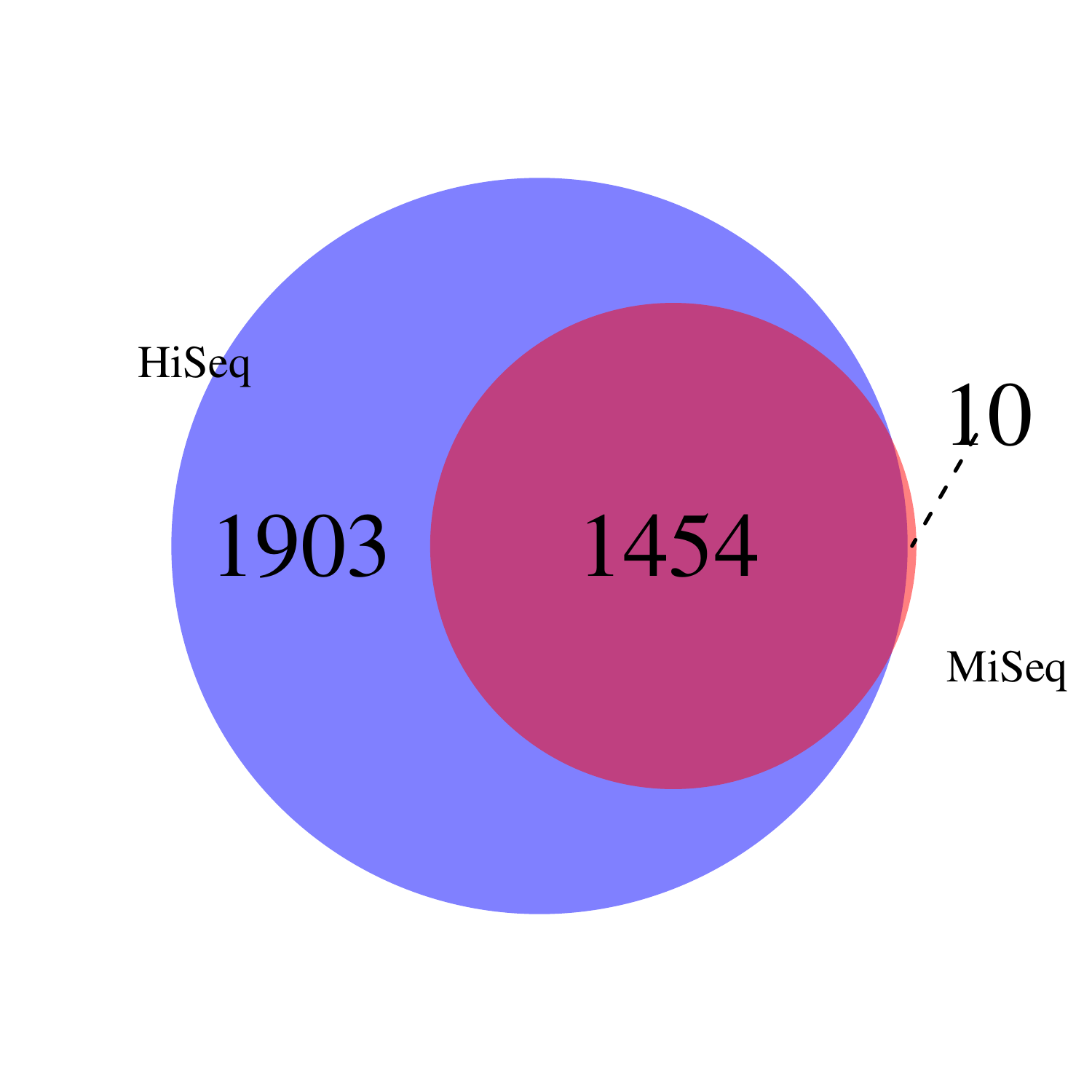} &
\includegraphics[scale=0.18]{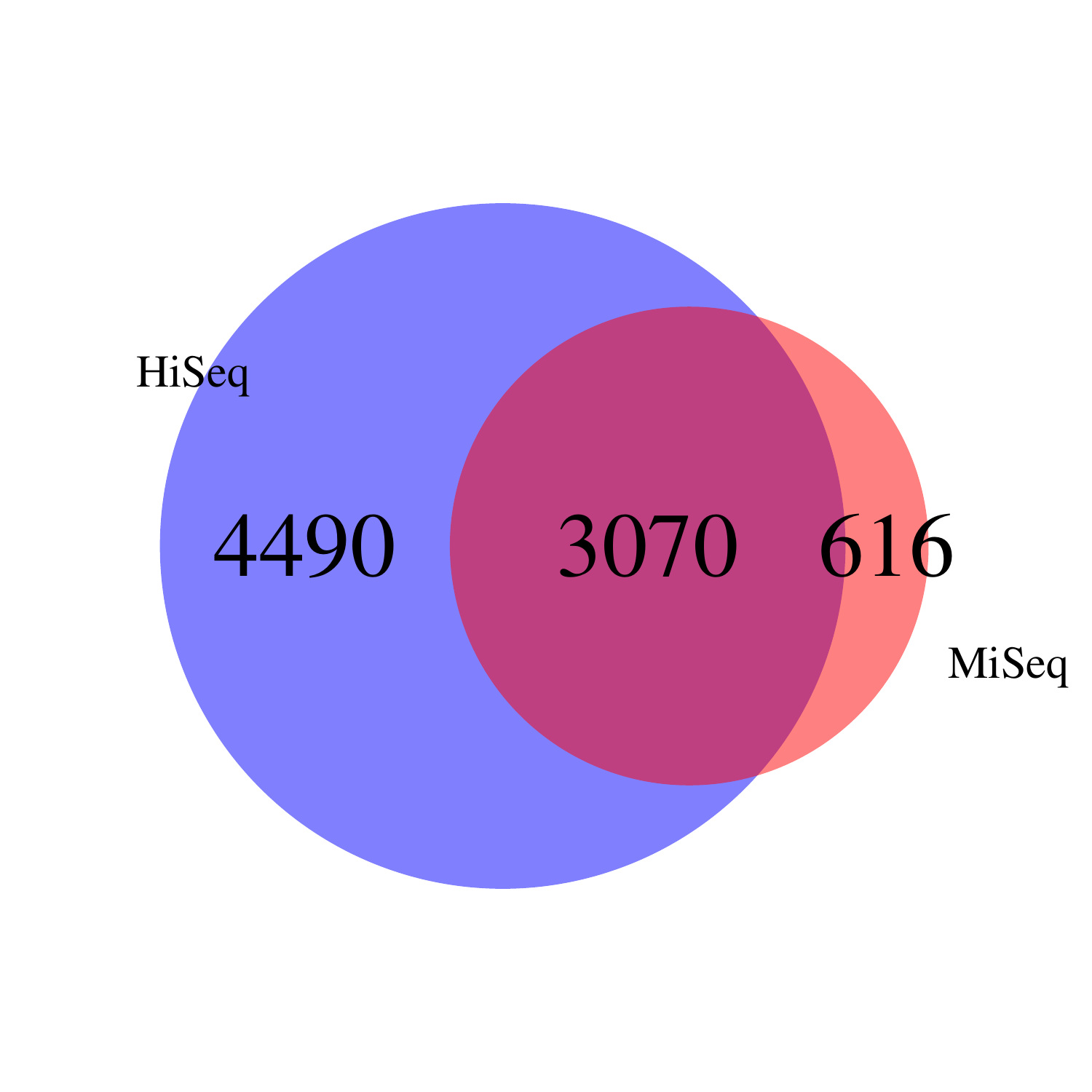}&  
\includegraphics[scale=0.18]{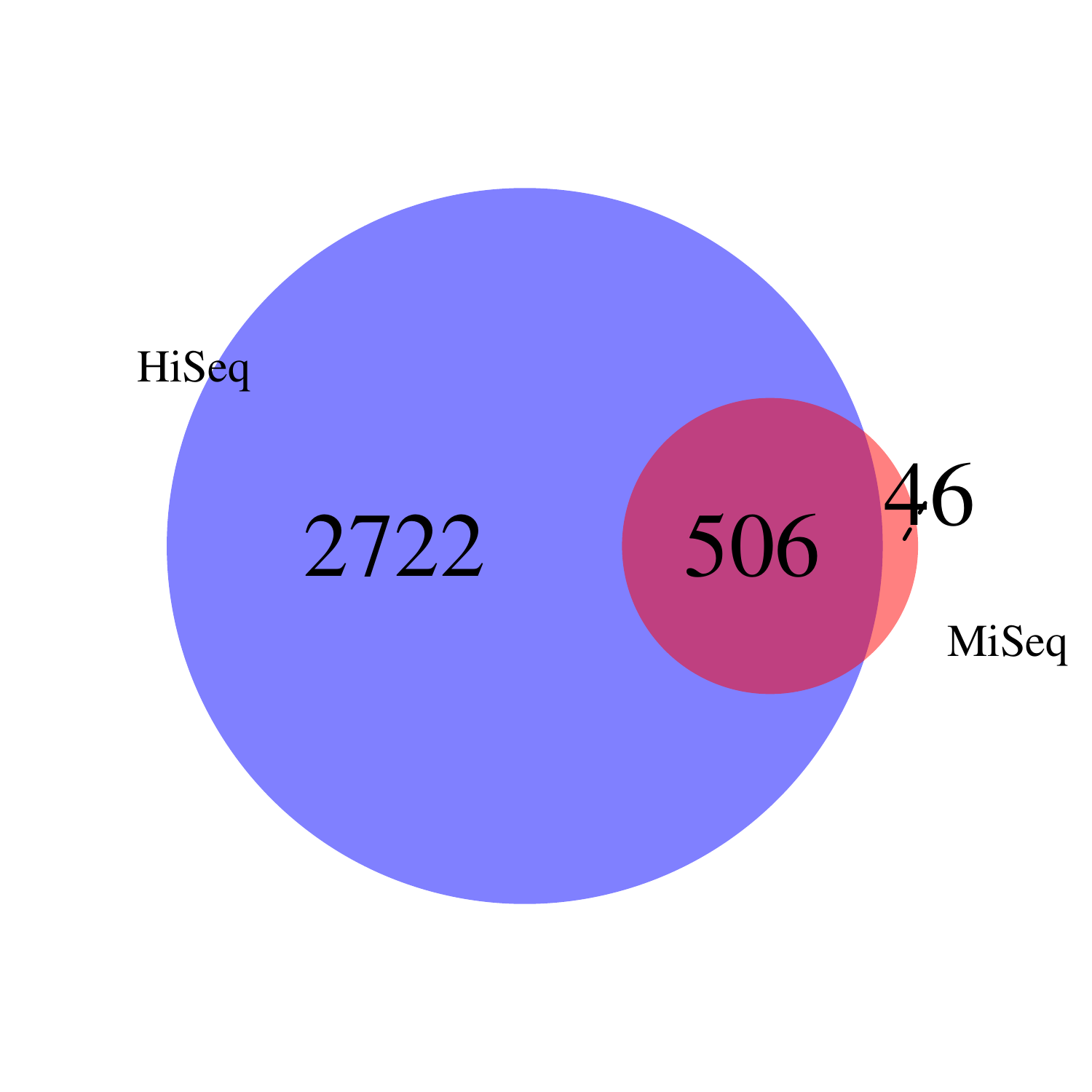}\\
(a) $50.2\%$ & 
(b) $43.3\%$ & 
(c) $40.6\%$  &
(d) $15.7\%$
\end{tabular}
\caption{HOXA1 knockdown dataset: Significant transcript list returned by cjBitSeq (a), BitSeq (b), EBSeq (c) and CuffDiff (d) when using HiSeq (blue) and MiSeq (red) data. FDR for cjBitSeq, EBSeq and CuffDiff set to $0.05$, while for BitSeq: $\mbox{PPLR} < 0.025$ or $\mbox{PPLR}>0.975$}\label{fig:hoxa1}
\end{figure}

\begin{figure}[t]
\centering
\begin{tabular}{cc}
\hspace{-2ex}\includegraphics[scale=0.40]{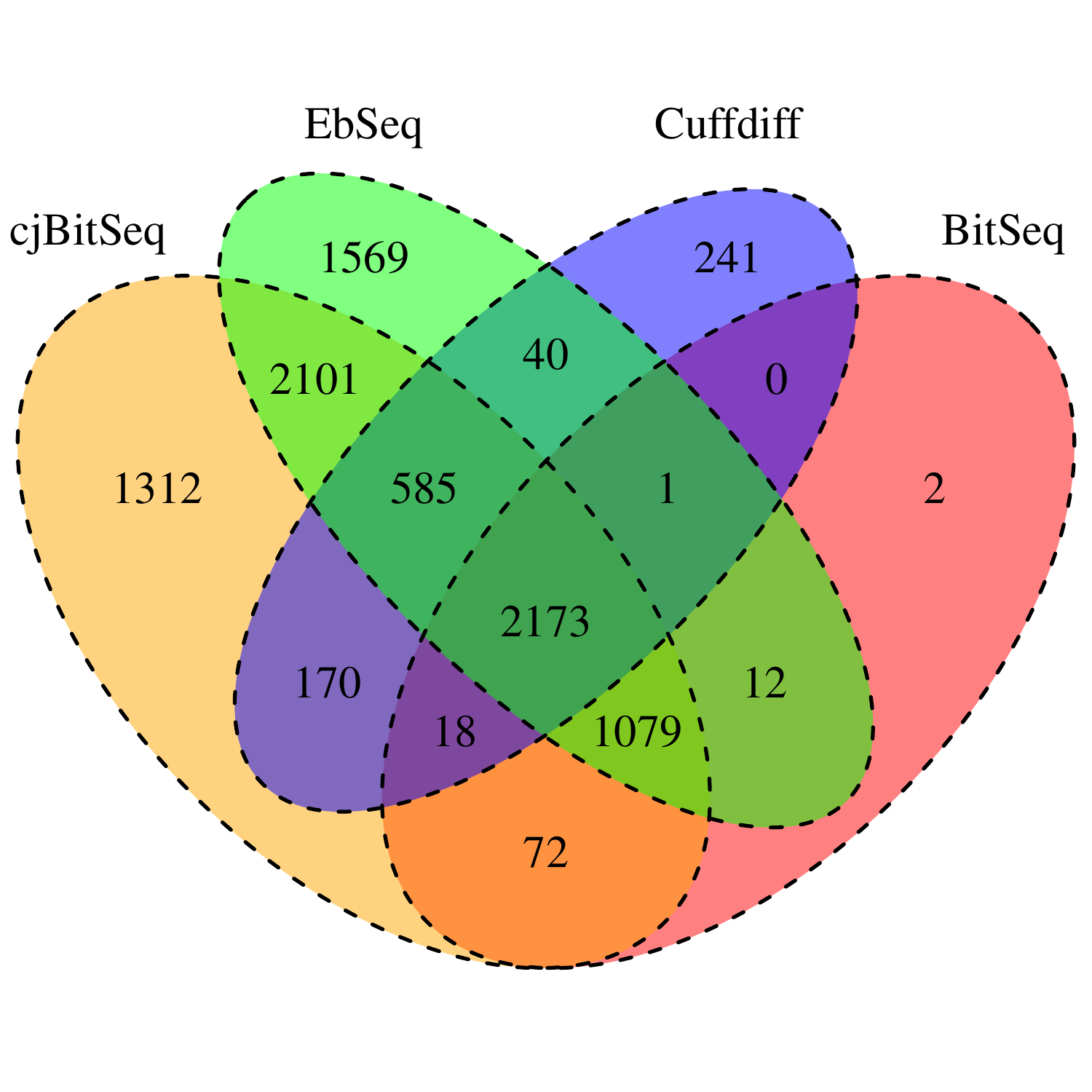} & 
\hspace{-2ex}\includegraphics[scale=0.40]{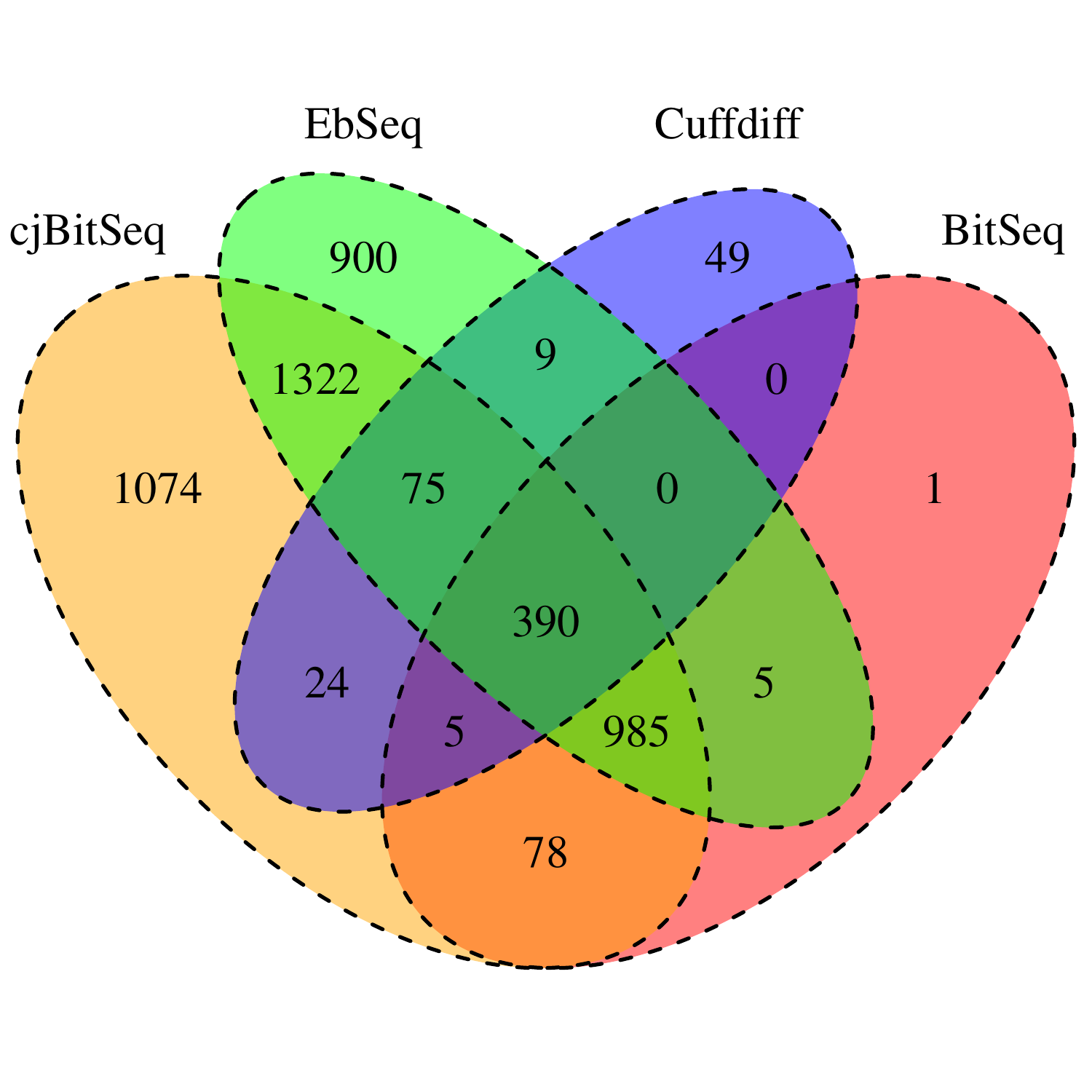} \\
\hspace{-2ex}(a) HiSeq data &
\hspace{-2ex} (b) MiSeq data
\end{tabular}
\caption{HOXA1 knockdown dataset: Contiguity of methods when using HiSeq (a) and MiSeq (b) data. FDR for cjBitSeq, EBSeq and CuffDiff set to $0.05$, while for BitSeq: $\mbox{PPLR} < 0.025$ or $\mbox{PPLR}>0.975$}\label{fig:hoxa}
\end{figure}

\section{Discussion}\label{sec:dis}

We have proposed a probabilistic model for the simultaneous estimation of transcript expression and differential expression between conditions. Building upon the BitSeq framework, the new Bayesian hierarchical model is conjugate for fixed dimension variables. A by-product is a new interpretation of the Generalized Dirichlet distribution, which naturally appears in \eqref{eq:u_gibbs} as the full conditional distribution of a random variable describing one of the free parameters corresponding to two proportion vectors under the constraint that some of the weights are equal to each other. We implemented two MCMC samplers, a reversible-jump and collapsed Gibbs sampler, and we found the collapsed Gibbs sampler to converge faster. To greatly reduce the dimensionality of the parameter space for inference we developed a transcript clustering approach which allows inference to be carried out independently on subsets of transcripts that share aligned reads. According to Lemma 3 in the supplementary material (Appendix H), this clustered version of the vanilla algorithm converges to the proper marginal distribution for each cluster. Thus, the algorithm has the nice property that it can be run in parallel for each cluster, while the memory requirements are quite low, providing a simple parallelisation option.

The applications to simulated and real RNA-seq data reveals that the proposed method is highly competitive with the current state of the art software dealing with DE analysis at the transcript level. Note that the simulated data was generated under a variety of different scenarios and including different levels of replication and biological variation. We simulated transcript RPK values with variability following either the Poisson or the Negative Binomial distribution with various levels for the dispersion around the mean. We conclude that our method is quite robust in expression estimation and in classifying transcripts as DE or not. Compared to standard two-stage pipelines it is ranked as the best method under a wide range of generative scenarios.

RNA-seq data are usually replicated such that there is more than one datasets available for each condition. In such a way, biological variability between repetitions of the same experiment can be taken into account. The amount of variability between replicates can be quite high depending on the experimental conditions. Two-stage approaches for estimating differential expression are strongly focused on modelling this inter-replicate variability. This is not the case for our method at present and all replicates of a given condition are effectively pooled together prior to inference. Modelling the variability between replicates would significantly increase the complexity of our approach as it is technically challenging to retain conjugacy. However, according to our simulation studies, we have found that pooling replicates together and jointly estimating expression and differential expression balances the loss through ignoring variability between replicates in many cases. Nevertheless, an extension to also model inter-replicate variability would be very interesting and could be expected to improve performance when there is high inter-replicate dispersion. 

The proposed method was developed focusing to the comparison of two conditions and its extension to more general settings is another interesting area for future research. 
A remarkable property of the parameterization introduced in Equations \eqref{eq:theta} and \eqref{eq:w} is that its extension is straightforward when $J > 2$: it can be shown that in this case there is one parameter of constant dimension and $J-1$ parameters of varying dimension. Let $\bs u = \bs u^{(1)}$ be the vector of relative abundances for condition 1. For a given condition $j = 2,\ldots,J$ define a vector $\bs v_j$ containing the expression of transcripts not being equal to any of the previous conditions $1,\ldots,j-1$. Note that $\bs v_j$ is a random variable with varying length (between 0 and $K$). Furthermore, for $j\geqslant 2$ define the vectors $\bs u_k^{(j)}$, $k = 1,\ldots,j-1$, containing the expression of  transcripts shared with condition $k$ but not with $1,\ldots,k-1$. It follows that $\bs u_k^{(j)}$ can be written as a function of $\bs u^{(1)}$ and $\bs v_k$, $k = 1,\ldots,j-1$. Hence, the relative transcript expression vector for condition $j$  can be expressed as a suitable permutation of $(\bs u_1^{(j)},\ldots,\bs u_{j-1}^{(j)},\bs v_{j})$. However, the question of whether the proposed model stays conjugate for fixed dimension updates remains an open problem. If yes, the design of more sophisticated move-types between different models would be also crucial to the convergence of the algorithm since the search space is increased.

The source code of the proposed algorithm is compiled for LINUX distributions and it is available at {\tt https://github.com/mqbssppe/cjBitSeq}. The simulation pipeline is available at {\tt https://github.com/ManchesterBioinference/cjBitSeq\_benchmarking}. Cluster discovery and MCMC sampling is coded in R and C++, respectively. Parallel runs of the MCMC scheme are implemented using the GNU parallel \citep{parallel} shell tool. The computing times needed for our datasets are reported in supplementary Table 2.

\section*{Acknowledgments}
The research was supported by MRC award MR/M02010X/1, BBSRC award BB/J009415/1 and EU FP7 project RADIANT (grant 305626). The authors acknowledge the assistance given by IT Services and the use of the Computational Shared Facility at The University of Manchester. We also thank the editor and two anonymous reviewers for their helpful comments and suggestions which helped us to improve the manuscript.

\noindent
\textbf{Supplementary material:}
We provide the proofs of our Lemmas and Theorems, a detailed description of the reversible jump proposal and the Gibbs updates of state vector of the collapsed sampler. Also included are details of alignment probabilities and some useful properties of the Generalized Dirichlet distribution. We also perform various comparisons between the rjMCMC and collapsed samplers and examine their prior sensitivity. Finally we describe the generative schemes for the simulation study and some guidelines for the practical implementation of the algorithm.

\bibliographystyle{rss.bst}      
\bibliography{cjBitSeq}   

\begin{thebibliography}{34}
\expandafter\ifx\csname natexlab\endcsname\relax\def\natexlab#1{#1}\fi
\expandafter\ifx\csname url\endcsname\relax
  \def\url#1{\texttt{#1}}\fi
\expandafter\ifx\csname urlprefix\endcsname\relax\def\urlprefix{URL}\fi

\bibitem[{Anders and Huber(2010)}]{Anders:2010}
Anders, A. and Huber, W. (2010) Differential expression analysis for sequence
  count data.
\newblock \textit{Genome Biology}, \textbf{11}, R106.

\bibitem[{Benjamini and Hochberg(1995)}]{benjamini1995}
Benjamini, Y. and Hochberg, Y. (1995) Controlling the false discovery rate: a
  practical and powerful approach to multiple testing.
\newblock \textit{Journal of the Royal Statistical Society. Series B
  (Methodological)}, 289--300.

\bibitem[{Connor and Mosimann(1969)}]{connor}
Connor, R.~J. and Mosimann, J.~E. (1969) Concepts of independence for
  proportions with a generalization of the {D}irichlet distribution.
\newblock \textit{Journal of the American Statistical Association},
  \textbf{64}, 194--206.

\bibitem[{Gelfand and Smith(1990)}]{gelfand}
Gelfand, A. and Smith, A. (1990) Sampling-based approaches to calculating
  marginal densities.
\newblock \textit{Journal of American Statistical Association}, \textbf{85},
  398--409.

\bibitem[{Geman and Geman(1984)}]{geman}
Geman, S. and Geman, D. (1984) Stochastic relaxation, {G}ibbs distributions,
  and the {B}ayesian restoration of images.
\newblock \textit{IEEE Transactions on Pattern Analysis and Machine
  Intelligence}, \textbf{PAMI-6}, 721--741.

\bibitem[{Glaus et~al.(2012)Glaus, Honkela and Rattray}]{bitseq}
Glaus, P., Honkela, A. and Rattray, M. (2012) Identifying differentially
  expressed transcripts from {R}{N}{A}-{S}eq data with biological variation.
\newblock \textit{Bioinformatics}, \textbf{28}, 1721--1728.

\bibitem[{Green(1995)}]{Green:95}
Green, P.~J. (1995) Reversible jump {M}arkov chain {M}onte {C}arlo computation
  and {B}ayesian model determination.
\newblock \textit{Biometrika}, \textbf{82}, 711--732.

\bibitem[{Gu et~al.(2014)Gu, Wang, Halakivi-Clarke, Clarke and Xuan}]{badge}
Gu, J., Wang, X., Halakivi-Clarke, L., Clarke, R. and Xuan, J. (2014)
  {B}{A}{D}{G}{E}: A novel {B}ayesian model for accurate abundance
  quantification and differential analysis of {R}{N}{A}-{S}eq data.
\newblock \textit{BMC Bioinformatics 2014}, \textbf{15}.

\bibitem[{Hensman et~al.(2015)Hensman, Papastamoulis, Glaus, Honkela and
  Rattray}]{bitseqVB}
Hensman, J., Papastamoulis, P., Glaus, P., Honkela, A. and Rattray, M. (2015)
  Fast and accurate approximate inference of transcript expression from
  {R}{N}{A}-seq data.
\newblock \textit{Bioinformatics}, \textbf{31}, 3881--3889.

\bibitem[{Jeffreys(1946)}]{jeffreys}
Jeffreys, H. (1946) An invariant form for the prior probability in estimation
  problems.
\newblock \textit{Proceedings of the Royal Society of London. Series A,
  Mathematical and Physical Sciences}, \textbf{186}, 453--461.

\bibitem[{Jia et~al.(2015)Jia, Guan, Yang, Xiao, Tang, Moravec, Margulies,
  Cappola, Li and Li}]{Jia2015}
Jia, C., Guan, W., Yang, A., Xiao, R., Tang, W. H.~W., Moravec, C.~S.,
  Margulies, K.~B., Cappola, T.~P., Li, C. and Li, M. (2015) {M}eta{D}iff:
  differential isoform expression analysis using random-effects
  meta-regression.
\newblock \textit{BMC Bioinformatics}, \textbf{16}, 1--12.

\bibitem[{Langmead et~al.(2009)Langmead, Trapnell, Pop and Salzberg}]{bowtie}
Langmead, B., Trapnell, C., Pop, M. and Salzberg, S. (2009) Ultrafast and
  memory-efficient alignment of short {D}{N}{A} sequences to the human genome.
\newblock \textit{Genome Biology}, \textbf{10}.

\bibitem[{Leng et~al.(2013)Leng, Dawson, Thomson, Ruotti, Rissman, Smits, Haag,
  Gould, Stewart and Kendziorski}]{EBSeq}
Leng, N., Dawson, J.~A., Thomson, J.~A., Ruotti, V., Rissman, A.~I., Smits,
  B.~M., Haag, J.~D., Gould, M.~N., Stewart, R.~M. and Kendziorski, C. (2013)
  {E}{B}{S}eq: An empirical {B}ayes hierarchical model for inference in
  {R}{N}{A}-{S}eq experiments.
\newblock \textit{Bioinformatics}.

\bibitem[{Li and Dewey(2011)}]{rsem}
Li, B. and Dewey, C.~N. (2011) {R}{S}{E}{M}: accurate transcript quantification
  from {R}{N}{A}-seq data with or without a reference genome.
\newblock \textit{BMC Bioinformatics}, \textbf{12}, 323.

\bibitem[{Lindley(1971)}]{lindley}
Lindley, D.~V. (1971) \textit{Making Decisions}.
\newblock Willey, New York.

\bibitem[{Mortazavi et~al.(2008)Mortazavi, Williams, McCue, Schaeffer and
  Wold}]{mort}
Mortazavi, A., Williams, B., McCue, K., Schaeffer, L. and Wold, B. (2008)
  Mapping and quantifying mammalian transcriptomes by {R}{N}{A}-{S}eq.
\newblock \textit{Nat Methods}, \textbf{5}, 621--628.

\bibitem[{M{\"u}ller et~al.(2006)M{\"u}ller, Parmigiani and Rice}]{fdrISBA}
M{\"u}ller, P., Parmigiani, G. and Rice, K. (2006) {F}{D}{R} and {B}ayesian
  multiple comparisons rules.
\newblock \textit{Proc. Valencia / ISBA 8th World Meeting on Bayesian
  Statistics}.

\bibitem[{M{\"u}ller et~al.(2004)M{\"u}ller, Parmigiani, Robert and
  Rousseau}]{fdrJASA}
M{\"u}ller, P., Parmigiani, G., Robert, C. and Rousseau, J. (2004) Optimal
  sample size for multiple testing.
\newblock \textit{Journal of the American Statistical Association},
  \textbf{99}, 990--1001.

\bibitem[{Nariai et~al.(2013)Nariai, Hirose, Kojima and Nagasaki}]{newvb}
Nariai, N., Hirose, O., Kojima, K. and Nagasaki, M. (2013) {TIGAR}: transcript
  isoform abundance estimation method with gapped alignment of {RNA-Seq} data
  by variational {B}ayesian inference.
\newblock \textit{Bioinformatics}, \textbf{18}, 2292--2299.

\bibitem[{Nicolae et~al.(2011)Nicolae, Mangul, Mandoiu and Zelikovsky}]{isoEM}
Nicolae, M., Mangul, S., Mandoiu, I. and Zelikovsky, A. (2011) Estimation of
  alternative splicing isoform frequencies from {R}{N}{A}-{S}eq data.
\newblock \textit{Algorithms for Molecular Biology}, \textbf{6:9}.

\bibitem[{Papastamoulis et~al.(2014)Papastamoulis, Hensman, Glaus and
  Rattray}]{papVB}
Papastamoulis, P., Hensman, J., Glaus, P. and Rattray, M. (2014) Improved
  variational {B}ayes inference for transcript expression estimation.
\newblock \textit{Statistical applications in Genetics and Molecular Biology},
  \textbf{13}, 213--216.

\bibitem[{Papastamoulis and Iliopoulos(2009)}]{papRJ}
Papastamoulis, P. and Iliopoulos, G. (2009) Reversible jump {M}{C}{M}{C} in
  mixtures of normal distributions with the same component means.
\newblock \textit{Computational Statistics and Data Analysis}, \textbf{53},
  900--911.

\bibitem[{Richardson and Green(1997)}]{Richardson:97}
Richardson, S. and Green, P.~J. (1997) On {B}ayesian analysis of mixtures with
  an unknown number of components.
\newblock \textit{Journal of the Royal Statistical Society: Series B},
  \textbf{59}, 731--758.

\bibitem[{Robinson et~al.(2010)Robinson, McCarthy and Smyth}]{edger}
Robinson, M., McCarthy, D. and Smyth, G. (2010) edge{R}: a {B}ioconductor
  package for differential expression analysis of digital gene expression data.
\newblock \textit{Bioinformatics}, \textbf{26}, 139--140.

\bibitem[{Rossell et~al.(2014)Rossell, Attolini, C, Kroiss and
  Stocker}]{casper}
Rossell, D., Attolini, C, S.-O., Kroiss, M. and Stocker, A. (2014) Quantifying
  alternative splicing from paired-end {R}{N}{A}-{S}equencing data.
\newblock \textit{Annals of Applied Statistics}, \textbf{8}, 309--330.

\bibitem[{Sing et~al.(2005)Sing, Sander, Beerenwinkel and Lengauer}]{rocr}
Sing, T., Sander, O., Beerenwinkel, N. and Lengauer, T. (2005) {R}{O}{C}{R}:
  visualizing classifier performance in {R}.
\newblock \textit{Bioinformatics}, \textbf{21}, 7881.

\bibitem[{Storey(2003)}]{storey2003}
Storey, J.~D. (2003) The positive false discovery rate: A {B}ayesian
  interpretation and the q-value.
\newblock \textit{Annals of statistics}, 2013--2035.

\bibitem[{Sturgill et~al.(2013)Sturgill, Malone, Sun, Smith, Rabinow, Samson
  and Oliver}]{spanki}
Sturgill, J., Malone, J.~H., Sun, X., Smith, H.~E., Rabinow, L., Samson, M.~L.
  and Oliver, B. (2013) Design of {R}{N}{A} splicing analysis null models for
  post hoc filtering of drosophila head {R}{N}{A}-{S}eq data with the splicing
  analysis kit ({S}panki).
\newblock \textit{BMC Bioinformatics}, \textbf{14}, 320.

\bibitem[{Tange(2011)}]{parallel}
Tange, O. (2011) {G}{N}{U} parallel - the command-line power tool.
\newblock \textit{;login: The USENIX Magazine}, \textbf{36}, 42--47.
\newblock \urlprefix\url{http://www.gnu.org/s/parallel}.

\bibitem[{Trapnell et~al.(2013)Trapnell, Hendrickson, Sauvageau, Goff, Rinn and
  Pachter}]{cuffdif2}
Trapnell, C., Hendrickson, D.~G., Sauvageau, M., Goff, L., Rinn, J.~L. and
  Pachter, L. (2013) Differential analysis of gene regulation at transcript
  resolution with {R}{N}{A}-{S}eq.
\newblock \textit{Nature Biotechnology}, \textbf{31}, 46--53.

\bibitem[{Trapnell et~al.(2009)Trapnell, Pachter and Salzberg}]{tophat}
Trapnell, C., Pachter, L. and Salzberg, S. (2009) Top{H}at: discovering splice
  junctions with {R}{N}{A}-{S}eq.
\newblock \textit{Bioinformatics}, \textbf{25}, 1105--1111.

\bibitem[{Trapnell et~al.(2010)Trapnell, Williams, Pertea, Mortazavi, Kwan, van
  Baren, Salzberg, Wold and Pachter}]{cuffdif}
Trapnell, C., Williams, B.~A., Pertea, G., Mortazavi, A., Kwan, G., van Baren,
  M.~J., Salzberg, S.~L., Wold, B.~J. and Pachter, L. (2010) Transcript
  assembly and quantification by {R}{N}{A}-{S}eq reveals unannotated
  transcripts and isoform switching during cell differentiation.
\newblock \textit{Nature Biotechnology}, \textbf{28}, 511--515.

\bibitem[{Wong(1998)}]{Wong:98}
Wong, T. (1998) Generalized {D}irichlet distribution in {B}ayesian analysis.
\newblock \textit{Applied Mathematics and Computation}, \textbf{97}, 165--181.

\bibitem[{Wong(2010)}]{Wong:10}
--- (2010) Parameter estimation for generalized {D}irichlet distributions from
  the sample estimates of the first and the second moments of random variables.
\newblock \textit{Computational Statistics and Data Analysis}, \textbf{54},
  1756--1765.

\end{thebibliography}

\appendix
\section*{Appendices}
\addcontentsline{toc}{section}{Appendices}
\renewcommand{\thesubsection}{\Alph{subsection}}
\numberwithin{equation}{section}

\section{Alignment probability}

In this section the component specific density \eqref{standardBitseq} is defined. For single-end reads, let $\ell_i>0$ denotes the length of read $x_i$, $i = 1,\ldots,n$. Assume that $x_i$ aligns at some position $p$ of a given transcript $k$, $k = 1,\ldots,K$ and that the corresponding transcript length equals to $L_k>0$. Note that both $L_k$, $\ell_i$ are known quantities. The general form of observing this alignment equals to
\begin{equation}
f_k(x_i) = P(x_i = p|k) = \frac{b_k(p)}{\sum_{j=1}^{L_k - \ell_i + 1}b_k(j)},
\end{equation}
where $b_k(j)$ denotes the bias for a particular position $p$ on transcript $k$. In case of a Uniform read distribution, the previous equation reduces to:
\begin{equation}
f_k(x_i) = \frac{1}{L_k - \ell_i + 1}.
\end{equation}
More complex choices are also available. In particular, a separate variable length Markov is used to capture the position and sequence specific biases for the $5'$ and $3'$ ends of the fragment. For more details the reader is referred to Glaus et al.~(2012). 

In case of paired-end reads, the fragment length $\ell$ is also taken into account. The fragment length distribution $f(\ell|k)$ is assumed to be log-normal with parameters given by the user or estimated from read pairs with only a single transcript alignment. In this case the alignment probability of a read pair is given as
\begin{equation}
f_k(x_i) = f_k(x_i = p, \ell) = f(\ell|k)\frac{b_k(p)}{\sum_{j=1}^{L_k - \ell_i + 1}b_k(j)}.
\end{equation}

Finally, the alignment probabilities also take into account base-calling errors using the Phred score. For full details see Glaus et al.~(2012). In our presented examples we assumed the Uniform read distribution.

\begin{figure}[t]
\centering
\begin{tabular}{c}
\includegraphics[width = 0.99\textwidth]{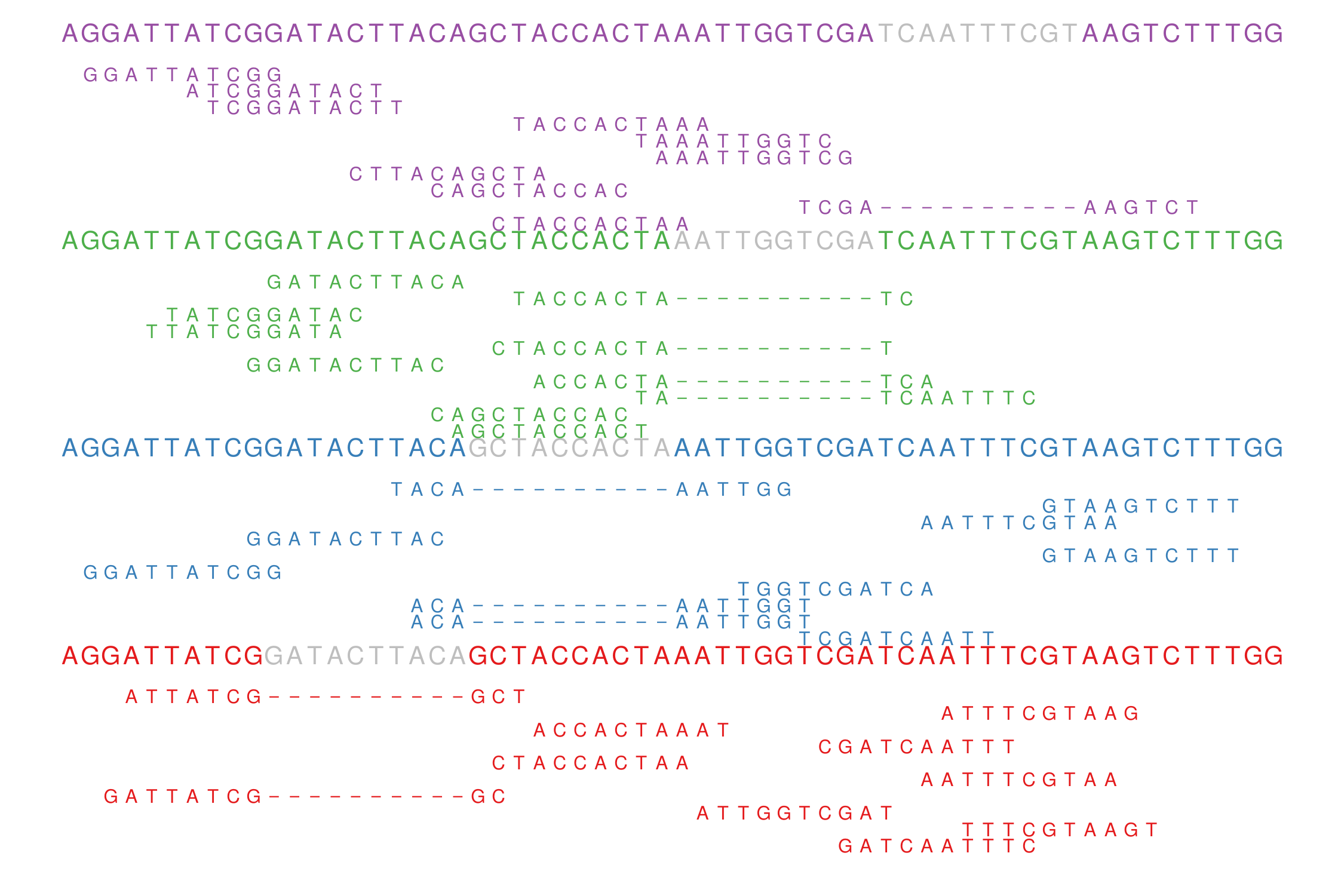}
\end{tabular}
      \caption{Illustration of the RNA-seq sampling scheme using single reads and a small set of four transcripts (red, blue, green and purple). Gray color corresponds to  skipped regions (exons). From each transcript we simulated 10 reads each one consisting of 10 base pairs, displayed under each transcript.}\label{fig:reads}
\end{figure}

The sampling scheme of the RNA-seq procedure for single-end reads is displayed in Figure \ref{fig:reads}. The four long sequences of letters correspond to transcripts which share specific parts of their sequence. The gray coloured regions are skipped, so each transcript is consisting only from the remaining region (coloured in red, blue, green and purple). The short reads are randomly generated sequences from each transcript. Note that most reads align to more than one transcript.

\section{The Generalized Dirichlet distribution}\label{sec:generalizeddirichlet}

This generalization of the Dirichlet distribution was originally introduced by Connor and Mossiman (1969). The most prominent difference with a typical Dirichlet is that the Generalized Dirichlet family has a richer covariance structure. For example, only negative correlation between any pairs of variables is allowed under the Dirichlet distribution, while the Generalized Dirichlet  can also allow positive correlation. Another difference is that any permutation of a vector of proportions which follows a Dirichlet distribution is also distributed as a Dirichlet distribution. However, this is not necessarily true for the Generalized Dirichlet distribution.

In this paper we follow the parameterization of the Generalized Dirichlet distribution introduced by Wong (1998). Let $\bs X = (X_1,\ldots,X_k;X_{k+1})$, with $\sum_{j=1}^{k}X_j \leqslant 1$,  $X_j\geqslant 0$ for $j = 1,\ldots,k$ and $X_{k+1} = 1 - X_1 - \ldots-X_k$. Assume that $\alpha_j>0$, $\beta_j >0$ be a set of parameters, $j = 1,\ldots,k$. Then, $\bs X \sim \mathcal GD(\alpha_1,\ldots,\alpha_k;\beta_1,\ldots,\beta_k)$ if the probability density function is written as
\begin{equation}\label{eq:gdpdf}
f_{\bs X}(\bs x) = \begin{cases}
\prod_{j=1}^{k}\frac{x_j^{\alpha_j-1}(1-x_1-\ldots-x_j)^{\gamma_j}}{B(\alpha_j,\beta_j)},& \sum_{j=1}^{k}x_j\leqslant 1, x_j \geqslant 0, j = 1,\ldots,k\\
0,& \mbox{otherwise}
\end{cases}
\end{equation}
where $\gamma_j = \beta_j - \alpha_{j+1}-\beta_{j+1}$ for $j = 1,\ldots,k - 1$, and $\gamma_k = \beta_k - 1$ and $B(\cdot,\cdot)$ denotes the Beta function. Note that when 
\begin{equation}\label{eq:gdequalsd}
\beta_j = \alpha_{j+1} + \beta_{j+1},\quad j = 1,\ldots,k-1,
\end{equation}
a Generalized Distribution reduces to a standard Dirichlet distribution. 

An important property of both Dirichlet and Generalized Dirichlet is that they can be constructed using a stick breaking process. The following result is from Connor and Mossiman (1969): Define $\zeta_1 = X_1$ and $\zeta_j = X_j/V_{j-1}$ for $j = 2,3,\ldots,k$, where $V_j = 1-X_1-\ldots-X_{j-1}$. If $\zeta_j\sim\mbox{Beta}(\alpha_j,\beta_j)$, independent for $j = 1,\ldots,k$. Hence we can construct $X$ as follows:
\begin{eqnarray*}
X_1 &=& \zeta_1\\
X_j &=& \zeta_j(1-X_1-\ldots-X_{j-1}) = \zeta_j\prod_{i=1}^{j-1}(1-\zeta_i), j = 2,3,\ldots,k\\
X_{k+1} & = & 1 -\prod_{i=1}^{k}(1-\zeta_i).
\end{eqnarray*}
In this case: $\bs X = (X_1,\ldots,X_k;X_{k+1})\sim\mathcal{GD}(\alpha_1,\ldots,\alpha_k;\beta_1,\ldots,\beta_k)$ (Connor and Mossiman, 1969). Notice that if $\beta_j = \sum_{k=j+1}^{k+1} \alpha_k$ and also define $\beta_{k+1} = \alpha_{k+1}$ for a given $\alpha_{k+1}>0$, then $\bs X\sim \mathcal{D}(\alpha_1,\ldots,\alpha_k,\alpha_{k+1})$. 

The previously described construction is closely related to the notion of neutrality which was also introduced by Connor and Mossiman (1969): ``a neutral vector of proportions do not influence the proportional division of the remaining interval among the remaining variables''. In particular: a vector of proportions is completely neutral if and only if $\zeta_i$'s are mutually independent (Theorem 2, Connor and Mossiman, 1969). The concept of complete neutrality as well as the representation through the $\zeta$ random variables characterize both the Dirichlet and Generalized Dirichlet distributions and it will be useful for the proof of Theorem 1.

\section{Proof of Theorem 1}

We start with the derivation of the marginal distribution of $\bs\theta$. According to \eqref{eq:theta}, for any given state vector $c$, $\bs\theta$ can be expressed as a suitable permutation of a random variable $\bs u~\sim \mathcal D(\alpha_1,\ldots,\alpha_K)$. Thus, we can write that:
\begin{equation}\label{eq:thetaMarginal}
f(\bs\theta) = \sum_{c\in\mathcal C}P(c)f(\tau^{-1}_c\bs u).
\end{equation}
Now recall that any permutation of $\bs u$ is also distributed according to a Dirichlet distribution and its parameters are just the corresponding permutation of the initial parameters. This means that $\tau^{-1}_c\bs u\sim\mathcal D_{K-1}(\tau_c^{-1}\bs\alpha)$, where $\bs\alpha = (\alpha_1,\ldots,\alpha_K)$. Hence, in the general case where $\bs\alpha$ is an arbitrary vector of strictly positive numbers, \eqref{eq:thetaMarginal} is a mixture of Dirichlet distributions. Now notice that if $\alpha_k = \alpha > 0$, for all $k = 1,\ldots,K$, then $\tau^{-1}_c\bs\alpha = \bs \alpha$ for $c\in\mathcal C$ and \eqref{eq:thetaMarginal} reduces to $\mathcal D_{K-1}(\bs\alpha)$.

The analogous result for $\bs w$ demands a little bit more effort. At first notice that for any given $c$, $\bs w$ can be expressed according to Equation \eqref{eq:w} as a suitable permutation of $$\bs\rho = (u_1,\ldots,u_{k^*},v_1 D_c,\ldots,v_{c_+} D_c),$$ where $D_c=\sum_{k = k^{*}+1}^{K}u_k$. Following the similar argument with $\bs\theta$, it will be sufficient to prove that $\bs\rho$ follows a Dirichlet distribution. From the discussion in Appendix \ref{sec:generalizeddirichlet}, it is equivalent to establish that $\bs\rho$ is completely neutral with $\zeta_k\sim\mbox{Beta}(\delta_j,\sum_{k = j+1}^{K}\delta_k)$ independent for $k = 1,\ldots,K-1$ for some $\delta_k>0$, $k=1,\ldots,K$. 

Let us define now the following variables:
\begin{eqnarray*}
\zeta_1 &=& u_1\\
\zeta_2 &=& \frac{u_2}{1-u_1}\\
&\vdots&\\
\zeta_{k^{*}} &=& \frac{u_{k^*}}{1-u_1-\ldots-u_{k^*-1}}\\
\zeta_{k^{*}+1} &=& v_1\\
\zeta_{k^{*}+2} &=& \frac{v_2}{1 - v1}\\
&\vdots&\\
\zeta_{K-1} &=& \frac{v_{K-1}}{1 - v_1 - \ldots - v_{K-2}}.
\end{eqnarray*}
Since $\bs u$ and $\bs v$ are independent and distributed according to \eqref{eq:uprior} and \eqref{eq:vprior} it follows that $\zeta_k\sim\mbox{Beta}(\alpha_k,\sum_{j = k+1}^{K}\alpha_j)$ for $k = 1,\ldots,k^{*}$ and $\zeta_{k^*+\ell}\sim\mbox{Beta}(\gamma_{\ell},\sum_{j = \ell+1}^{c_+}\gamma_{j})$ for $\ell = 1,\ldots,c_+$. Furthermore, $\zeta_k$ are mutually independent for $k = 1,\ldots,K-1$. 

Now, observe that $\rho_1 = \zeta_1$, $\rho_k = \frac{\zeta_k}{ 1-\rho_1-\ldots-\rho_{k-1}}$, $k = 2,\ldots,K - 1$ and $\rho_K = 1 - \sum_{j=1}^{K-1}\rho_{j}$. But $\zeta$'s are mutually independent and Beta distributed, consequently $\rho$ follows a Generalized Dirichlet distribution:
\begin{equation}\label{eq:rhoDist}
\bs\rho\sim\mathcal{GD}\left(\alpha_1,\ldots,\alpha_{k^{*}},\gamma_1,\ldots,\gamma_{c_+};\sum_{j=2}^{K}\alpha_j,\ldots,\sum_{j=k^{*}+1}^{K}\alpha_j,
\sum_{j=2}^{c_+}\gamma_{j},\ldots,\gamma_{c_+}\right).
\end{equation}
Since $\bs w = \tau^{-1}_c\rho$ for any given $c$, in general, the marginal prior distribution of $\bs w$ is a mixture of permutations of Generalized Dirichlet distributions (as previously discussed, the Generalized Dirichlet distribution is not permutation invariant). In the special case that $\alpha_k = \gamma_k = \alpha >0$ for all $k = 1,\ldots,K$, the property \eqref{eq:gdequalsd} implies that the distribution \eqref{eq:rhoDist} reduces to $\mathcal D(\bs\alpha)$. The result follows using the same argument as the one used for $\bs\theta$.

\section{Proof of Lemma 2}\label{sec:conditionals}

From \eqref{eq:model} we have that: 
\begin{align}\nonumber
\bs u,\bs v|\cdots&\propto&\prod_{i=1}^{r}\bs\theta(\tau,\bs u)_{\xi_i}
\prod_{j=1}^{s}\bs w(\tau,\bs u,\bs v)_{z_j}\prod_{k=1}^{K}u_k^{\alpha_k-1}
\prod_{\ell=1}^{c+}v_\ell^{\gamma_\ell-1}\\ \nonumber
&\propto&\prod_{k=1}^{K}\bs\theta(\tau,\bs u)_k^{s_k(\bs\xi)}
\prod_{k=1}^{K}\bs w(\tau,\bs u,\bs v)_k^{s_k(\bs z)}\prod_{k=1}^{K}u_k^{\alpha_k-1}
\prod_{\ell=1}^{c+}v_\ell^{\gamma_\ell-1}\\ \nonumber
&\propto &\prod_{k=1}^{K}\tau^{-1}\bs u_k^{s_k(\bs\xi)}
\prod_{k\in C_0(c)}\bs w(\tau,\bs u)_k^{s_k(\bs z)}\prod_{k\in C_1(c)}\bs w(\tau,\bs u,\bs v)_k^{s_k(\bs z)}\\ \nonumber
&&\times\prod_{k=1}^{K}u_k^{\alpha_k-1}
\prod_{\ell=1}^{c+}v_\ell^{\gamma_\ell-1}\\ \nonumber
&\propto &\prod_{k=1}^{K}u_k^{s_{\tau_k}(\bs\xi)}
\prod_{k=1}^{k^*}u_k^{s_{\tau_k}(\bs z)}\prod_{k= k^*+1}^{K}\left(v_{k-k^*}\sum_{j=k^*+1}^{K}u_j\right)^{s_{\tau_k}(\bs z)}\\ \nonumber
&&\times\prod_{k=1}^{K}u_k^{\alpha_k-1}\prod_{\ell=1}^{c_+}
v_{\ell}^{\gamma_{\ell}-1}\\ \label{eq:lastline}
&\propto &\prod_{k=1}^{k^*}u_k^{\alpha_k+s_{\tau_k}(\bs\xi)+s_{\tau_k}(\bs z)-1}
\prod_{k=k^*+1}^{K}u_k^{\alpha_k+s_{\tau_k}(\bs \xi)-1}
\left(\sum_{j=k^*+1}^{K}u_j\right)^{\sum\limits_{j=k^*+1}^{K}s_{\tau_j}(\bs z)}\\&&\times
\prod_{\ell= 1}^{c_+}v_{\ell}^{\gamma_{\ell}+s_{\tau_{\ell+k^*}}(\bs z)-1}\nonumber
\end{align}
The last expression yields to conditional independence of $\bs u$ and $\bs v$. Moreover, it is straightforward to see that the full conditional distribution of $\bs v$ is the one defined in expression \eqref{eq:v_gibbs}. The easiest way to see that the conditional distribution of $\bs u$ is the one defined in \eqref{eq:u_gibbs} is to evaluate the density function \eqref{eq:gdpdf} with the parameters given in \eqref{eq:u_gibbs}, make all simplifications and then end up to the first row of last equation. 

Finally, it is important to stress here the convenience of defining $\bs u$ in a way that the set of Equally Expressed transcripts ($C_0$) is followed by the set of Differentially Expressed transcripts ($C_1$), as well as the permutation of the indices as in Definition \ref{def:dead_alive}. Note that the term corresponding to $\sum_{j = k^* + 1}^{K}u_j$ in expression \ref{eq:lastline} refers to the sum of weights of the Differentially Expressed transcripts. If $C_1$ would be a random subset of indices and not the one corresponding to the last $c_+ = K-k^*$ ones, then it would not be possible to directly express the first line of \ref{eq:lastline} as a member of the Generalized Dirichlet family, but rather as a permutation of a Generalized Dirichlet distributed random variable.

\section{Proof of Theorem 2}

Let $\mathcal A_c = \mathcal P_{K-1}\times\mathcal P_{c_+-1}$ and also note that when $c_+ = 0$ then $\mathcal A_c$ reduces to $\mathcal P_{K-1}$. From Equation \eqref{eq:model} and Lemma 2 we have that:
\begin{equation}\label{eq:start}
f(\bs\xi,\bs z|\bs x,\bs y,c) \propto \int\limits_{\mathcal A_c}H(\bs u,\bs v,\bs \xi,\bs z,c)\mathrm{d}\bs u \mathrm{d}\bs v\prod\limits_{i=1}^{r}f_{\xi_i}(x_i)\prod\limits_{j=1}^{s}f_{z_j}(y_j),
\end{equation}
where $H(\bs u,\bs v,\bs \xi,\bs z,c)$ denotes the expression \eqref{eq:lastline}. Now recall that according to Lemma 2, the full conditional distribution of $\bs u,\bs v|\ldots$ becomes a product of independent Generalized Dirichlet and Dirichlet distributions. This means that 
\begin{eqnarray}\nonumber
\int\limits_{\mathcal A_c}H(\bs u,\bs v,\bs \xi,\bs z,c)\mathrm{d}\bs u \mathrm{d}\bs v &=& \prod_{k=1}^{K-1}B(\lambda_k,\beta_k)\frac{\prod\limits_{\ell = 1}^{c_+}\Gamma(\gamma_\ell + s_{\tau_{\ell+k^*}}(\bs z))}{\Gamma(\sum_{\ell = 1}^{c_+}\gamma_\ell+s_{\tau_{\ell + k^{*}}}(\bs z))}\\ \label{eq:lakis}
&=& \prod_{k=1}^{K-1}B(\lambda_k,\beta_k)\frac{\prod\limits_{k\in C_1}\Gamma(\gamma_{\tau^{-1}_k - k^*} + s_{k}(\bs z))}{\Gamma(\sum_{k\in C_1}\gamma_{\tau^{-1}_k - k^*}+s_{k}(\bs z))}.
\end{eqnarray}
Define $\beta_0 = \sum_{j=1}^{K}\alpha_j+r+s$. Observe that $\beta_k+\lambda_k = \beta_{k-1}$ for all $k \neq k^{*} + 1$. Now simplify the product of Beta functions as follows:
\begin{eqnarray*}
\prod_{k=1}^{K-1}B(\lambda_k,\beta_k)&=&\prod_{k=1}^{K-1}\frac{\Gamma(\lambda_k)\Gamma(\beta_k)}{\Gamma(\lambda_k + \beta_k)}\\
&=&\frac{\left(\prod_{k=1}^{k^*}\Gamma(\lambda_k)\right)\Gamma(\beta_{k^*})}{\Gamma(\beta_0)}\frac{\left(\prod_{k=k^*+1}^{K-1}\Gamma(\lambda_k)\right)\Gamma(\beta_{K-1})}{\Gamma(\beta_{k^*+1}+\lambda_{k^*+1})}\\
&=&\frac{\Gamma(\beta_{k^*})\Gamma(\beta_{K-1})}{\Gamma(\beta_0)\Gamma(\beta_{k^*+1}+\lambda_{k^*+1})}\prod_{k=1}^{K-1}\Gamma(\lambda_k).
\end{eqnarray*}
Substituting $\lambda_k$ and $\beta_k$, $k=1,\ldots,K-1$, the last expression yields:
\begin{eqnarray}\nonumber
\prod_{k=1}^{K-1}B(\lambda_k,\beta_k)
&=&\frac{\Gamma\left(\sum\limits_{k\in C_1}\alpha_{\tau^{-1}_k} + s_k(\bs\xi) + s_k(\bs z)\right)}{\Gamma(\beta_0)\Gamma\left(\sum\limits_{k\in C_1}\alpha_{\tau^{-1}_k} + s_k(\bs\xi)\right)}\\
&&\times\prod_{k \in C_0}\Gamma(\alpha_{\tau^{-1}_k} + s_k(\bs\xi) + s_k(\bs z))\prod_{k \in C_1}\Gamma(\alpha_{\tau^{-1}_k} + s_k(\bs \xi)).\label{eq:lakis2}
\end{eqnarray}
Note here that $\Gamma(\beta_0)$ does not depend on $\bs \xi$ or $\bs z$, hence substituting Equations \eqref{eq:lakis}, \eqref{eq:lakis2} into \eqref{eq:start} we obtain \eqref{eq:marginal}, as stated.

Next we proceed to deriving the distributions of $\xi_i|\bs \xi_{[-i]},\bs z,\bs y,c$ and $z_j|\bs z_{[-j]},\bs z,\bs y,c$, for $i = 1,\ldots,r$; $j = 1,\ldots,s$. Let us focus first at the probability a specific read $i = 1,\ldots,r$ of the first condition being assigned to a specific transcript $k = 1,\ldots,K$, given the allocations of all remaining reads ($\bs\xi_{[-i]}, \bs z$) and the state vector ($c$). After discarding all irrelevant terms from Equation \eqref{eq:marginal} we obtain that: 
\begin{eqnarray*}
f(\xi_i|\bs \xi[-i],\bs z,c,\bs x)&\propto&
\frac{\Gamma\left(\sum\limits_{t\in C_1}\alpha_{\tau^{-1}_t} + s_t(\bs\xi) + s_t(\bs z)\right)}{\Gamma\left(\sum\limits_{t\in C_1}\alpha_{\tau^{-1}_t} + s_t(\bs\xi)\right)}\prod_{t \in C_1}\Gamma(\alpha_{\tau^{-1}_t} + s_t(\bs \xi))\\
\nonumber
&\times&\prod_{t \in C_0}\Gamma(\alpha_{\tau^{-1}_t} + s_t(\bs\xi) + s_t(\bs z))f_{\xi_i}(x_i).
\end{eqnarray*}
Now, notice that: $s_t(\bs\xi) = s_t^{(i)}(\bs\xi)$ for $t\neq k$ while $s_k(\bs\xi) = s_k^{(i)}(\bs\xi) + 1$ and recall that $\Gamma(x + 1) = x\Gamma(x)$. Hence, the last equation simplifies to:
\begin{equation*}
P(\xi_i = k|\bs \xi[-i],\bs z,c,\bs x)\propto\
\begin{cases}
(\alpha_{\tau^{-1}_k} + s_k^{(i)}(\bs\xi) + s_k(\bs z))f_{k}(x_i), & k\in C_0\\
\frac{\sum\limits_{t\in C_1}\alpha_{\tau^{-1}_t} + s_{t}^{(i)}(\bs\xi) + s_{t}(\bs z)}{\sum\limits_{t\in C_1}\alpha_{\tau^{-1}_t} + s_{t}^{(i)}(\bs\xi)}(\alpha_{\tau^{-1}_k} + s_{k}^{(i)}(\bs\xi))f_k(x_i),& k\in C_1
\end{cases}
\end{equation*}
which is Equation \eqref{eq:xi_collapsed}, as stated. Equation \eqref{eq:z_collapsed} is derived after following the similar arguments for $z_j|\bs z_{[-j]},\bs\xi,\bs y$.

\section{Update of state vector in the rjMCMC sampler}

In this section we introduce the reversible jump proposal for updating the state vector $c$ and $\bs v$. 

\textbf{Birth move:} Assume that the current state of the chain is $$g:=(c,\tau,\bs u,\bs v,\bs\theta,\bs w,\bs\xi,\bs z).$$ We propose to obtain a new state for the chain
$$g=(c,\tau,\bs u,\bs v,\bs\theta,\bs w,\bs\xi,\bs z)\rightarrow g'=(c',\tau',\bs u',\bs v',\bs\theta',\bs w',\bs\xi',\bs z'),$$
by a birth type move. This will increase the number of differentially expressed transcripts: either by one (if $c_+ \geqslant 2$) or by two (if $c_+ = 0$). At first, we choose a move of this specific type with probability proportional to the number of elements in $C_0(c)$, that is, $K-c_+$. Then, if $c_+\geqslant 2$ we select at random an index $k_0\in C_0(c)$ which we propose to add to $C_1(c)$. If $c_+= 0$ we select at random two indexes $\{k_1,k_2\}\in C_0(c)$ which we propose to move to (the previously empty) $C_1(c)$. The probability of selecting such a move type is, 
\begin{equation}\label{eq:birth}
P_{\mbox{birth}}(c\rightarrow c') = \begin{cases}
\frac{K-c_+}{K}\frac{1}{K-c_+}=\frac{1}{K}, & \mbox{if } 2\leqslant c_+ \leqslant K-1
\\
\frac{K-0}{K}\frac{1}{{K \choose 2}}=\frac{2}{K(K-1)}, &  \mbox{if } c_+ = 0.
\end{cases}
\end{equation}
Moreover, define the corresponding death probability
\begin{equation}\label{eq:death}
P_{\mbox{death}}(c\rightarrow c') = \begin{cases}
\frac{c_+}{K}\frac{1}{c_+}=\frac{1}{K}, & \mbox{if } 3\leqslant c_+ \leqslant K
\\
\frac{2}{K}, &  \mbox{if } c_+ = 2.
\end{cases}
\end{equation}

If $c_+ = 0$, assume without loss of generality that $k_1 <k_2$. Then, $C_1(c')=\{k_1,k_2\}$ and $C_0(c')=\{1,\ldots,K\}-C_1(c')$. Now assume that  $c_+\geqslant 2$. It is obvious that in this case $c'_k =
 c_k$ for all $k\neq k_0$ and $c'_{k_0} = 1$. Moreover, the dead and alive subsets of the new state is obtained by deleting $k_0$ from $C_0(c)$ and adding it to $C_1(c)$. Let $$j:=\sum_{k\in C_1(c)}I(k<k_0) + 1=\sum_{k= 1}^{k_0}c_k + 1.$$
Then, the alive subset of the new state is 
\begin{equation}
C_1(c')=\begin{cases}
\{C_1(c)\}_k, & k<j\\
k_0, & k=j\\
\{C_1(c)\}_{k-1}, & j<k\leqslant c_+ +1
\end{cases}
\end{equation}
and the dead subset will simply be $C_0(c') = \{1,\ldots,K\}-C_1(c')$. Finally, the new permutation is defined as $\tau' = (C_0(c'),C_1(c'))$.

Now, we have to propose the values of $\bs u', \bs v'$. Recall that the dimension of $\bs u$ is always constant, but the dimension of $\bs v'$ will be increased by one. We consider them separately in order to keep it as simple as possible. For $\bs u$ we propose to jump to a new state $\bs u'$ which arises deterministically as the corresponding permutation of its previous values. In order to do this we just have to match the position of each element of $\tau'$ in $\tau$. This means that  
\begin{equation}\label{eq:u_new}
 \bs u' = (\tau^{-1}\tau')\bs u=\tau'[(\tau^{-1}\bs u)].
\end{equation}

In order to construct a valid Metropolis-Hastings acceptance probability for the dimension changing move from $\bs v$ to $\bs v'$, we should take into account the dimension matching assumption of Green (1995). In our set up, this assumption states that the jump from $\bs v\rightarrow \bs v'$ should be done by producing one random variable that will bridge the dimensions, that is: 
\begin{equation*}
\bs v' = h(\bs v,\delta),
\end{equation*}
where $\delta$ denotes a (univariate) random variable and $h(\cdot,\cdot)$ an invertible transformation. We design this transformation following similar ideas from the standard birth and death moves of Richardson and Green (1997), Papastamoulis and Iliopoulos (2009). Let $\delta\sim f_{\mbox{prop}}$, where $f_{\mbox{prop}}$ denotes the density function of a distribution with support $(0,1)$. Then, the new parameter is obtained as
\begin{equation}\label{eq:transformation}
\bs v' = h(\bs v,\delta) :=
\begin{cases}
\left(v_1(1-\delta),\ldots, v_{j-1}(1-\delta),\delta,v_{j+1}(1-\delta),\ldots,v_{c_+}(1-\delta)\right), & c_+\geqslant 2\\
(\delta,1-\delta), & c_+ = 0
\end{cases}
\end{equation}
Finally, we have to compute the absolute value of the Jacobian of the transformation in \eqref{eq:transformation}. Now, recall that $\bs v$ consists of $c_+ -1$ independent elements, so the dimension of the Jacobian is $c_+\times c_+$ (and not $(c_++1)\times (c_++1)$). Then, a routine calculation leads to:
\begin{equation}\label{eq:jacobian}
|J(\delta,c)| =\begin{cases} (1-\delta)^{c_+ -1}, & c_+\geqslant 2\\
1, & c_+ = 0
\end{cases}
\end{equation}

The new values of transcript expression $\bs \theta',\bs w'$ are as follows. By Equations \eqref{eq:theta} and \eqref{eq:u_new}
\begin{equation}\label{eq:theta_new}
\bs\theta'=\tau'^{-1}\bs u'=\tau'^{-1}\tau'[(\tau^{-1}\bs u)]\Rightarrow
\bs\theta'=\bs\theta,
\end{equation}
and applying Equation \eqref{eq:w}:
\begin{equation}\label{eq:w_new}
\bs w' = \tau'^{-1}\left(\{u'_{\tau'^{-1}_{k}}:k\in C_0(c')\},\bs v'\sum_{k\in C_1(c')} u'_{\tau'^{-1}_{k}}\right).
\end{equation}

Finally, we propose to reallocate all observations according to the new values $\bs\theta'$ and $\bs w'$. This is simply done by using the full conditional distributions of $\bs \xi',\bs z'$. Let $P(\bs \xi',\bs z'|\bs \theta',\bs w')$ denote the probability of the allocations, according to the general form given in Equations \eqref{eq:xi_cond} and \eqref{eq:z_cond}. Note here that such a reallocation it is not necessary, however it is suggested because improves the acceptance rate of the proposed move.

\begin{slemma}
The acceptance probability of the birth move is $\min\{1,A(g,\delta,g')\}$, where
\begin{align}\label{eq:accept}
A(g,\delta,g') &=& \frac{f(\bs x,\bs y,\bs z',\bs \xi'|\bs \theta',\bs w')P(\bs \xi,\bs z|\bs \theta,\bs w)}{f(\bs x,\bs y,\bs z,\bs \xi|\bs \theta,\bs w)P(\bs \xi',\bs z'|\bs \theta',\bs w')}\\
&\times&\frac{P_{death}(c'\rightarrow c)P(c')f(\bs u',\bs v'|\bs\alpha,\bs\gamma)|J(\delta,c)|}{P_{birth}(c\rightarrow c')f_{\mbox{prop}}(\delta)P(c)f(\bs u,\bs v|\bs\alpha,\bs\gamma)}.\nonumber
\end{align}
\end{slemma}
\begin{proof}
See the acceptance probability in Green (1995).
\end{proof}
Note that for the Jeffeys prior \eqref{eq:beta_prior}, \eqref{eq:jeffreys_prior} it holds that:
\begin{equation*}
\frac{P(c')}{P(c)} = \begin{cases}
\frac{\pi}{1-\pi}, & c_+ \geqslant 2\\
\frac{\pi^{2}}{(1-\pi)^{2}}, & c_+ = 0.
\end{cases}
\end{equation*}

\textbf{Death move:} A death proposal is the reverse move of a birth. Suppose that we propose a transition 
$$g=(c,\tau,\bs u,\bs v,\bs\theta,\bs w,\bs\xi,\bs z)\rightarrow g'=(c',\tau',\bs u',\bs v',\bs\theta',\bs w',\bs\xi',\bs z'),$$
via a death move. At first we choose at random an element of the alive subset and propose to move and paste it to the dead subset (in the case that the alive subset consists of only two transcripts, then we essentially setting the alive subset to the empty set). This reduces the number of differentially expressed transcripts either by one (if $c_+\geqslant 2$) or by two (if $c_+ =  2$).

If $c_+ =  2$ let $C_1(c) = \{k_1,k_2\}$ and $\bs v = (v_1, v_2)$, with $v_2 = 1-v_1$. Then the random variable that we have to produce during the reverse move is deterministically set to $\delta = v_{1}$, and $\bs v' = \emptyset$. In any other case, assume that the chosen alive transcript index is $k_1$, then define $$j:=\sum_{k\in C_0(c)}I(k < k_1) + 1=\sum_{k= 1}^{k_1}c_k.$$ 
Then, the reverse transformation of \eqref{eq:transformation} implies that 
\begin{equation}\label{eq:inv_transformation}
(\bs v',\delta) = h^{-1}(\bs v) :=
\begin{cases}
\left\{\left(\frac{v_1}{1-v_{j}},\ldots, \frac{v_{j-1}}{1-v_{j}},\frac{v_{j+1}}{1-v_{j}},\ldots,\frac{v_{c_+}}{1-v_{j}}\right),v_{j}\right\}, & c_+\geqslant 2\\
v_j, & c_+ = 0
\end{cases}
\end{equation}
Everything works in a reverse way compared to the birth move, so the acceptance probability of a death move is then simply given by $\min \{1,A^{-1}(g',v_j,g)\}$.

\section{Update of state vector in the collapsed sampler}

According to Equation \eqref{eq:model}, the conditional distribution of $c$ is written as:
\begin{eqnarray*}
f(c|\bs\xi,\bs z,\pi,\bs x,\bs y)\propto f(\bs\xi,\bs z|c,\bs x,\bs y)f(c|\pi)h_c,
\end{eqnarray*} 
where $f(c|\pi)$ denotes the prior distribution of $c$ defined in Equation \eqref{eq:jeffreys_prior}, $f(\bs\xi,\bs z|c,\bs x,\bs y)$ is defined in Equation \eqref{eq:marginal} of Theorem 2 and $h_c = \frac{\Gamma(\sum_{\ell=1}^{c_+}\gamma_\ell)}{\prod_{\ell=1}^{c_+}\Gamma(\gamma_\ell)}$ corresponds to the constant term of the prior distribution for $\bs v$. However, in order to fully update the state vector we would have to compute this quantity for all $c\in\mathcal C$, and this would be time consuming. 

An alternative is to update two randomly selected indices, given the configuration of remaining ones. Hence, if $j_1$ and $j_2$ denote two distinct transcript indices, then we perform a Gibbs update to $c_{j_1,j_2}|c_{-[j_1,j_2]},\bs\xi,\bs z,\pi,\bs x,\bs y$. Let $d = \sum_{k\neq j_1,j_2}c_k$. Since $c_+=\sum_{k}c_k\neq 1$, we have to differentiate the subsequent procedure between the following cases: $d = 0$, $d = 1$ and $d > 1$. If $d = 0$ then $c_{j_1,j_2}\in\{(1,1),(0,0)\}$. In case that $d=1$ then $c_{j_1,j_2}\in\{(1,1),(1,0),(0,1)\}$. Finally, if $d>1$ then $c_{j_1,j_2}\in\{(1,1),(0,0),(1,0),(0,1)\}$. Hence, the following full conditional distribution is derived:
\begin{align}\label{eq:rs1}
P(c_{j_1}=1,c_{j_2}=1|c_{-[j_1,j_2]},\bs\xi,\bs z,\pi,\bs x,\bs y)&\propto& f(\bs\xi,\bs z|c,\bs x,\bs y)\pi^2h_c \\
\label{eq:rs2}
P(c_{j_1}=0,c_{j_2}=0|c_{-[j_1,j_2]},\bs\xi,\bs z,\pi,\bs x,\bs y)&\propto& \begin{cases}
f(\bs\xi,\bs z|c,\bs x,\bs y)(1-\pi)^2 h_c, & d\neq 1\\
0, & d = 1
\end{cases}\\
\label{eq:rs3}
P(c_{j_1}=1,c_{j_2}=0|c_{-[j_1,j_2]},\bs\xi,\bs z,\pi,\bs x,\bs y)&\propto& \begin{cases}
f(\bs\xi,\bs z|c,\bs x,\bs y)\pi(1-\pi)h_c, & d\neq 0\\
0, & d = 0
\end{cases}\\
\label{eq:rs4}
P(c_{j_1}=0,c_{j_2}=1|c_{-[j_1,j_2]},\bs\xi,\bs z,\pi,\bs x,\bs y)&\propto& \begin{cases}
f(\bs\xi,\bs z|c,\bs x,\bs y)(1-\pi)\pi h_c, & d\neq 0\\
0, & d = 0.
\end{cases}
\end{align}

\begin{slemma}\label{lem:stateVectorUpdate}
The update of a randomly selected block of $c$ in the collapsed sampler: 
\begin{enumerate}
\item Select randomly two distinct indices $\{j_1,j_2\}$ from the set $\{1,\ldots,K\}$
\item Update $c_{j_1,j_2}|c_{-[j_1,j_2]},\bs\xi,\bs z,\pi,\bs x,\bs y$ as detailed in Equations \eqref{eq:rs1}--\eqref{eq:rs4}
\end{enumerate}
corresponds to a Metropolis-Hastings step in which the proposed state is always accepted.
\end{slemma}
\begin{proof}
Assume that the current state of the chain is $g = (c,\bs\xi,\bs z, \pi)$ and we propose to move to state $g' = (c', \bs\xi,\bs z, \pi)$, where $c'_k = c_k$ if $k \neq j_1,j_2$ and $c_{j_1,j_2}$ is drawn from the full conditional distribution. The proposal density in this case can be expressed as
\[
P(g \rightarrow g')=\frac{1}{\binom{K}{2}}f(c'_{j_1,j_2}|c_{-[j_1,j_2]}, \bs\xi, \bs z, \pi,\bs x, \bs y)\propto \frac{1}{\binom{K}{2}}f(c', \bs\xi, \bs z, \pi|\bs x, \bs y) = \frac{1}{\binom{K}{2}}f(g'|\bs x, \bs y).
\]
The probability of proposing the reverse move (from $g'$ to $g$) equals to
\[
P(g' \rightarrow g)=\frac{1}{\binom{K}{2}}f(c_{j_1,j_2}|c_{-[j_1,j_2]}, \bs\xi, \bs z, \pi,\bs x, \bs y)\propto \frac{1}{\binom{K}{2}}f(c, \bs\xi, \bs z, \pi|\bs x, \bs y)= \frac{1}{\binom{K}{2}}f(g|\bs x, \bs y).
\]
Thus, the Metropolis-Hastings ratio for the transition $g\rightarrow g'$ is expressed as 
\[
\frac{f(g'|\bs x,\bs y)P(g' \rightarrow g)}{f(g|\bs x,\bs y)P(g \rightarrow g')} = 
\frac{f(g'|\bs x,\bs y)\frac{1}{\binom{K}{2}}f(g|\bs x, \bs y)}{f(g|\bs x,\bs y)\frac{1}{\binom{K}{2}}f(g'|\bs x, \bs y)}=1.
\]
\end{proof}

\section{Clustering of reads and transcripts}

Let $Q=(q)_{ij}$ be a $K\times K$ symmetric matrix. For $i = 1,2,\ldots,K$ and $j = 1,\ldots,K$ let $N_{ij}$ denotes the number of reads that map to both transcripts $i$ and $j$. Define:
\[
 q_{ij} :=
  \begin{cases}
   1 & \text{if } N_{ij}>0 \\
   0       & \text{otherwise.}
  \end{cases}
\]
Clearly, $Q$ would be a diagonal matrix if all reads were uniquely mapped, but for real datasets $Q$ is a sparse and almost diagonal matrix. A typical graphic representation of $Q$ is illustrated in Figure
 \ref{fig:clusters} using a set of simulated reads from the Drosophila Melanogaster transcriptome. Each pixel corresponds to a pair of transcripts that contain at least one read aligned to both transcripts
 of the pair. If all reads were uniquely aligned, this figure would consist only of the diagonal line and in this case each expressed transcript would form its own cluster. Note that the white gaps on the 
diagonal line indicate non-expressed transcripts. Many reads, however, map to more than one transcript resulting in clusters of transcripts, as the red, cyan, blue and green ones in Figure \ref{fig:cluster
s}. The number of transcripts per cluster can have a wide range of values as displayed in Figure \ref{fig:clustersfreq}, but the majority of clusters consist of a very small number of transcripts compared 
to their total number. Next we formally define the notion of a cluster of transcripts.

\begin{figure}[t]
  \centering
    \includegraphics[width = 0.55\textwidth]{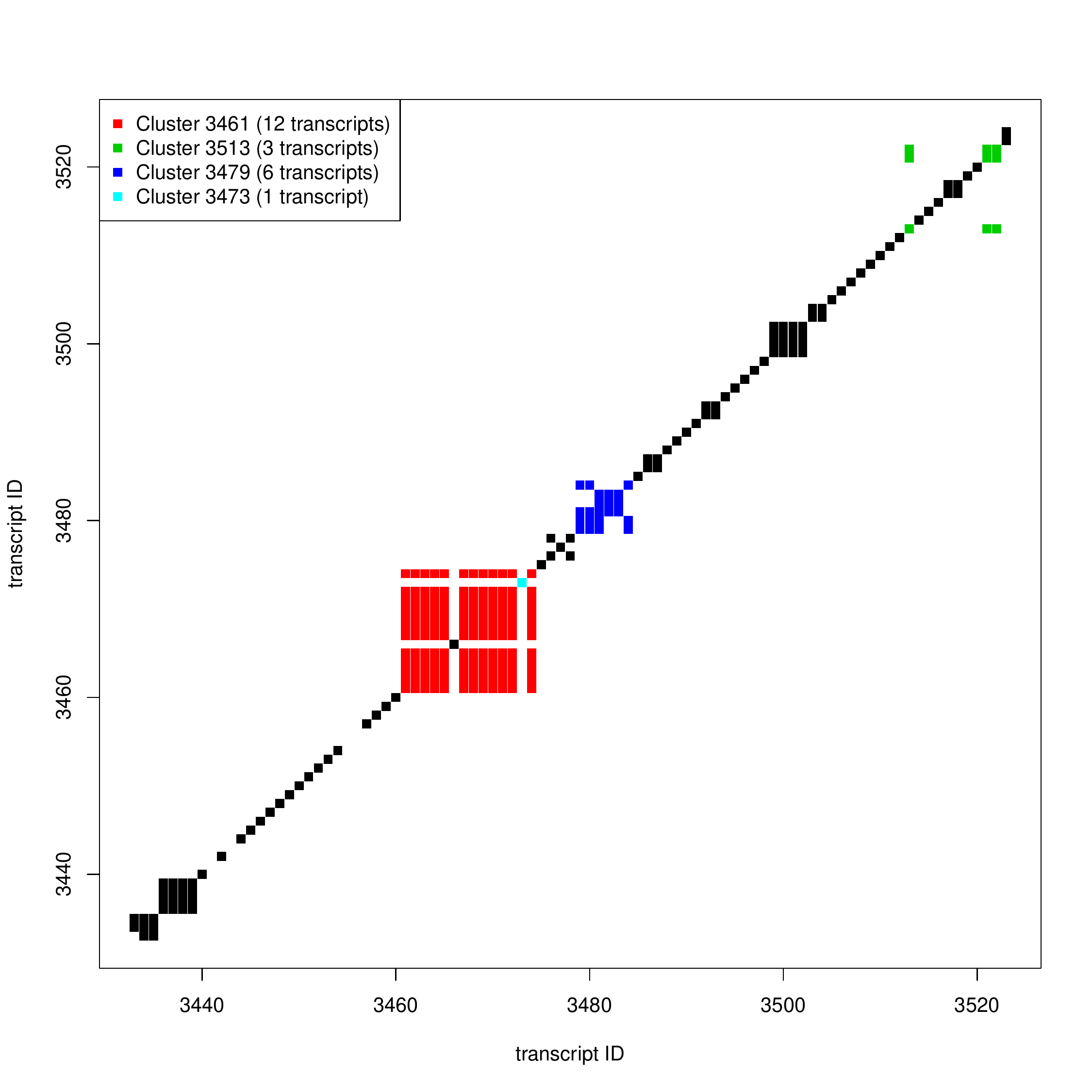}
      \caption{Clusters of transcripts for a simulated set of 75 bp paired  reads from the Drosophila transcriptome, containing $K=28763$ transcripts. For illustration purposes, only a subset consisting of
 281 transcripts is shown and four clusters are emphasized using different colours. White colour corresponds to $N_{ij} = 0$ aligned reads to both transcripts $i,j$.}\label{fig:clusters}
\end{figure}

\begin{mydef}[Associated transcripts]\label{def:adjacent}
Transcript $k_1$ is associated to transcript $k_2$ ($k_1 \leftrightarrow k_2$) if $q_{k_1 k_2} = 1$ or if exists a subset of indices $\{i_1,i_2,\ldots,i_m\}\subseteq\mathcal K:=\{1,\ldots,K\}$, $m \geqslant 1$, such that $q_{k_1i_1} + q_{i_1i_2} +\ldots+ q_{i_{m-1}i_m}+q_{i_mk_2} = m+1$. 
\end{mydef}

\begin{mydef}[Cluster of transcripts]\label{def:cluster}
The set of all associated transcripts of a given transcript $k\in\mathcal K:\sum_{i=1}^{K}N_{ik}>0$: 
$\mathcal C_k:=\{j\in\mathcal K: j\leftrightarrow k\}$,
denotes the cluster of $k$.
\end{mydef}
Note that according to definition \ref{def:cluster}: $k_1\leftrightarrow k_2 \Leftrightarrow \mathcal C_{k_1}= \mathcal C_{k_2}$. We uniquely label each cluster by referring to its minimum index, as follows:
\begin{mydef}[Cluster labels]\label{def:labels}
The label of cluster $\mathcal C_k$ is defined as $\mathcal L_\ell$, with $\ell = \min\{j \in \mathcal C_k\}$. Conventionally, we set $\mathcal L_{0}:=\{k\in\mathcal K:\mathcal C_k=\emptyset\}$.
\end{mydef}

Let $n_c$ be the total number of clusters and assume that $K_j$ is the number of transcripts associated with cluster $\mathcal L_j$. It holds that $\cup_{j=1,\ldots,n_c}\mathcal L_{j}=\mathcal K$ and $\mathcal L_{i}\cap \mathcal L_{j}=\emptyset$ for $i\neq j$, that is, $\{\mathcal L_0,\mathcal L_1,\ldots,\mathcal L_{n_c}\}$ is a partition of $\mathcal K$. Finally, let $r(\mathcal L_k)$ and $s(\mathcal L_k)$ be the number of reads assigned to cluster $\mathcal L_k$ from the first and second condition, respectively. 

\begin{figure}[t]
  \centering
  \includegraphics[width = 0.99\textwidth]{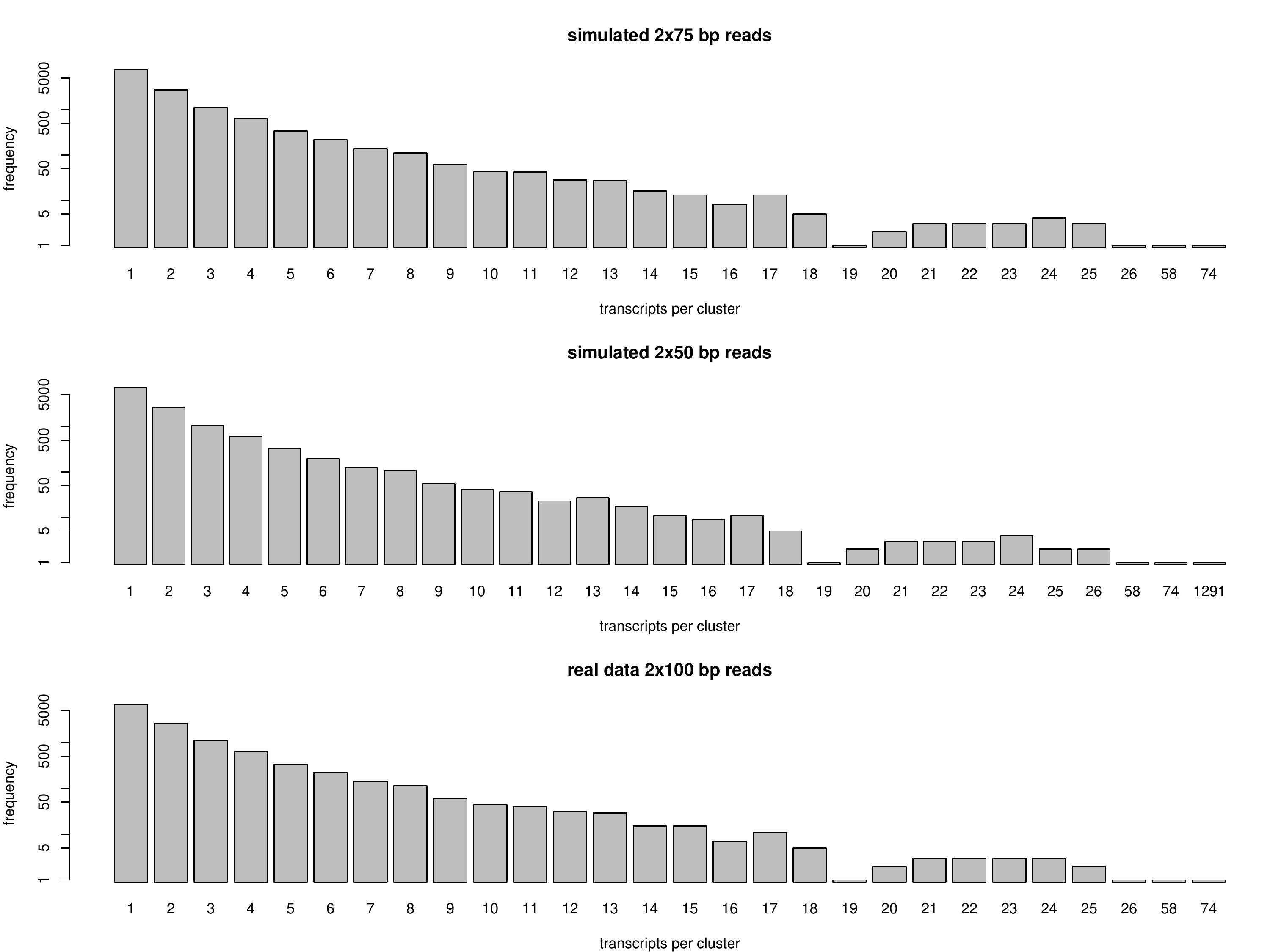}
      \caption{Frequencies (in log-scale) of the number of transcripts per cluster using paired-end reads from the Drosophila Melanogaster transcriptome. Top and middle: ($4489712$) simulated reads, bottom: ($22571142$) reads from real data.}\label{fig:clustersfreq}
\end{figure}

Next, assume that the proposed method is applied separately to each  cluster. This would not lead to the same answer as the one with the full set of reads due to the fact that now each transcript weight corresponds to the relative expression inside each cluster. In order to ensure that the analysis will result to the same answer we should artificially augment each cluster with an extra pseudo-transcript that will contain information of the relative weight of each cluster. There are $r(\mathcal L_j)$ and $s(\mathcal L_j)$ reads from the first and second condition, respectively, exclusively aligned to cluster $\mathcal L_j$, $j =1,\ldots,n_c$. Equivalently, there are $r-r(\mathcal L_j)$ and $s-s(\mathcal L_j)$ reads from the first and second condition, exclusively aligning to the remaining clusters. Assume now that each cluster is augmented with an additional pseudo-transcript containing all remaining reads from both conditions. We conventionally set the label of the pseudo-transcript to $K_j+1$. Given a set of reads from two biological conditions aligned to the reference transcriptome, the pipeline of the algorithm is the following.
\begin{itemize}
\item Partition the reference transcriptome and aligned reads into clusters. 
\item For each cluster $j = 1,\ldots,n_c$, containing $K_j\geqslant 1$ transcripts:
\begin{itemize}
\item augment the cluster by the remaining pseudo-transcript containing $r-r(\mathcal L_j)$  and $s-s(\mathcal L_j)$ reads from the first and second condition, respectively.
\item Run the rjMCMC or the collapsed sampler.
\end{itemize}
\end{itemize}

The following Lemma ensures that it is valid to apply this sampling scheme per cluster in order to estimate the marginal posterior distribution of expression and differential expression for the set of transcripts assigned to each cluster. Apparently, this is not equivalent to simultaneously sampling from the joint posterior distribution of the whole transcriptome, which is computationally prohibitive, however the estimation of the marginal behaviour of each cluster is feasible and computationally efficient due to the dimension reduction.

\begin{slemma}\label{lem:cluster} Let $\tilde{\bs\theta}_j:=(\{\theta_j;j\in\mathcal L_j\},\sum_{k\neq\mathcal L_j}\theta_k)$,  $\tilde{\bs w}_j:=(\{w_j;j\in\mathcal L_j\},\sum_{k\neq\mathcal L_j}w_k)$ denote the augmented transcript expressions for the first and second condition respectively and  $\tilde{c}_j:=(\{c_j;j\in\mathcal L_j\},c_{K_j+1})$, at cluster $j = 1,\ldots,n_c$. A priori assume:
\begin{eqnarray}
\tilde{\bs u}_j&\sim& \mathcal D_{K_j}\left(\{\alpha_j;j\in{\mathcal L_{j}}\},\sum_{k\notin\mathcal L_j}\alpha_k\right)\label{eq:newU}\\
\tilde{\bs v}_j|\tilde{\bs{c}}_j&\sim& \mathcal D_{\tilde{c}_{+}}\left(\gamma_1,\ldots,\gamma_{K_j+1}\right)\label{eq:newV}.
\end{eqnarray}
Then for each cluster $j =1,\ldots,n_c$, the parallel rjMCMC or collapsed algorithm converge to $f(\tilde{\bs\theta}_{j},\tilde{\bs w}_{j},\tilde{\bs c}_{j}|\bs x,\bs y)$ and $f(\tilde{\bs c}_{j}|\bs x,\bs y)$, respectively.
\end{slemma}
\begin{proof}
The distribution \eqref{eq:newU} is derived by \eqref{eq:uprior} by applying the aggregation property of Dirichlet distribution, while distribution \eqref{eq:newV} is the same as \eqref{eq:vprior} given $\bs c=\tilde{\bs c}$. Recall that according to Definition \ref{def:cluster} there are $\sum_{i=1}^{r}I(z_i=K_j+1) = r - r(\mathcal L_j)$ and $\sum_{i=1}^{s}I(\xi_i=K_j+1)=s - s(\mathcal L_j)$ reads allocated to the component labelled as $K_j+1$ for cluster $j$. This means that the update scheme:
\begin{enumerate}
\item Update allocation variables $\tilde{\bs \xi}_j$ and $\tilde{\bs z}_j$ and set $s_{K_j+1}(\tilde{\bs \xi}_j):=r-r(\mathcal L_j)$, $s_{K_j+1}(\tilde{\bs z}_j):=s-s(\mathcal L_j)$.
\item Update free parameters $\tilde{\bs u}_j$ and $\tilde{\bs v}_j$
\item Update expression parameters $\tilde{\bs\theta}_j$ and $\tilde{\bs w}_j$
\item Update state vector $\tilde{\bs c}_j$ 
\end{enumerate}
updates the collapsed parameter vector: 
$$\left(\{\theta_j;j\in\mathcal L_j\},\sum_{k\neq\mathcal L_j}\theta_k\right),\left(\{w_j;j\in\mathcal L_j\},\sum_{k\neq\mathcal L_j}w_k\right),\left(\{c_j;j\in\mathcal L_j\},c_{K_j+1}\right)$$ using the full conditional distributions for steps 1, 2, 3 and the reversible jump acceptance ratio in step 4 (in the case of rjMCMC sampler) or the random scan Gibbs step (in case of collapsed Gibbs). Hence it converges to $f(\tilde{\bs\theta}_{j},\tilde{\bs w}_{j},\tilde{\bs c}_{j}|\bs x,\bs y)$. 
\end{proof}

Note that the previous result assumes a fixed prior probability of DE. In practice, the prior probability of DE is a random variable, following the Jeffreys' prior distribution. Hence, the clustered sampling is equivalent to joint sampling only in case of fixed prior probability of DE. But we have found that this has not any impact in practice since according to our simulations the Jeffreys' prior outperforms the fixed prior probability of DE.

If the reads are sufficiently large, the clusters of transcripts are essentially genes (or groups of genes). It should be clear that the number of clusters as well as the cluster with the largest number of transcripts depends on the read length: if the read length is small, there will be many reads that map to multiple genes and in such a case all of these reads will form a very large cluster, as the one displayed in Figure \ref{fig:clustersfreq} (middle) containing 1291 transcripts. The convergence of the MCMC algorithm for such clusters is questionable. However, even in such cases the majority of transcripts and reads are still forming a large number of small clusters. It is worth mentioning here that the large cluster is created by a very small number of reads: in total there are 417709 reads belonging to this cluster. However, the number of reads that actually map to more than ones genes is equal to 1474. Hence, we could break the bonds of this large number of transcripts by simply discarding or filtering out this small portion of reads. 

\section{Initialization, burn-in and number of MCMC iterations per cluster}
\label{sec:conv}

After partitioning the reads and transcripts into clusters, the rjMCMC or collapsed sampler is applied as previously discussed. For each run (MCMC per cluster), $\mbox{mcmc}_n$ independent chains are obtained using randomly selected initial values for parameters $\bs u$ and $\bs v$, drawn from \eqref{eq:uprior} and \eqref{eq:vprior}. The first half of the chains is initialized from $c_+= 0$ (all transcripts are equally expressed), while the reverse (all transcripts are differentially expressed) holds for the initial state of the second half. The pseudo-transcript of each cluster (i.e.~the mixture component labelled as $K_j+1$) is always initialized as differentially expressed. Given $c,\bs u, \bs v$, the initial relative transcript expressions are computed according to \eqref{eq:dirac} and \eqref{eq:tau}. Each chain runs for a fixed number of $\mbox{mcmc}_m$ iterations, following a pre-specified number $\mbox{mcmc}_b$ of burn-in draws. The posterior means are estimated by averaging the ergodic means across all chains, using a thinning of $\mbox{mcmc}_t$ steps. The proposal distribution in the reversible jump step is an equally weighted finite mixture of Beta distributions: $f_{\mbox{prop}} = \frac{1}{J}\sum_{j = 1}^{J}\mathcal B(1,\beta_j)$. 
All results reported are obtained using: $\mbox{mcmc}_n = 6$, $\mbox{mcmc}_m = 5000$, $\mbox{mcmc}_b = 1000$, $\mbox{mcmc}_t = 5$, $J=5$ and $\{\beta_j;j=1,\ldots,5\}=\{1,10,100,250,500\}$.

\section{Comparison of samplers}\label{sec:comparison}

\begin{figure}[p]
\centering
\begin{tabular}{c}
\includegraphics[width = 0.0825\textwidth,angle =270]{strip.pdf}\\
\includegraphics[width = 0.99\textwidth]{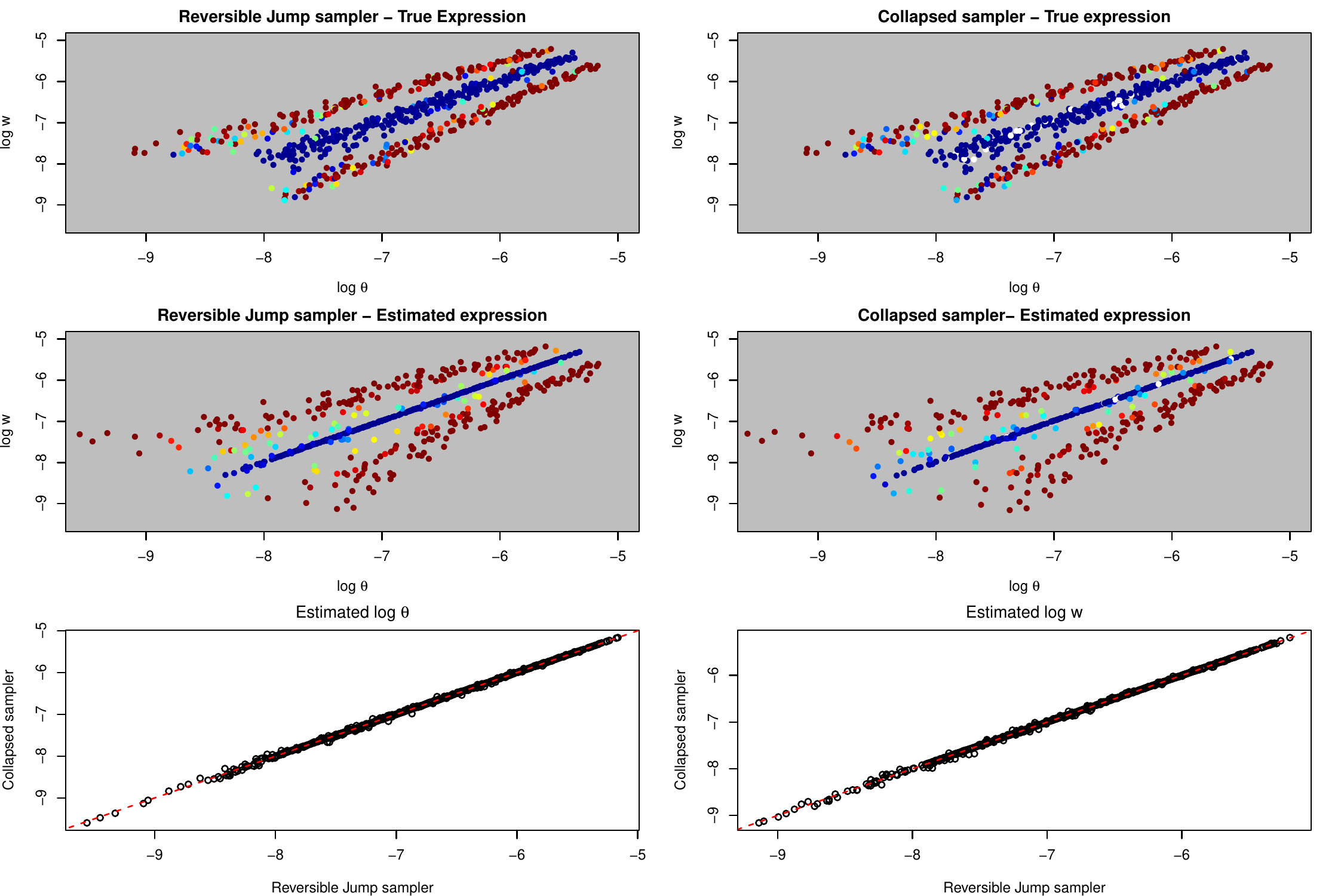}
\end{tabular}
      \caption{True log-relative expression values for the toy example. The colour corresponds to the posterior probability of differential expression according to each sampler under the Jeffreys' prior (blue, green and red colors denote values close to 0, 0.5 and 1 respectively).}\label{fig:toy-true}
\end{figure}

\begin{figure}[p]
\centering
\begin{tabular}{c}
\includegraphics[width = 0.96\textwidth]{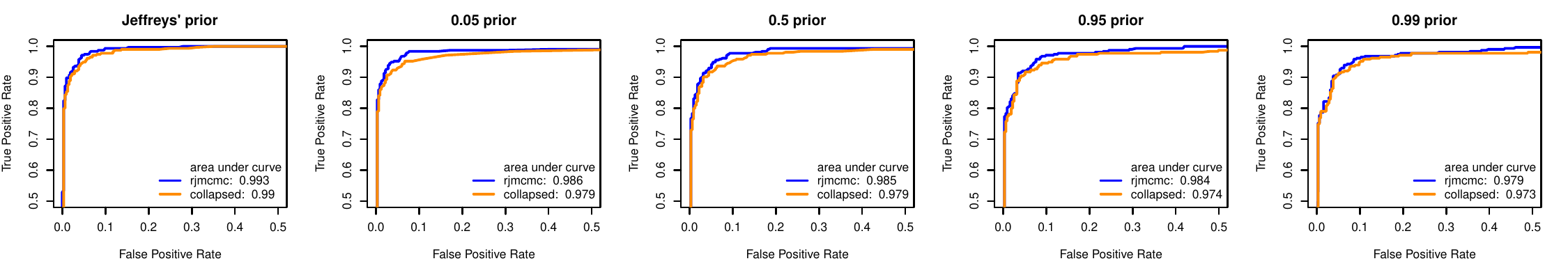}\\
\includegraphics[width = 0.96\textwidth]{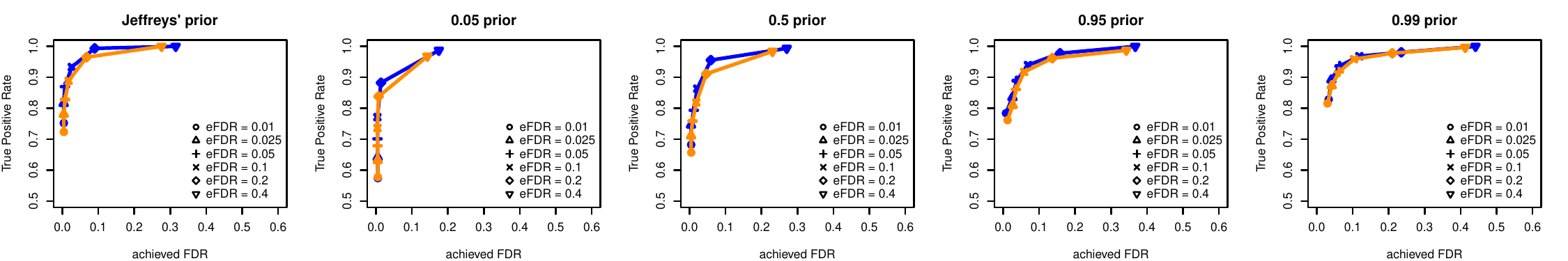}  
\end{tabular}
      \caption{ROC curves (up) and power-to-achieved FDR (down) for the toy example using different prior distribution on the probability of differential expression.}\label{fig:toy-roc}
\end{figure}

\begin{figure}[p]
\centering
\begin{tabular}{c}
\includegraphics[width = 0.99\textwidth]{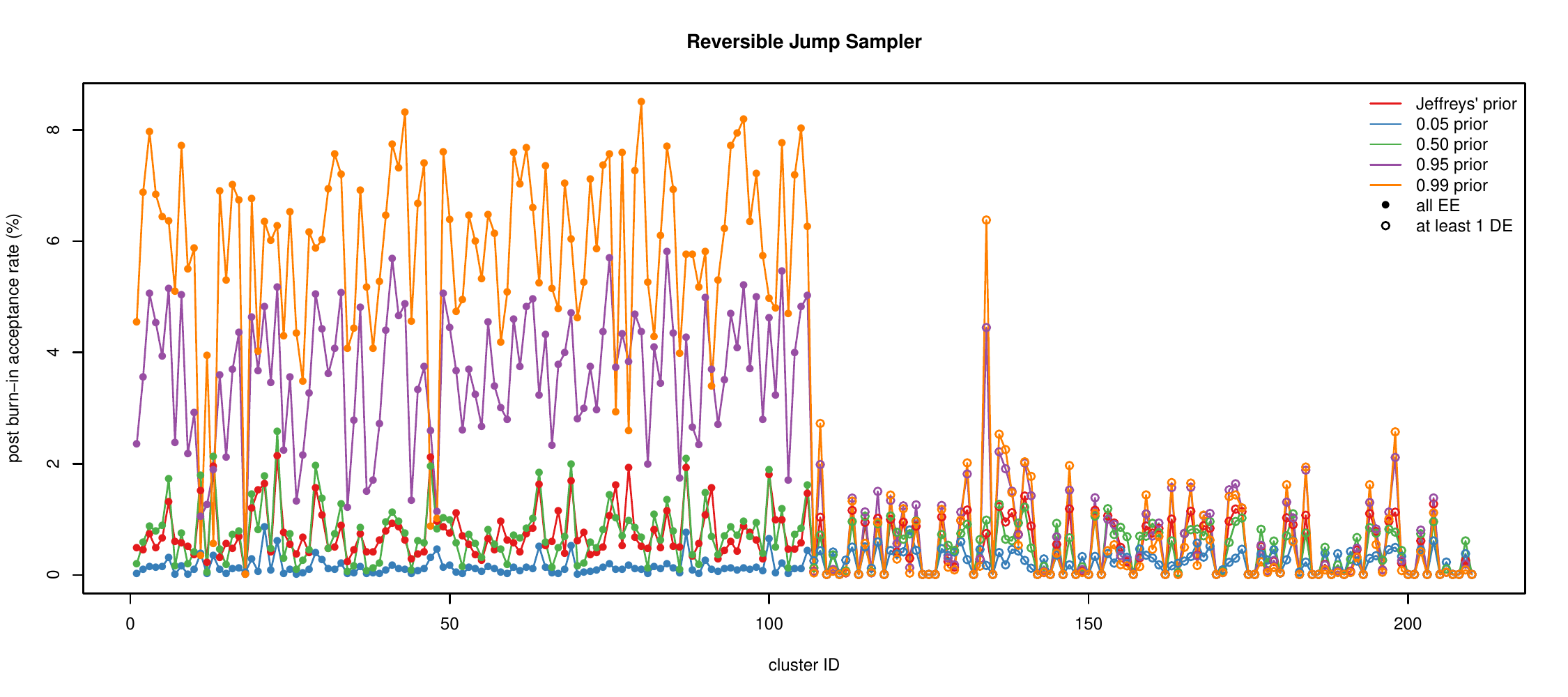}  
\end{tabular}
      \caption{Reversible Jump proposal acceptance rates per cluster (after discarding the MCMC draws which correspond to the burn-in period) for the update of $c,\bs v$ using different prior distributions. Note that the points are reordered so that clusters exclusively consisting of (truely) EE transcripts are shown first (solid points) followed by the clusters which contain at least one (truely) DE transcript (circles).}\label{fig:toy-rate}
\end{figure}

\begin{figure}[p]
\centering
\begin{tabular}{c}
\includegraphics[width = 0.99\textwidth]{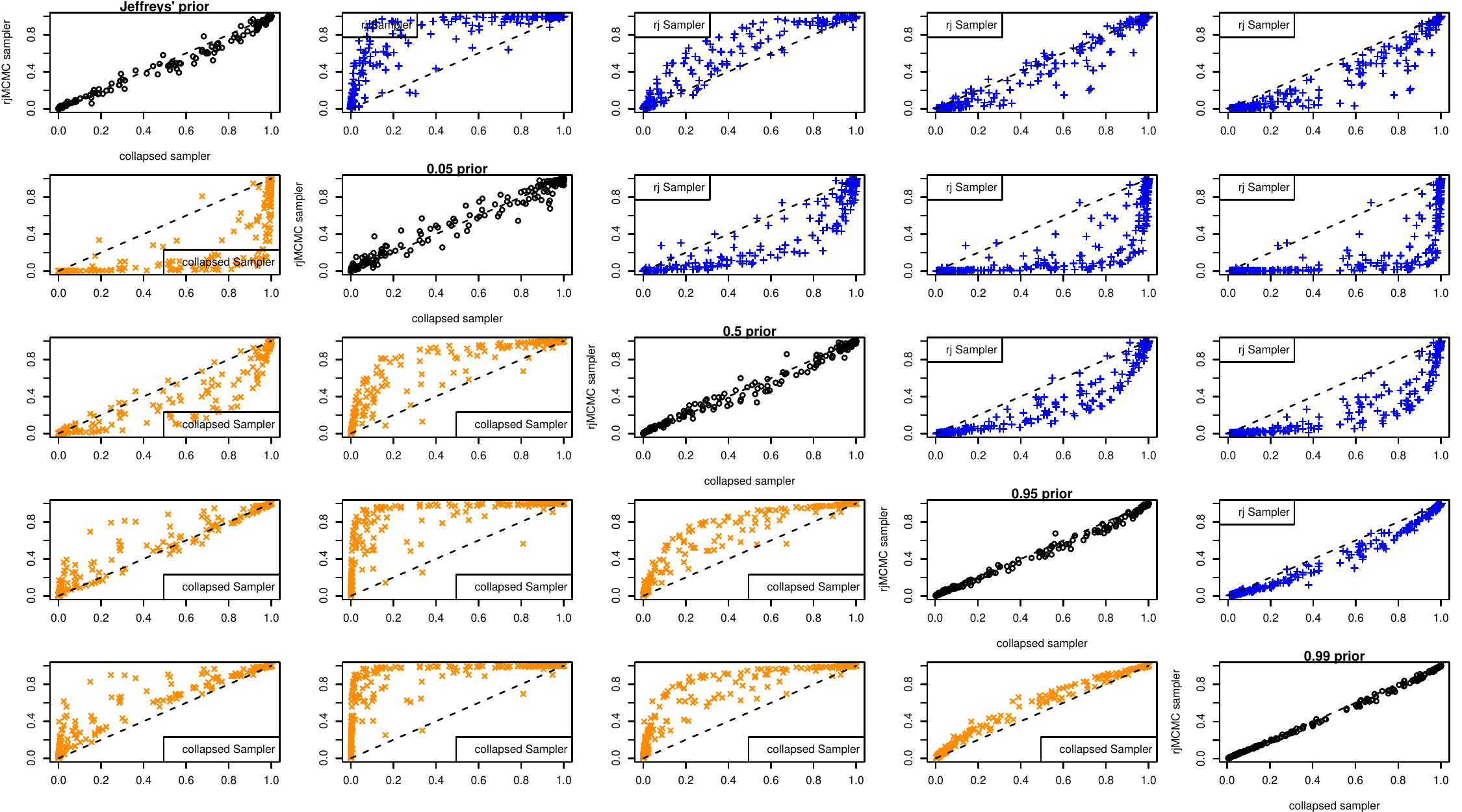}
\end{tabular}
      \caption{Prior sensitivity of the rjMCMC (blue) and collapsed (orange) sampler. The main diagonal contains scatterplots of the posterior probability of DE between the rjMCMC and collapsed samplers for each prior distribution. The scatterplots of the same posterior probabilities for all possible prior combinations per sampler is shown at the upper (rjMCMC) and lower (collapsed) diagonal.}\label{fig:toy-priors}
\end{figure}

In this section we compare the Reversible Jump and the Collapsed version of our method as well we test the sensitivity of these samplers with respect to the prior probability of Differential Expression. In particular, we compare the Jeffreys' prior with a fixed prior probability of DE at $0.05$, $0.50$ and $0.95$. We also examine the acceptance rates of the reversible jump proposal for updating the state vector $c$. Finally, a comparison between the clusterwise and raw sampler is made. For this purpose we used a toy example with relatively small number of reads and transcripts.

We simulated approximately $300000$ reads per sample, which arise from a set of $K = 630$ transcripts. The true values of the mixture weights used for the simulation are shown in Figure \ref{fig:toy-true}. Almost half of transcripts are Differentially Expressed and they correspond to the points that diverge from the identity line. The colour of each point corresponds to the posterior probability of differential expression for each sampler using the Jeffreys' prior distribution. The corresponding ROC curves for each sampler are shown in Figure \ref{fig:toy-roc}, using also different prior distributions on the probability of differential expression. We conclude that the rjMCMC sampler tends to achieve higher true positive rate and a larger area under the curve. The achieved false discovery rates are shown at the second of \ref{fig:toy-true}. Compared to their expected values (eFDR) we conclude that both samplers achieve to control the False Discovery Rate at the desired levels, even when the prior favours DE transcripts (0.95 prior).

\begin{figure}[t]
\centering
\begin{tabular}{c}
\includegraphics[width = 0.99\textwidth]{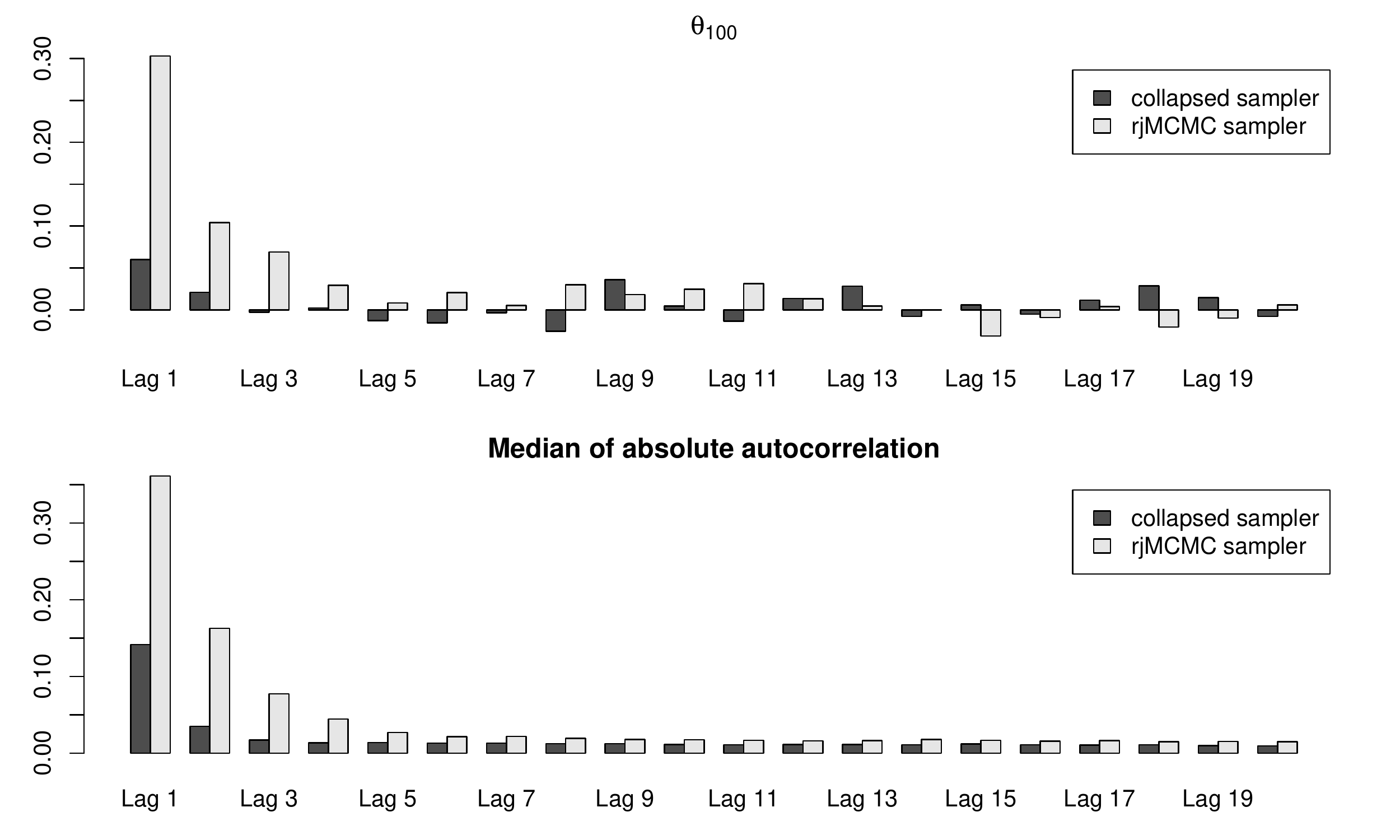}
\end{tabular}
      \caption{Top: estimated autocorrelation function of the collapsed and rjMCMC algorithm for the sampled values of $\log\theta_{100}$. Bottom: Median of absolute autocorrelations for $\log\theta_k$; $k = 1,\ldots,630$.}\label{fig:acf}
\end{figure}

We also examine the acceptance rate of the reversible jump proposal, shown in Figure \ref{fig:toy-rate}. Overall, there is a small acceptance rate of proposed moves and there is a notable increase  when the prior favours DE transcripts (0.95 or 0.99 prior). This mainly affects clusters consisting exclusively of EE transcripts. This conservative behaviour of rjMCMC sampler may indicates that the mixing of the algorithm is poor for the case of EE transcripts. 

Next, we compare the autocorrelation function between the two samplers when using the Jeffrey's prior distribution for the probability of DE. A typical behaviour is shown in Figure \ref{fig:acf} (top), displaying the autocorrelation function of $\log\theta_k$ for a single transcript ($k = 100$). In order to summarize the behaviour of autocorrelations across all $K=630$ transcripts we have computed the median of absolute autocorrelations for all $\theta_k$; $k = 1,\ldots,K$, as shown at the bottom of Figure \ref{fig:acf}. We conclude that the mixing of the collapsed sampler is notably better.

\begin{figure}[p]
\centering
\begin{tabular}{c}
\includegraphics[width = 0.99\textwidth]{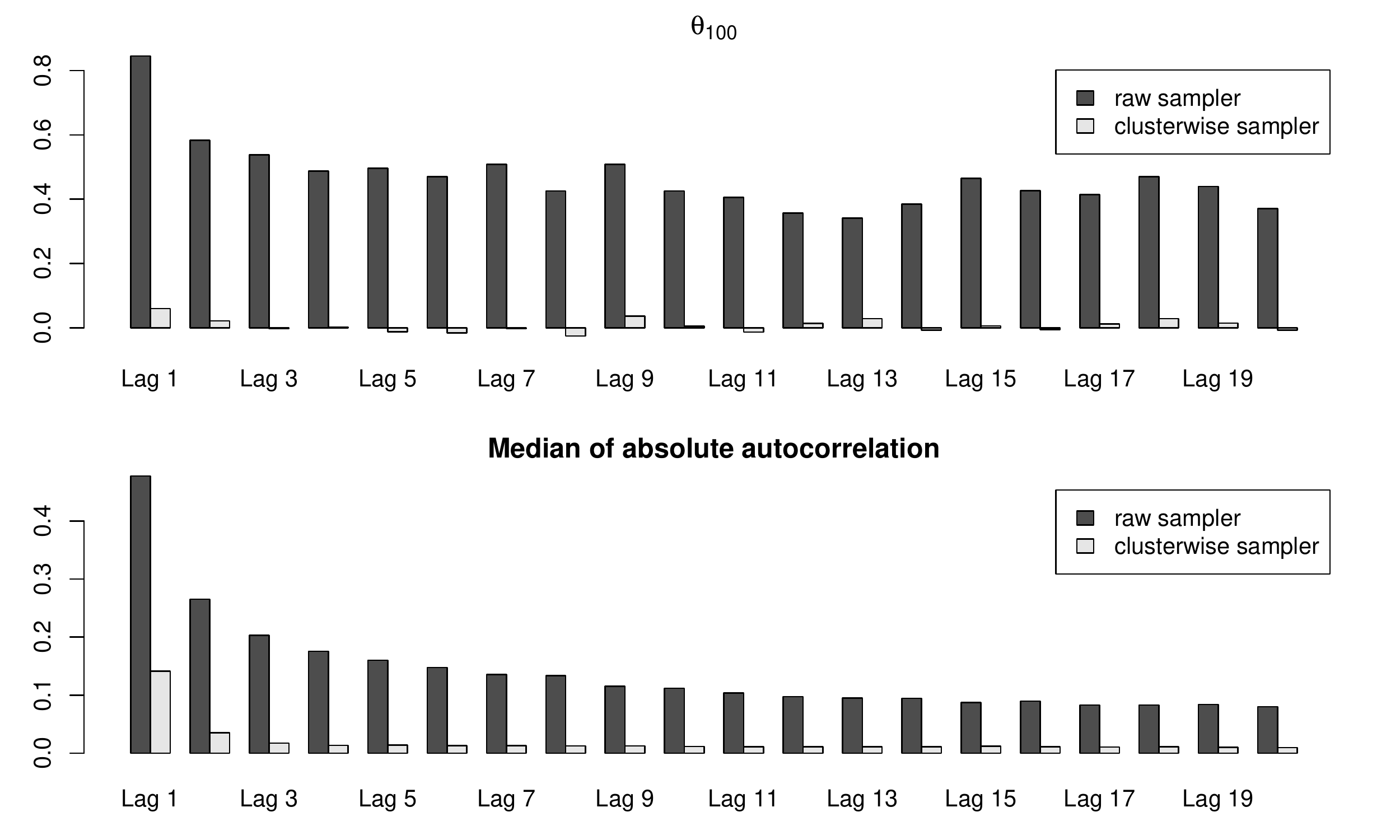}\\
\includegraphics[width = 0.99\textwidth]{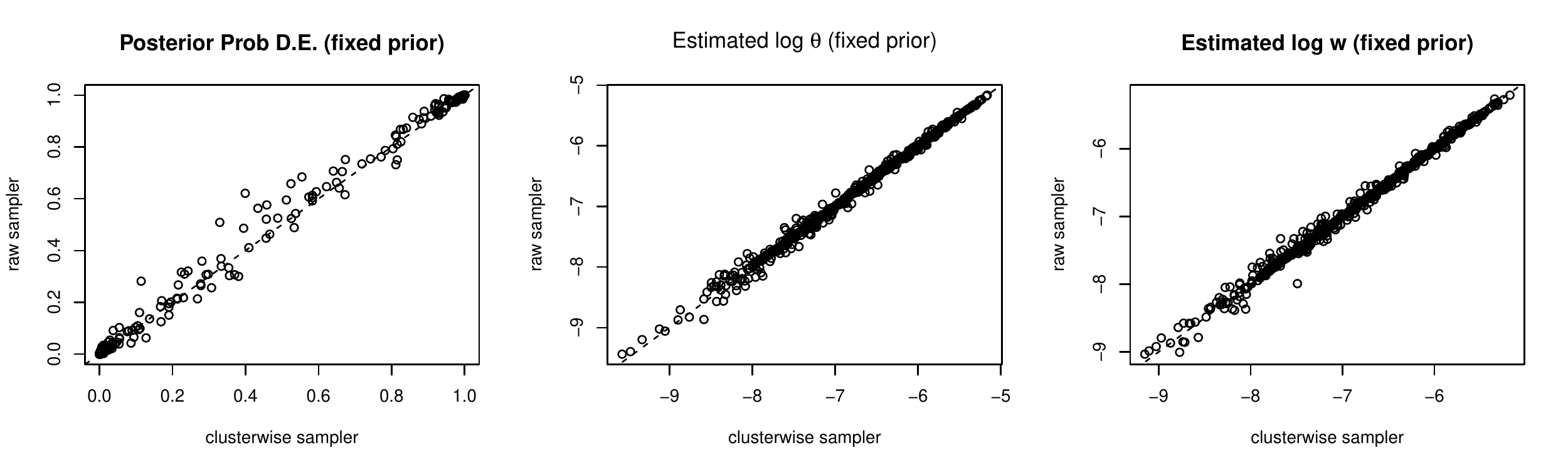}\\
\includegraphics[width = 0.99\textwidth]{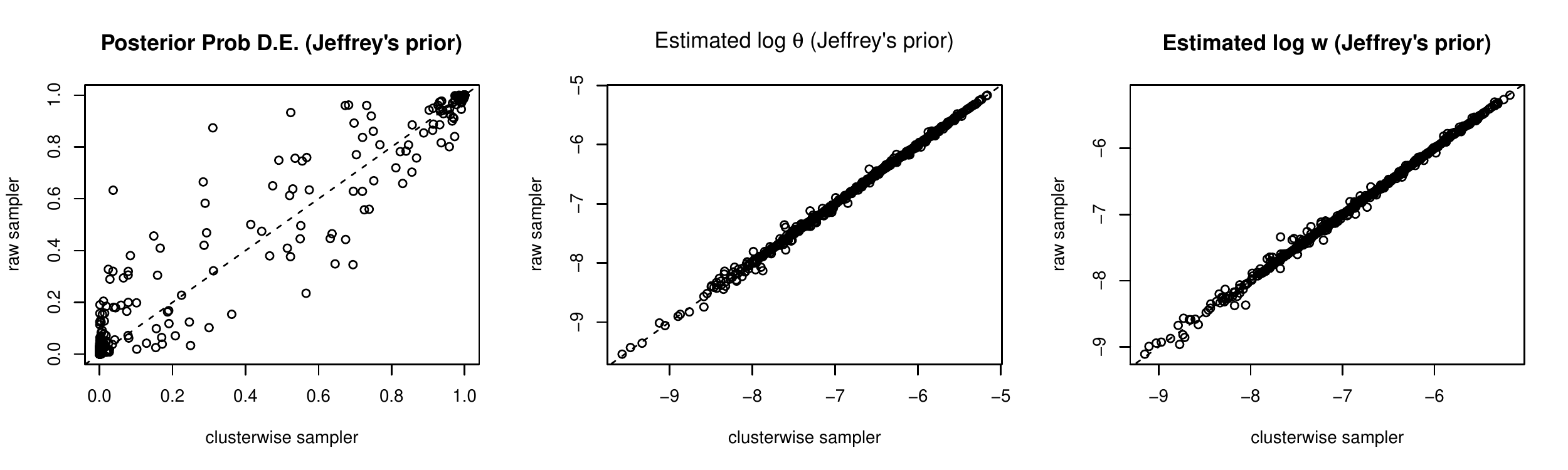}\\
\end{tabular}
      \caption{Comparison of the clusterwise and raw MCMC algorithm. First row: estimated autocorrelation function for the sampled values of $\log\theta_{100}$. Second row: Median of absolute autocorrelations for $\log\theta_k$; $k = 1,\ldots,630$. Third row: Comparison of estimated posterior means.}\label{fig:clusterVSraw}
\end{figure}

Finally we perform a comparison between the raw sampler (that is, taking into account the whole set of reads and transcripts) and the clusterwise one. For this reason we have run the raw collapsed MCMC sampler with a fixed prior of DE (equal to 0.5) as well as the Jeffrey's prior. As shown in Figure \ref{fig:clusterVSraw} (first two rows), the raw MCMC sampler exhibits very large autocorrelations compared to the clusterwise sampler (the autocorrelation function is nearly identical for both prior choices). The resulting estimates of DE and posterior means of transcript expression are shown in the third and fourth row of \ref{fig:clusterVSraw}. Note that the estimates of posterior probability of DE exhibit larger variability under the Jeffrey's prior. In both cases, the transcript expression estimates exhibit strong agreement. The number of iterations of the raw sampler was set to 2000000, following a burn-in period of 200000 iterations. Such a large number of iterations in general will not be sufficient in cases that the number of transcripts grows to typical values of RNA-seq datasets, hence running the raw MCMC sampler becomes prohibitive in general cases.

\section{Simulation study details}

In the sequel, $\mathcal P$ and $\mathcal{NB(\mu,\phi)}$ denote the Poisson and Negative Binomial distributions respectively, where for the latter the parameterization with mean equal to $\mu$ and variance equal to $\mu + \mu^2/\phi$ is used, $\mu\geqslant 0$, $\phi > 0$. Finally, let $\mbox{RPK}^{(A)}_{jk}$ and $\mbox{RPK}^{(B)}_{jk}$ denote the rpk values for transcript $k$ at replicate $j$ of condition A and B, respectively.

\begin{figure}[t]
\centering
\begin{tabular}{c}
\includegraphics[width = 0.99\textwidth]{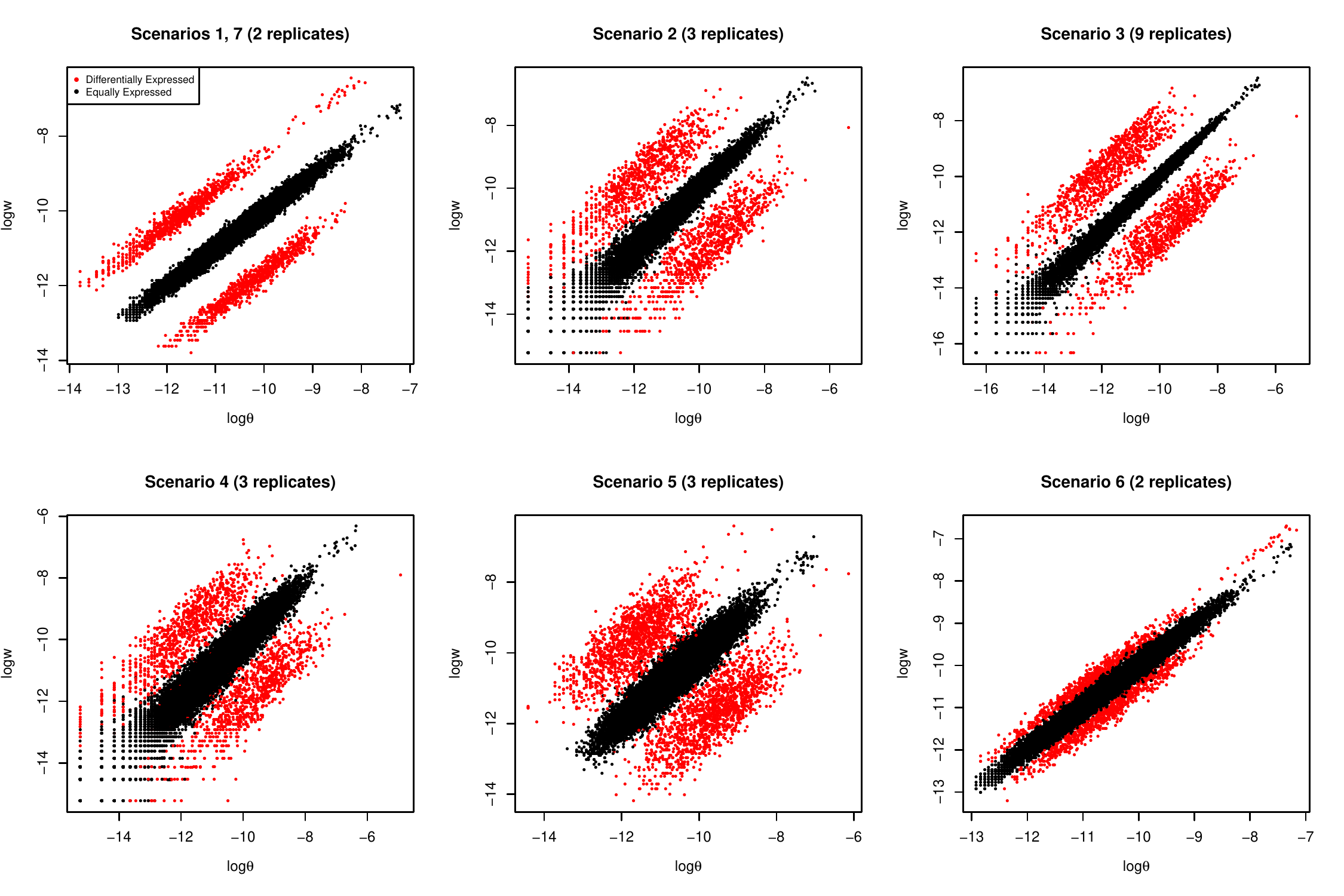}
\end{tabular}
      \caption{Logarithm of true relative expression levels for seven simulation scenarios, averaged across the corresponding number of replicates.}\label{fig:realValsPerScenario}
\end{figure}

\paragraph{Scenario 1 (2 Poisson replicates per condition)} Reads are simulated according to the following generative process. 
\begin{eqnarray*}
\mu_k &=& 65,\quad k=1,\ldots,K, \quad
n_{d} = 2g,\quad g = 1000\\
\{k_1,\ldots,k_{2g}\} &:& \mbox{random sample of indices (without replacement)}\subseteq\{1,\ldots,K\}\\ 
\delta^{(1)}_k & = & 0.65, \quad k=k_1,\ldots,k_g  \\
\delta^{(2)}_k &=& 3.25, \quad k = k_{g + 1},\ldots,k_{2g}\\
(\mu^{(A)}_{k},\mu^{(B)}_{k}) &=& 
\begin{cases}
(1,1)\mu_k,& \quad k\neq k_1,\ldots,k_{n_{d}}\\ \left(
\frac{1}{\delta^{(1)}_k},\frac{1}{\delta^{(2)}_k}\right)\mu_k& \quad k=k_1,\ldots,k_g\\
\left(
\frac{1}{\delta^{(2)}_k},\frac{1}{\delta^{(1)}_k}\right)\mu_k,& \quad k=k_{g+1},\ldots,k_{2g}
\end{cases} \\
\mbox{RPK}^{(A)}_{jk} &\sim& \mathcal P(\mu^{(A)}_{k}),\quad 
\mbox{RPK}^{(B)}_{jk} \sim \mathcal{P}(\mu^{(B)}_{k}),\quad k=1,\ldots,K, j = 1,2.
\end{eqnarray*}
The rpk values determined by this scenario used as input in Spanki and $\approx 2400000$ reads per replicate are simulated ($\approx 9600000$ reads in total). For non-differentially expressed transcripts, rpk values are simulated from a Poisson distribution with mean equal to $65$ for both replicates of each condition. Next, $n_{d} = 2000$ differentially expressed transcripts simulated with mean fold changes equal to $\mu^{(A)}_{k}/\mu^{(B)}_{k} = 1/5$, $k = 1,\ldots,g$ and $\mu^{(A)}_{k}/\mu^{(B)}_{k} = 5$, $k = g+1,\ldots,2g$. More specifically, rpk values generated either from the $\mathcal P(20)$ or  $\mathcal P(100)$ distribution. The averaged relative log-expression based on the true values are shown in Figure \ref{fig:realValsPerScenario} and the points close to the identity line correspond to $26763$ no-DE transcripts. The rest $2000$ points that are far away from the identity line correspond to the DE transcripts. Apparently, this scenario corresponds to a clear cut case of separation between DE and non-DE transcripts at the two conditions.

\paragraph{Scenario 2 (3 Negative Binomial replicates per condition)} Reads are simulated according to the following generative process.
\begin{eqnarray*}
\mu_k &\sim& \mathcal U(0,70),\quad k=1,\ldots,K, \quad n_{d} = 2g,\quad g = 1000\\
\{k_1,\ldots,k_{2g}\} &:& \mbox{random sample of indices (without replacement)}\subseteq\{1,\ldots,K\}\\ 
\delta_k &\sim& \mathcal U(\sqrt{3},\sqrt{5}), \quad k=k_1,\ldots,k_g  \\
(\mu^{(A)}_{k},\mu^{(B)}_{k}) &=& 
\begin{cases}
(1,1)\mu_k,& \quad k\neq k_1,\ldots,k_{n_{d}}\\
(\delta_k,1/\delta_k)\mu_k,& \quad k=k_1,\ldots,k_g\\
(1/\delta_k,\delta_k)\mu_k,& \quad k=k_{g+1},\ldots,k_{2g}
\end{cases}\\
\mbox{RPK}^{(A)}_{jk} &\sim& \mathcal{NB}(\mu^{(A)}_{k},50),\quad 
\mbox{RPK}^{(B)}_{jk}\sim \mathcal{NB}(\mu^{(B)}_{k},50),\quad k=1,\ldots,K, j = 1,2,3.
\end{eqnarray*}
The rpk values determined by this scenario used as input in Spanki and $\approx 1335000$ reads per replicate are simulated ($\approx 8010000$ reads in total). For non-differentially expressed transcripts, rpk values are simulated from the $\mathcal{NB}(65,50)$ for all three replicates of each condition. Next, $n_{d} = 2000$ differentially expressed transcripts simulated with mean fold changes varying in the  $\mu^{(A)}_{k}/\mu^{(B)}_{k} \in(3,5)$, $k = 1,\ldots,g$ and $\mu^{(A)}_{k}/\mu^{(B)}_{k} \in (1/5,1/3)$, $k = g+1,\ldots,2g$. The averaged relative log-expression based on the true values are shown in Figure \ref{fig:realValsPerScenario} and the points close to the identity line correspond to $26763$ no-DE transcripts. The rest $2000$ points correspond to the DE transcripts. Compared to Scenario 1, this case exhibits less separation between DE and non-DE transcripts at the two conditions due to (a) smaller fold changes, (b) increased replicate variability due to the Negative Binomial distribution and (c) larger range of transcript expression values.

\paragraph{Scenario 3 (9 Negative Binomial replicates per condition)} 
The generative process is the same as Scenario 2 but with three times larger number of replicates per condition. In total $\approx 24030000$ reads simulated. The averaged relative log-expression based on the true values are shown in Figure \ref{fig:realValsPerScenario}. Compared to Scenario 2, there should be more signal in the data in order to detect changes in expression due to the increased number of replicates.

\paragraph{Scenario 4 (3 Negative Binomial replicates per condition, enhanced inter-replicate variance)} 
The generative process and the number of simulated reads is the same as Scenario 2 but with larger levels of variability among replicates. In particular we set $\phi = 10$, corresponding to five times larger variability compared to Scenario 2. The averaged relative log-expression based on the true values are shown in Figure \ref{fig:realValsPerScenario}. Compared to Scenario 2, there should be more uncertainty in the data in order to detect changes in expression due to the increased number of replicates. 

\paragraph{Scenario 5 (3 Negative Binomial replicates per condition, enhanced inter-replicate variance, smaller range for the mean)} 
The generative process and the number of simulated reads is the same as Scenario 4 but with more concentrated levels for the mean of true rpk values among replicates. In particular, we set 
\begin{eqnarray*}
\mu_k & = & 60,\quad k=1,\ldots,K, \quad n_{d} = 2g,\quad g = 1000\\
\{k_1,\ldots,k_{2g}\} &:& \mbox{random sample of indices (without replacement)}\subseteq\{1,\ldots,K\}\\ 
\delta_k &\sim& \mathcal U(\sqrt{3},\sqrt{5}), \quad k=k_1,\ldots,k_g  \\
(\mu^{(A)}_{k},\mu^{(B)}_{k}) &=& 
\begin{cases}
(1,1)\mu_k,& \quad k\neq k_1,\ldots,k_{n_{d}}\\
(\delta_k,1/\delta_k)\mu_k,& \quad k=k_1,\ldots,k_g\\
(1/\delta_k,\delta_k)\mu_k,& \quad k=k_{g+1},\ldots,k_{2g}
\end{cases}\\
\mbox{RPK}^{(A)}_{jk} &\sim& \mathcal{NB}(\mu^{(A)}_{k},10),\quad 
\mbox{RPK}^{(B)}_{jk}\sim \mathcal{NB}(\mu^{(B)}_{k},10),\quad k=1,\ldots,K, j = 1,2,3.
\end{eqnarray*}
Note that the difference with Scenario 2 is that now $\mu_k$ has a constant value for all $k = 1,\ldots,K$ and that the selected value of $\phi$ results to five times higher dispersion. The averaged relative log-expression based on the true values are shown in Figure \ref{fig:realValsPerScenario}. It is obvious that now the range of relative expression values is smaller compared to Scenarios 2,3 and 4.

\paragraph{Scenario 6 (2 Poisson replicates per condition, small fold change)} 
This is a revision of Scenario 1 under a smaller fold change between DE and EE transcripts. In this case we set $\delta_k^(1)=65/80$ and $\delta_k^(2)=65/50$, resulting to a fold change of 1.6 for DE transcripts (instead of 5 as used at Scenario 1). As shown in Figure \ref{fig:realValsPerScenario}, the classification of DE and EE transcripts is not obvious.

\paragraph{Scenario 7 (2 Poisson replicates per condition, unequal total number of reads)} 
This is a revised version of Scenario 1 under different sample sizes between the two conditions. Now the first condition contains approximately $46\%$ larger amount of data than the second one. In particular, we simulated $2.81$ and $1.93$ million reads per replicate of the first and second condition, respectively. However, the relative expression levels are the same as in Scenario 1, as shown at the first plot of Figure \ref{fig:realValsPerScenario}.

\begin{figure}[t]
\centering
\begin{tabular}{c}
\includegraphics[width = 0.99\textwidth]{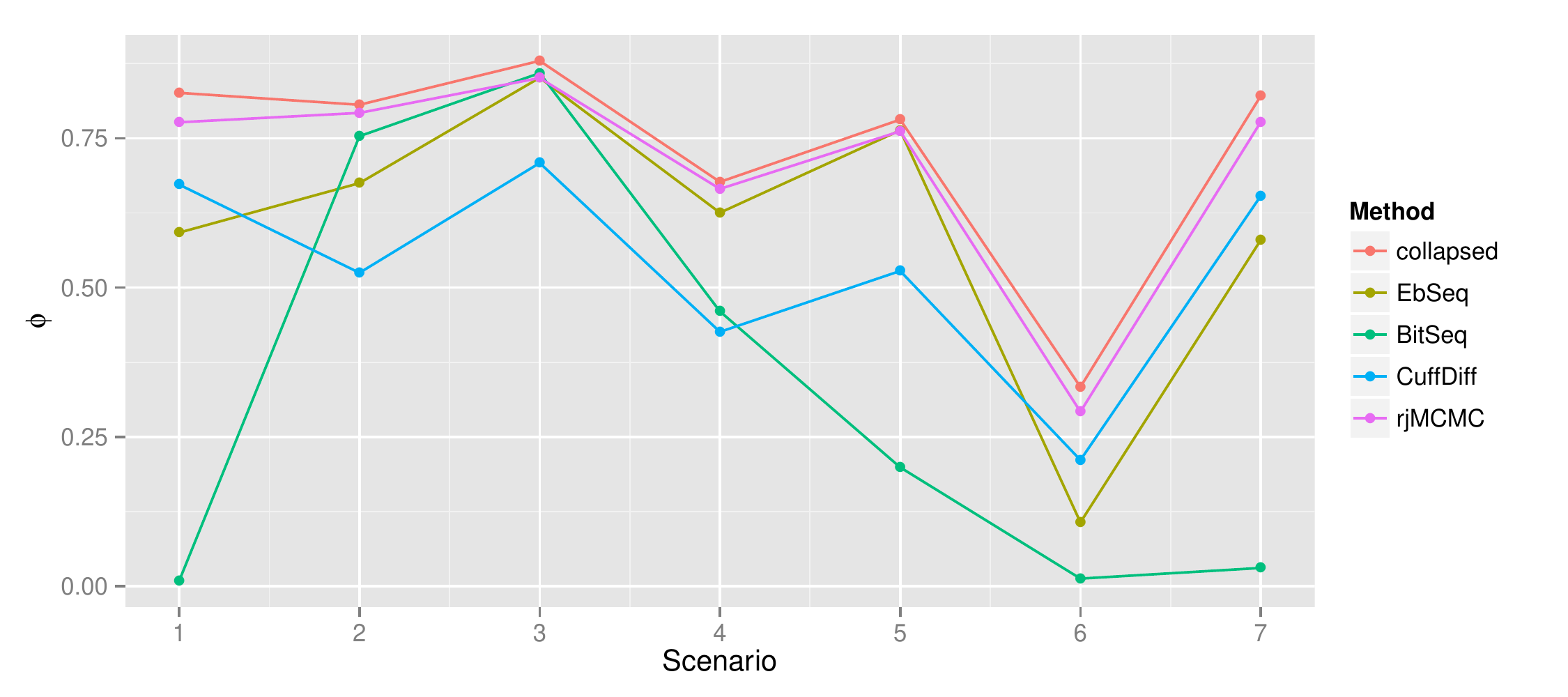}
\end{tabular}
      \caption{$\phi$-coefficient between ground truth of DE and EE transcripts and the inferred classifications per method at the $0.05$ level, for each simulation scenario.}\label{fig:phi}
\end{figure}

Figure \ref{fig:phi} displays the correlation between the true configuration of DE and EE transcripts and the estimated classification per method at the $0.05$ level. Note that our collapsed sampler is ranked as the best method on every scenario. Moreover, our rjMCMC sampler is marginally the second best method. An interesting remark is that methods that control the false discovery rate exhibit a similar pattern across different scenarios, something that it is not the case for the standard BitSeq implementation. However note the improvement of standard BitSeq performance  when the number of replicates is larger than two.

\begin{figure}[p]
\centering
\begin{tabular}{c}
\includegraphics[width = 0.99\textwidth]{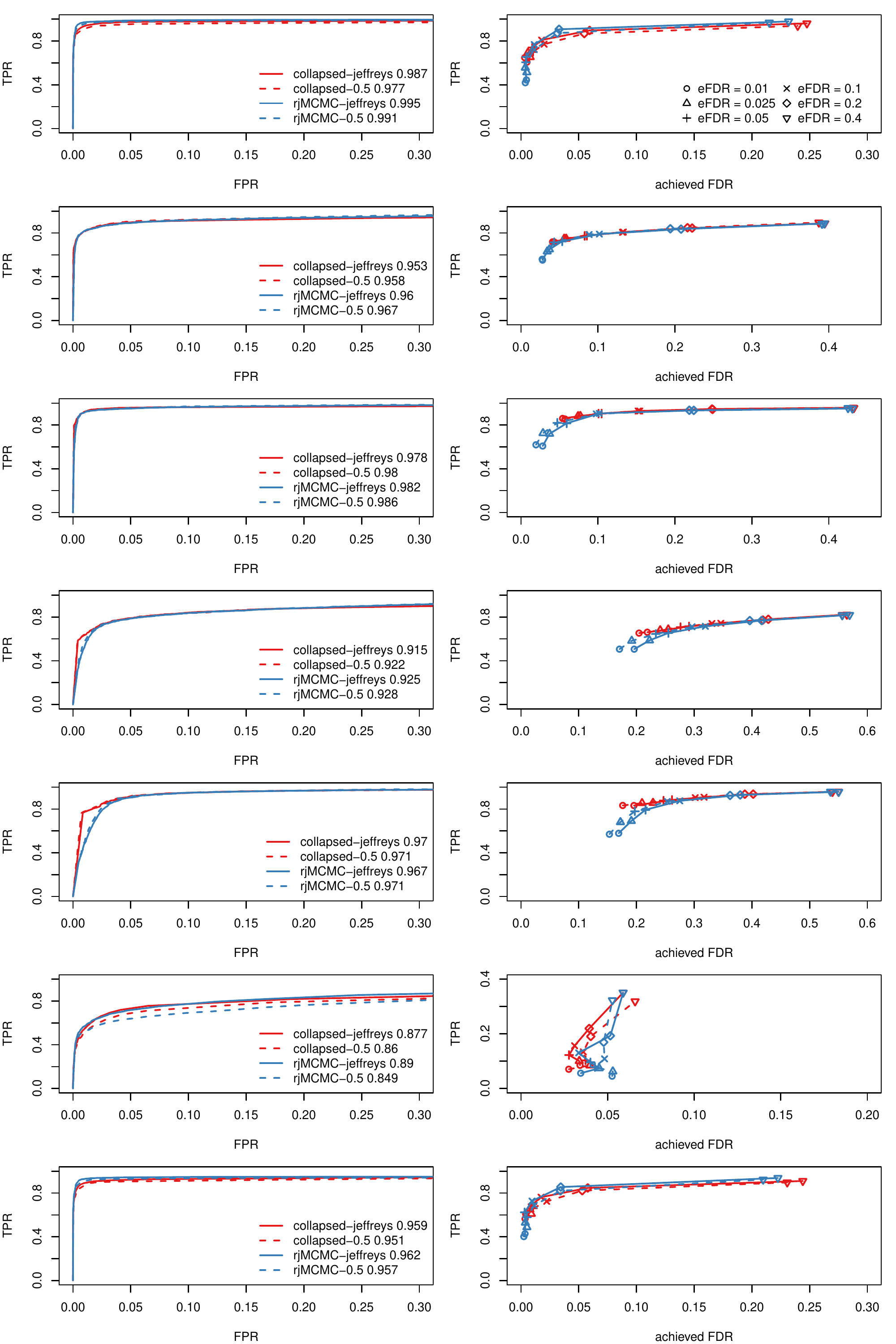}
\end{tabular}
      \caption{ROC curves (left) and power - to achieved plots (right) per simulation scenario, using different prior distribution on the probability of differential expression.}\label{fig:scenariosPrior}
\end{figure}

Figure \ref{fig:scenariosPrior} displays the ROC curves (left) and the true positive rate versus the achieved false discovery rate for the rjMCMC and collapsed samplers. The continuous lines correspond to the Jeffreys prior while the dashed lines correspond to a fixed probability of DE (equal to $0.5$). The results are essentially the same for most scenarios. A notable difference is observed at Scenario 6 where we conclude the superior performance of our method under the Jeffreys prior.

\section{Implementation of the algorithm}

At first, the short reads (.fastq files) for each condition (A and B) are mapped to the reference transcriptome using Bowtie. The alignments (.sam files) are pre-processed using the {\tt parseAlignment} command of BitSeq in order to compute the alignment probabilities for each read (.prob files). These files are used as the input of the proposed algorithm in order to (a) compute the clusters of reads and transcripts and (b) run the MCMC algorithm for each cluster. The output is a file containing the estimates of relative transcript expression for each condition and the posterior probability of differential expression.

\begin{figure}[t]
  \centering
    \includegraphics[width = 0.99\textwidth]{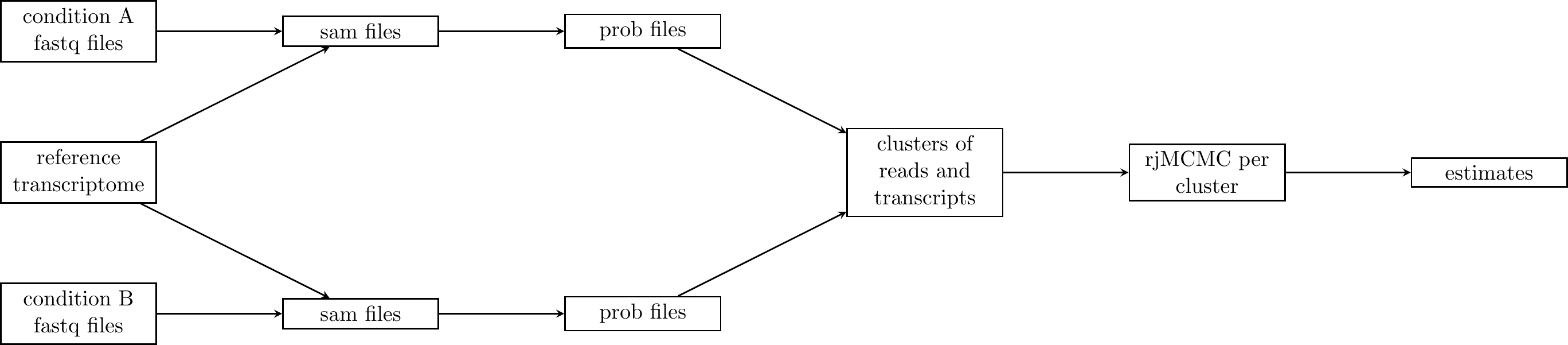}
      \caption{General work-flow of the algorithm.}\label{fig:pipe}
\end{figure}

Assume that there are two replicates per sample consisting of paired-end reads: {\tt A1\_1.fastq}, {\tt A1\_2.fastq}, {\tt A2\_1.fastq} and  {\tt A2\_2.fastq} for sample A and {\tt B1\_1.fastq}, {\tt B1\_2.fastq}, {\tt B2\_1.fastq} and {\tt B2\_2.fastq} for sample B. Denote by {\tt reference.fa} the fasta file with the transcriptome annotation. Let {\tt outputRJ} and {\tt outputCollapsed} denote the output directory of the rjMCMC and collapsed samplers, respectively. The following code describes a typical implementation of the whole pipeline, assuming that all input files are in the working directory (replace by the full paths otherwise).

\begin{verbatim}
# build bowtie2 indices and align reads
bowtie2-build -f reference.fa reference
bowtie2 -q -k 100 --no-mixed --no-discordant -x reference 
                       -1 A1_1.fastq -2 A1_2.fastq -S A1.sam 
bowtie2 -q -k 100 --no-mixed --no-discordant -x reference       
                       -1 A2_1.fastq -2 A2_2.fastq -S A2.sam 
bowtie2 -q -k 100 --no-mixed --no-discordant -x reference
                       -1 B1_1.fastq -2 B1_2.fastq -S B1.sam 
bowtie2 -q -k 100 --no-mixed --no-discordant -x reference
                       -1 B2_1.fastq -2 B2_2.fastq -S B2.sam 

# compute alignment probabilities with BitSeq
parseAlignment A1.sam -o A1.prob --trSeqFile reference.fa
                                                   --uniform
parseAlignment A2.sam -o A2.prob --trSeqFile reference.fa
                                                   --uniform
parseAlignment B1.sam -o B1.prob --trSeqFile reference.fa
                                                   --uniform
parseAlignment B2.sam -o B2.prob --trSeqFile reference.fa
                                                   --uniform
# compute clusters and apply the rjMCMC sampler
rjBitSeq outputRJ A1.prob A2.prob C B1.prob B2.prob
# compute clusters and apply the collapsed sampler
cjBitSeq outputCollapsed A1.prob A2.prob C B1.prob B2.prob
\end{verbatim}
The output of the rjMCMC and collapsed samplers is written to {\tt outputRJ/estimates.txt} and {\tt outputCollapsed/estimates.txt}, respectively. The overall  work-flow is summarized in Figure \ref{fig:pipe}.

\section{Additional tables and figures}

Table 1 illustrates the correlation between the resulting classifications for the two real datasets in Section \ref{sec:real2}. Table \ref{tab:runtime} reports the running time needed for our experiments using 8 threads.  The run-times reported for our method contains both cluster discovery and MCMC sampling. It should be mentioned that a significant portion of the reported run-times is allocated to the clustering part which is not optimized for speed ($20\%-35\%$ and $40\%-45\%$ for the rjMCMC and collapsed samplers, respectively). More details regarding the computing time and memory usage demanded by our method are shown in Figure \ref{fig:comptime}.

\begin{table}
\caption{\label{tab:phi}$\phi$-coefficient between the resulting classifications at the $0.05$ level for  HiSeq (lower diagonal) and MiSeq (upper) data.}
\centering
\fbox{
\begin{tabular}{c|cccc}
\em Method&\em cuffdiff &\em BitSeq &\em EBSeq &\em cjBitSeq\\
\hline
cuffdiff &1 & $0.43$ & $0.32$& $0.32$\\
BitSeq & $0.64$& 1& $0.58$& $0.59$\\
EBSeq & $0.52$& $0.61$&1 & $0.70$ \\
cjBitSeq& $0.56$& $0.63$& $0.75$& 1\\
\end{tabular}}
\end{table}

\begin{table}
\caption{\label{tab:runtime}Approximate total number of reads (in millions) and run-time in hours for each example.}
\centering
\fbox{
\begin{tabular}{c|c|ccccc}
\em dataset &\em reads&\em cufflinks &\em BitSeq &\em rsem/EBSeq &\em rjMCMC &\em collapsed\\
\hline
scenario 1 & $9.4$ & $0.9$ & $4.4$ & $2.2$& $4.8$ & $2.7$  \\
scenario 2 & $8.0$&$0.8$ & $3.3$ & $1.8$& $4.5$ & $3.4$ \\
scenario 3 & $24.0$ &$2.1$ & $9.1$ & $6.2$& $9.8$ & $8.4$\\
scenario 4 & $8.0$ &$0.9$ & $3.4$ & $2.5$ &  $4.4$ & $3.2$\\
scenario 5 & $14.0$&$1.1$ & $9.7$ & $3.7$&  $6.6$ & $5.1$\\
scenario 6 & $9.4$&$0.8$ & $4.5$ & $1.9$&  $4.3$ & $2.6$\\
scenario 7 & $9.5$&$0.7$ & $5.8$ & $1.7$&  $4.1$ & $2.5$\\
MiSeq &   $21.3$  &$1.0$& $4.8$& $2.4$ & $6.8$ & $3.9$\\
HiSeq & $97.0$&$2.4$ & $22.3$ & $11.1$ & $26.1$ & $19.8$\\
\end{tabular}}
\end{table}

\begin{figure}
 \centering
    \includegraphics[width = 0.8\textwidth]{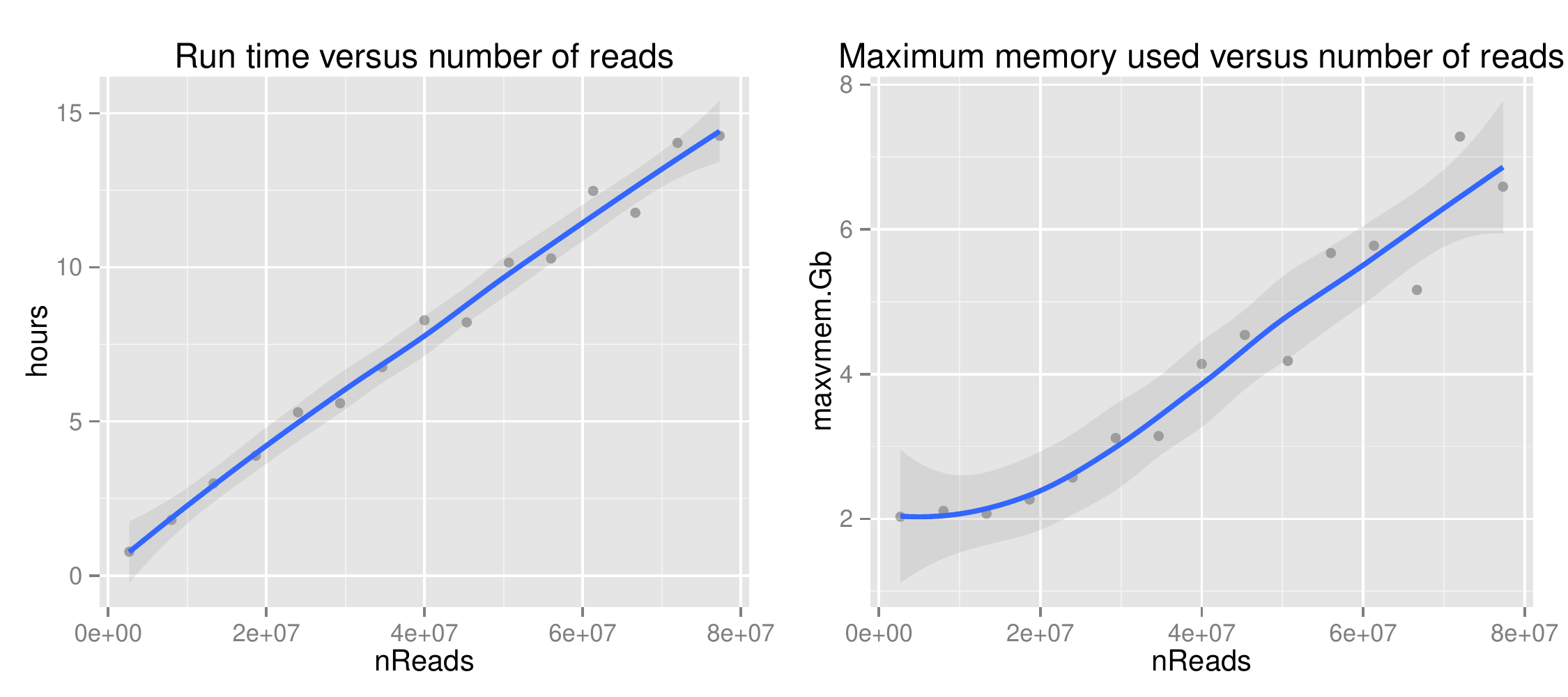}
      \caption{Run time of the algorithm (left) and maximum virtual memory used (right) versus total number of (mapped) reads corresponding to the collapsed algorithm using 12 cores.}\label{fig:comptime}
\end{figure}

\end{document}